\documentclass[a4paper]{amsart}\usepackage[a4paper,margin=2.5cm,includefoot]{geometry}
\usepackage[utf8]{inputenc}
\usepackage[english]{babel}
\usepackage[T1]{fontenc}

\usepackage{amsmath,amssymb,mathtools,bm}
\usepackage{csquotes}
\usepackage{amstext}
\usepackage{amsthm}
\usepackage{thmtools}
\usepackage{bbm}
\usepackage{enumerate}
\usepackage[unicode,psdextra]{hyperref}
\hypersetup{pdfborder={0 0 0.06},citebordercolor=[rgb]{0.8196,0.2275,0.5098},linkbordercolor=[rgb]{0.1765,0.5490,0.8118},urlbordercolor=[rgb]{0.7059,0.5333,0.1137},pdfcreator = {},pdfproducer = {}}
\usepackage[longnamesfirst,numbers,sort&compress]{natbib}
\usepackage[nameinlink,capitalise,sort&compress]{cleveref}
\crefname{section}{\textsection}{\textsection}
\crefname{subsection}{\textsection}{\textsection}
\crefname{subsubsection}{\textsection}{\textsection}
\crefname{paragraph}{\textparagraph}{\textparagraph}
\crefname{thm}{Theorem}{Theorems}

\usepackage{braket}
\usepackage{dsfont}
\usepackage{color}
\usepackage{tikz}
\usepackage{tikz-cd}
\usepackage{pgfplots}
\makeatletter
\def\@seccntformat#1{%
	\protect\textup{\protect\@secnumfont
		\ifnum\pdfstrcmp{subsection}{#1}=0 \bfseries\fi
		\ifnum\pdfstrcmp{subsubsection}{#1}=0 \itshape\bfseries\fi
		\csname the#1\endcsname
		\protect\@secnumpunct
	}%
}
\renewcommand{\@upn}{}
\DeclareRobustCommand{\crefnosort}[1]{%
	\begingroup\@cref@sortfalse\cref{#1}\endgroup
      }      
\makeatother

\numberwithin{equation}{section}
\newtheorem{thm}{Theorem}[section]
\newtheorem{lem}[thm]{Lemma}
\newtheorem{prop}[thm]{Proposition}
\newtheorem{cor}[thm]{Corollary}

\theoremstyle{definition}
\newtheorem{defn}[thm]{Definition}

\renewcommand*{\thehyp}{\Alph{hyp}}

\theoremstyle{remark}
\newtheorem{rem}[thm]{Remark}

\newtheorem*{rem*}{Remark}
\newtheorem*{rems*}{Remarks}
\newtheorem*{ex*}{Example}
\newtheorem*{exs*}{Examples}

\crefname{hyp}{Hypothesis}{Hypotheses}
\Crefname{hyp}{Hypothesis}{Hypotheses}
\crefname{lem}{Lemma}{Lemmas}
\Crefname{lem}{Lemma}{Lemmas}
\crefname{thm}{Theorem}{Theorems}
\Crefname{thm}{Theorem}{Theorems}
\crefname{prop}{Proposition}{Propositions}
\Crefname{prop}{Proposition}{Propositions}
\crefname{enumi}{}{}
\Crefname{enumi}{}{}
\creflabelformat{enumi}{#2(#1)#3}
\crefname{equation}{}{}
\Crefname{equation}{}{}
\crefname{rem}{Remark}{Remarks}
\Crefname{rem}{Remark}{Remarks}

\makeatletter
\renewcommand{\@upn}{} 
\makeatother

\usepackage[inline]{enumitem}
\newlist{enumthm}{enumerate}{1} 
\setlist[enumthm]{label=\upshape(\roman*),ref=\thethm\,(\roman*)}  
\crefalias{enumthmi}{thm} 
\newlist{enumcor}{enumerate}{1}
\setlist[enumcor]{label=\upshape(\roman*),ref=\thecor\,(\roman*)}
\crefalias{enumcori}{cor}
\newlist{enumlem}{enumerate}{1}
\setlist[enumlem]{label=\upshape(\roman*),ref=\thelem\,(\roman*)}
\crefalias{enumlemi}{lem}
\newlist{enumprop}{enumerate}{1}
\setlist[enumprop]{label=\upshape(\roman*),ref=\theprop\,(\roman*)}
\crefalias{enumpropi}{prop}
\newlist{enumhyp}{enumerate}{1}
\setlist[enumhyp]{label=\upshape(\roman*),ref=\thehyp\,(\roman*)}
\crefalias{enumhypi}{hyp}
\newlist{enumproof}{enumerate*}{1}
\setlist[enumproof]{label=\upshape(\roman*)}
\newlist{enumdef}{enumerate}{1}
\setlist[enumdef]{label=\upshape(\roman*),ref=\thedefn\,(\roman*)}
\crefalias{enumdefi}{defn}

\makeatletter
\newcounter{subcreftmpcnt} %
\newcommand\romansubformat[1]{(\roman{#1})} 
\def\subcref{\@ifstar\@@subcref\@subcref}
\newcommand\@subcref[2][\romansubformat]{%
	\ifcsname r@#2@cref\endcsname
	\cref@getcounter {#2}{\mylabel}%
	\setcounter{subcreftmpcnt}{\mylabel}%
	\hyperref[#2]{\romansubformat{subcreftmpcnt}}%
	\else ?? \fi}   
\newcommand\@@subcref[2][\romansubformat]{%
	\ifcsname r@#2@cref\endcsname
	\cref@getcounter {#2}{\mylabel}%
	\setcounter{subcreftmpcnt}{\mylabel}%
	\romansubformat{subcreftmpcnt}%
	\else ?? \fi}   
\makeatother

\makeatletter
\DeclareRobustCommand{\crefnosort}[1]{%
	\begingroup\@cref@sortfalse\cref{#1}\endgroup
}
\makeatother

\def\endstepsymbol{$\lozenge$}
\def\endclaimsymbol{$\lozenge$}
\newcounter{proofstep}
\AtBeginEnvironment{proof}{\setcounter{proofstep}{0}}

\crefname{proofstep}{Step}{Steps}
\Crefname{proofstep}{Step}{Steps}
\newcounter{proofclaim}
\AtBeginEnvironment{proof}{\setcounter{proofclaim}{0}}

\crefname{proofclaim}{Claim}{Claims}
\Crefname{proofclaim}{Claim}{Claims}

\newcommand{\cB}{{\mathcal B}}
\newcommand{\cD}{{\mathcal D}}\newcommand{\cF}{{\mathcal F}}
\newcommand{\cG}{{\mathcal G}}\newcommand{\cH}{{\mathcal H}}
\newcommand{\cL}{{\mathcal L}}
\newcommand{\cM}{{\mathcal M}}\newcommand{\cN}{{\mathcal N}}

\newcommand{\cS}{{\mathcal S}}

\newcommand{\fB}{{\mathfrak B}}

\newcommand{\fa}{{\mathfrak a}}

\newcommand{\fh}{{\mathfrak h}}
\newcommand{\fk}{{\mathfrak k}}

\newcommand{\fq}{{\mathfrak q}}

\newcommand{\BC}{{\mathbb C}}

\newcommand{\BN}{{\mathbb N}}
\newcommand{\BR}{{\mathbb R}}

\usepackage{dsfont} 

\newcommand{\dsone}{{\mathds 1}}

\usepackage{mathrsfs}
\newcommand{\sC}{{\mathscr C}}
\newcommand{\sD}{{\mathscr D}}\newcommand{\sF}{{\mathscr F}}

\newcommand{\sM}{{\mathscr M}}

\newcommand{\sT}{{\mathscr T}}

\newcommand{\sfH}{{\mathsf H}}

\newcommand{\sfS}{{\mathsf S}}\newcommand{\sfT}{{\mathsf T}}
\newcommand{\sfV}{{\mathsf V}}

\newcommand{\sfd}{{\mathsf d}}\newcommand{\sff}{{\mathsf f}}
\newcommand{\sfi}{{\mathsf i}}

\newcommand{\sfn}{{\mathsf n}}

\newcommand{\rmd}{{\mathrm d}}\newcommand{\rme}{{\mathrm e}}

\newcommand{\IN}{\BN}\newcommand{\IR}{\BR}\newcommand{\IC}{\BC}

\newcommand{\hs}{\fh}\newcommand{\HS}{\cH}

\newcommand{\eps}{\varepsilon}\newcommand{\ph}{\varphi}

\newcommand{\e}{\rme}\newcommand{\Id}{\dsone} \renewcommand{\d}{\sfd}

\renewcommand{\Re}{\operatorname{Re}}\renewcommand{\Im}{\operatorname{Im}}

\newcommand{\ran}{\operatorname{ran}}

\DeclareFontFamily{U}{mathx}{\hyphenchar\font45}
\DeclareFontShape{U}{mathx}{m}{n}{
	<5> <6> <7> <8> <9> <10>
	<10.95> <12> <14.4> <17.28> <20.74> <24.88>
	mathx10
}{}
\DeclareSymbolFont{mathx}{U}{mathx}{m}{n}
\DeclareFontSubstitution{U}{mathx}{m}{n}
\DeclareMathAccent{\widecheck}{0}{mathx}{"71}
\DeclareMathAccent{\wideparen}{0}{mathx}{"75}

\DeclareFontFamily{OMX}{MnSymbolE}{}
\DeclareFontShape{OMX}{MnSymbolE}{m}{n}{
	<-6>  MnSymbolE5
	<6-7>  MnSymbolE6
	<7-8>  MnSymbolE7
	<8-9>  MnSymbolE8
	<9-10> MnSymbolE9
	<10-12> MnSymbolE10
	<12->   MnSymbolE12}{}
\DeclareSymbolFont{mnlargesymbols}{OMX}{MnSymbolE}{m}{n}
\SetSymbolFont{mnlargesymbols}{bold}{OMX}{MnSymbolE}{b}{n}
\DeclareMathDelimiter{\llangle}{\mathopen}{mnlargesymbols}{'164}{mnlargesymbols}{'164}
\DeclareMathDelimiter{\rrangle}{\mathclose}{mnlargesymbols}{'171}{mnlargesymbols}{'171}
\DeclareMathDelimiter{\lsem}{\mathopen}{mnlargesymbols}{'102}{mnlargesymbols}{'102}
\DeclareMathDelimiter{\rsem}{\mathclose}{mnlargesymbols}{'107}{mnlargesymbols}{'107}
\DeclareMathDelimiter{\langlebar}{\mathopen}{mnlargesymbols}{'152}{mnlargesymbols}{'152}
\DeclareMathDelimiter{\ranglebar}{\mathclose}{mnlargesymbols}{'157}{mnlargesymbols}{'157}
\DeclareMathDelimiter{\lWavy}{\mathopen}{mnlargesymbols}{'137}{mnlargesymbols}{'137}
\DeclareMathDelimiter{\rWavy}{\mathopen}{mnlargesymbols}{'137}{mnlargesymbols}{'137}

\newcommand{\chr}{\mathbf 1}
\newcommand{\abs}[1]{\lvert#1\lvert}
\newcommand{\norm}[1]{\lVert#1\lVert}

\newcommand{\FGamma}{\Gamma}
\newcommand{\FS}{\cF}\newcommand{\dG}{\sfd\FGamma}\newcommand{\ad}{a^\dagger}
\newcommand{\FSfin}{\cF_{\sff\sfi\sfn}}

\title[Wave Function Renormalization in the Hamiltonian Formalism]{Wave Function
  Renormalization for Particle-Field Interactions}
\author[M.\ Falconi]{Marco Falconi}
\address{M. Falconi, Politecnico di Milano\\Dipartimento di Matematica\\Campus Leonardo,
	P.zza Leonardo da Vinci 32, 20133 Milano\\Italy}
\email{marco.falconi@polimi.it}

\author[B. Hinrichs]{Benjamin Hinrichs}
\author[J. Valent\'in Mart\'in]{Javier Valent\'in Mart\'in}
\address{B. Hinrichs, J. Valent\'in Mart\'in, Universit\"at Paderborn, Institut f\"ur Mathematik, Institut f\"ur Photonische Quantensysteme, Warburger Str. 100, 33098 Paderborn, Germany}
\email{benjamin.hinrichs@math.upb.de}
\email{javiervm@math.uni-paderborn.de}

\usepackage{xcolor}\usepackage{cancel}

\begin{document}

\begin{abstract}
  \noindent
  In this paper, we develop a wave function renormalization scheme for models of
  non-relativistic quantum particles interacting with a quantized
  relativistic field, in the
  Hamiltonian formalism of quantum field theory. We construct the interacting Hamilton operator, in its ground
  state representation, for a large class of particle-field interactions;
  hereby addressing a number of open problems related to both ultraviolet and infrared singularities in the spin-boson and
  Nelson models, where the infinite wave function renormalization plays
  a fundamental role.
\end{abstract}

\maketitle

\section{Introduction}

Rigorous renormalization of relativistic quantum field theories in $3+1$
dimensions is one of the foremost and longest-standing open problems in
mathematical physics. After the groundbreaking results obtained in the
sixties and seventies of past century \citep[see, \emph{e.g.},][and
references
therein]{segal1956am+tams,wightman1956pr,segal1959mfmdvs,segal1960jmp,segal1961cjm,haag1964jmp,hepp1965cmp,wightman1965afksv,glimm1968pr1,hepp1969,GlimmJaffe.1970,glimm1972mcp,simon1972jfa,osterwalder1973cmp,glimm1974am,simon1974psp,guerra1975am,osterwalder1975cmp,ReedSimon.1975,reed1979III,GlimmJaffe.1985,glimm1985cp2,glimm1987qp,haag1992tmp,streater2000plp},
that go under the collective name of Constructive Quantum Field
  Theory (CQFT), the topic has been revived in the latest years with new inputs
from probability and stochastic analysis \citep[see, \emph{e.g.},][and
references
therein]{hairer2014im,jaffe2014arxiv,hairer2016cdm,hairer2016cmp,bruned2019im,aizenman2021am,gubinelli2021cmp,chandra2024im,chatterjee2024cmp},
and is now (again) a very active field of research in mathematics.

In between and in parallel with the development of CQFT for fully
relativistic theories, crucial advances have been made concerning the
mathematics of partially relativistic models, that describe the interaction
between non-relativistic particles and relativistic quantum fields. These
models have the advantage of being more manageable than fully relativistic
theories -- still retaining many of their most interesting features -- and have
proved to be of great interest for condensed matter physics and quantum
optics. While it would be nigh impossible to cite all the literature on the
various mathematical aspects of particle-field interactions, we invite the
interested reader to refer to
\citep{lieb1958pr,trubatch1968pr,Gross.1973,donsker1983cpam,Spohn.1987,lieb1997cmp,bach1998advm,AraiHirokawaHiroshima.1999,bach1999cmp,bach1999cmp2,Hiroshima.2000,derezinski2001jfa,frohlich2001am,griesemer2001im,hiroshima2001tams,ammari2004jfa,frohlich2004cmp252,frohlich2004cmp251,hirokawa2005am,Arai.2006,bach2006cmp,Moller.2006b,merkli2007prl,hasler2008rmp,HaslerHerbst.2008a,HaslerHerbst.2010,HaslerHerbst.2011b,abdesselam2012cmp,gerard2012jfa,ammari2014jsp,GubinelliHiroshimaLorinczi.2014,Falconi.2015,koenenberg2015cmp,HidakaHiroshima.2016,ammari2017sima,BachBallesterosKoenenbergMenrath.2017,Matte.2017,MatteMoller.2018,lieb2019arxiv,DybalskiSpohn.2020,MukherjeeVaradhan.2020,MukherjeeVaradhan.2020b,FrankSeiringer.2021,hidaka2021jst,LeopoldMitrouskasRademacherSchleinSeiringer.2021,bach2022ahp,brooks2022arxiv,BetzPolzer.2022,hiroshima2022rmp,bazaes2023arxiv,betz2023cmp,brooks2023cmp,correggi2022iqm22,correggi2023apde,correggi2023jems,Brooks.2024,correggi2024quantum,HinrichsMatte.2024,sellke2022duke}
for an overview on the more recent results, and to references therein for a
more complete literature. This is also, as of now, a very active field of
research in mathematics.

This paper is devoted to a systematic treatment of the renormalization of the
wave function in the so-called Hamiltonian approach. In fact, this type of
renormalization allows us to address some long-standing open problems in
non-relativistic QFT, caused both by small and large momentum singularities.

\subsection*{Wave function renormalization and dressing transformations in the
  Hamiltonian approach: a historical overview}
\label{sec:wave-funct-renorm-4}

The wave function renormalization (also called field renormalization) is
introduced in interacting field theories by means of a redundant
(\emph{i.e.},\ not observable) coupling parameter, that in its simplest form
amounts to a rescaling of the field in the action. Despite it not being
observable (contrarily to other renormalized parameters such as, \emph{e.g.},
energy and mass), the flow of such coupling constant has to be taken into
account in renormalization schemes, most prominently in the Renormalization
Group approach \citep[see,
\emph{e.g.},][]{wilson1975rmp,weinberg1996I,weinberg1996II,weinberg2000III,mastropietro2008npr,bauerschmidt2019lnm}.
The most singular case is that of an infinite wave function renormalization
\citep[see][]{benfatto2007cmp}, a signature of the interacting vacuum being
disjoint from the bare (free) vacuum \citep[cf.\ Haag's
theorem,][]{haag1992tmp,streater2000plp}. From a mathematical standpoint, the
most established way of constructing interacting (and fully relativistic)
quantum field theories is by means of Wightman's reconstruction theorem
\citep{wightman1956pr}, starting from the study of (euclidean) correlation
functions. Such correlation functions must be proved to satisfy a suitable
set of axioms
\citep{wightman1965afksv,osterwalder1973cmp,osterwalder1975cmp}: this is
achieved through functional integration, rigorous renormalization group
analysis, or stochastic calculus/quantization \citep[see, \emph{e.g.},][and
references
therein]{glimm1968pr1,GlimmJaffe.1970b,glimm1974am,guerra1975am,glimm1987qp,brydges1983cmp,gawedzki1985cmp,brydges1994arg,benfatto2007cmp,hairer2016cmp,hairer2016cdm,gubinelli2021cmp,chandra2024im}. The
wave function renormalization plays a crucial, albeit not always explicit,
role in all these schemes.

Another, perhaps more direct, approach to CQFT focuses on defining the
self-adjoint generator of time translations, either directly or through a
renormalization procedure. This so-called Hamiltonian approach is
particularly suitable for non-fully-relativistic theories, where other
symmetries, such as spatial translation invariance and Lorentz covariance,
might be broken; furthermore, the spectral properties of the Hamiltonian
operator, as well as its behavior in the InfraRed (IR) and UltraViolet (UV)
sectors are of foundational importance both from a mathematical and physical
perspective. In the Hamiltonian approach, however, renormalization has been
successfully carried out only in a few instances: E.\ Nelson renormalized a
model of Schrödinger particles in interaction with a scalar field
\citep{nelson1964jmp}, J.\ Glimm renormalized Euclidean $\mathsf{Yukawa}_2$
and $(\varphi^4)_3$ \citep{glimm1967cmpyuka1,glimm1967cmpyuka2,glimm1968cmp}, while
J.\ Ginibre and G.\ Velo renormalized $(\varphi^2)_d$ with arbitrary space
dimension \citep{ginibre1970cmpren}. Let us revisit these results in some
detail.

The common framework is that of a self-adjoint Hamilton operator
$H(\sigma_0,\sigma)$ on a Hilbert space $\mathscr{H}_0$ endowed with a representation of the Canonical Commutation Relations (CCR). The operator $H(\sigma_0,\sigma)$ is the
generator of the dynamics for a scalar field with a regularized interaction:
$\sigma_0\geq 0$ is the IR cutoff (if necessary), while $\sigma<+\infty$ is the UV cutoff;
thanks to the regular interaction, $\mathscr{H}_0$ is either the Fock space
representing the free field, or has it as a factor (the Nelson model couples
the field with non-relativistic particles, and the Yukawa model with a
fermionic field).

To remove the cutoffs, a renormalization procedure is needed: $H(0,\infty)$ is in
fact only a formal operator, and at most it gives rise to a densely defined,
but not closable, quadratic form $q_{0,\infty}(\cdot )= \langle\, \cdot \,, H(0,\infty)\,\cdot\,
\rangle_{\mathscr{H}_0}$. The first kind of divergences that shall be handled are ``additive
ones'': there are a finite number of counterterms $C_i(\sigma_0,\sigma)$ that must be
added to the Hamiltonian, \emph{i.e.},\ one defines the renormalized
Hamiltonian as
\begin{equation*}
  H_{\mathsf{ren}}(\sigma_0,\sigma)= H(\sigma_0,\sigma) + \sum_i^{}C_i(\sigma_0,\sigma)\; .
\end{equation*}
The physical intuition behind the counterterms is that they diverge as the
cutoffs are removed, canceling out the divergences of $H(0,\infty)$; typical
examples are the constant counterterm, that takes the name of
\emph{self-energy renormalization}, and the counterterm proportional to the
free field's mass term that takes the name of \emph{mass renormalization}.

In the least singular cases, such as the Nelson model, additive
renormalization is enough to define the Hamiltonian:
$H_{\mathsf{ren}}(\sigma_0,\sigma)$ converges (typically in resolvent sense) to a
self-adjoint operator $H_{\mathsf{ren}}(0,\infty)$ on $\mathscr{H}_0$, as $\sigma_0\to 0$ and $\sigma\to
\infty$. A crucial tool to prove convergence is a unitary map introduced by E.P.\
Gross \citep{gross1962ap} and known as \emph{dressing transformation}. The
dressing transformation $U(\sigma_0,\sigma)$ ``almost diagonalizes'' the Hamiltonian by
canceling the most singular parts of the interaction:
\begin{equation*}
  U(\sigma_0,\sigma)H(\sigma_0,\sigma)U^{*}(\sigma_0,\sigma)= \hat{H}_{\mathsf{ren}}(\sigma_0,\sigma)-\sum_i^{}U(\sigma_0,\sigma)C_i(\sigma_0,\sigma)U^{*}(\sigma_0,\sigma)\;,
\end{equation*}
and thus
\begin{equation*}
  U(\sigma_0,\sigma)H_{\mathsf{ren}}(\sigma_0,\sigma)U^{*}(\sigma_0,\sigma)= \hat{H}_{\mathsf{ren}}(\sigma_0,\sigma)\;.
\end{equation*}
The (self-adjoint) operator $\hat{H}_{\mathsf{ren}}(\sigma_0,\sigma)$ is called the
\emph{renormalized dressed Hamiltonian} (still with cutoffs). The dressed
Hamiltonian is typically easier to study than the undressed one (in the
Nelson model, it is the sum of a positive operator and a small form
perturbation uniformly w.r.t.\ $\sigma_0$ and $\sigma$), and its convergence to a
self-adjoint operator in the limit $\sigma_0\to 0$ and $\sigma\to \infty$ can be proved;
furthermore, the unitary dressing map $U(\sigma_0,\sigma)$ itself converges (in strong
sense) to a unitary map $U(0,\infty)$, thus leading to the convergence of the
undressed Hamiltonian (the limit point is itself defined this way):
\begin{equation*}
  H_{\mathsf{ren}}(\sigma_0,\sigma) \xrightarrow[\sigma\to \infty]{\sigma_0\to 0} H_{\mathsf{ren}}(0,\infty)\coloneqq  U^{*}(0,\infty) \hat{H}_{\mathsf{ren}}(0,\infty) U(0,\infty)\;.
\end{equation*}
If the interaction is more UV singular,\footnote{The models we are reviewing
	here have increasing levels of UV singularity, that ultimately require a
	wave function renormalization. Within our framework, the wave function
	renormalization cures both IR and UV singularities, \emph{at the same time}
	(if both are present).}
 wave function renormalization becomes necessary:
the Hilbert space $\mathscr{H}_0$ itself must be renormalized to a space
$\mathscr{H}_{\mathsf{ren}}$ carrying a new inner product. The new inner product is
defined through a non-unitary dressing map $T(\sigma_0,\sigma)$, such that on a dense
subspace $\mathscr{D}\subset \mathscr{H}_0$ containing the free vacuum $\Omega$ (that does not depend on
$\sigma_0$ and $\sigma$),
\begin{equation*}
  \langle \Psi  , \Phi \rangle_{\mathsf{ren}}\coloneqq  \lim_{\substack{\sigma_0\to 0\\\sigma\to \infty}}\frac{\langle T(\sigma_0,\sigma) \Psi , T(\sigma_0,\sigma)\Phi \rangle_{\mathscr{H}_0}}{\lVert T(\sigma_0,\sigma)\Omega  \rVert_{\mathscr{H}_0}^2}
\end{equation*}
is a well-defined sesquilinear form. Intuitively, the dressing $T(0,\infty)$ is
singular as an operator on $\mathscr{H}_0$, however if wave functions are renormalized
to take account of its divergent action on the vacuum, a regular inner
product emerges. The renormalized space $\mathscr{H}_{\mathsf{ren}}$ is then the
Hilbert space completion of $\mathscr{D}$ with respect to $\langle \cdot , \cdot \rangle_{\mathsf{ren}}$. The
wave function renormalization can be tied to scattering theory
\citep[see][for a possible precise formulation in the easy case of a finite
wave function renormalization]{derezinski2003ahp}, and $T(0,\infty)$ formally
captures the most singular parts of the wave operator $\Omega^+$ intertwining
between the free and interacting Hamiltonians, as it is computed
perturbatively through Feynman diagrams \citep{glimm1968cmp}.

The renormalized quadratic form $q_{\mathsf{ren}}$ associated to
$H_{\mathsf{ren}}(0,\infty)$ on $\mathscr{H}_{\mathsf{ren}}$ is thence obtained as the limit
\begin{equation*}
  q_{\mathsf{ren}}(\Psi)= \lim_{\substack{\sigma_0\to 0\\\sigma\to \infty}}\frac{\langle T(\sigma_0,\sigma) \Psi , H_{\mathsf{ren}}(\sigma_0,\sigma)T(\sigma_0,\sigma)\Phi \rangle_{\mathscr{H}_0}}{\lVert T(\sigma_0,\sigma)\Omega  \rVert_{\mathscr{H}_0}^2}\;.
\end{equation*}
The biggest drawback of these germinal attempts at Hamiltonian wave function
renormalization -- a drawback that we are able to overcome in this paper -- is
that it was not possible to prove that the obtained densely defined quadratic
form $q_{\mathsf{ren}}$ is closed/closable; in \citep{ginibre1970cmpren} the authors do that indirectly, resorting
	to an exact diagonalization by means of Bogoliubov transformations, but that
	is only possible in such precisely quadratic toy model.

\subsection*{A systematic approach to wave function renormalization and
  singular dressing}
\label{sec:wave-funct-renorm-4}
In this paper we revisit the ideas above, and construct a systematic and
mathematically rigorous wave function renormalization scheme in the
Hamiltonian formalism; then, we use it as the key tool to resolve both IR and
UV singularities in some of the most used models of interactions between
particles and fields, addressing some open questions and providing some
important missing links with respect to the existing literature. Our approach
highlights a previously hidden structure, allowing us to develop a robust
mathematical framework, that, we believe, shall be applicable to tackle
further open problems. Let us outline the wave function renormalization scheme here,
even though the adaptation to each single model requires some degree of fine
tuning (for which we defer, together with the complete constructions and
results, to the model-specific sections,
\crefrange{sec:van-hove-miyatake}{sec:wave-funct-renorm-2}).

Let us start with perhaps the most important observation: the renormalized
Hilbert space $\mathscr{H}_{\mathsf{ren}}$ is (on the field sector) a Fock space of
\emph{regular excitations} on top of a \emph{singular vacuum}. More
precisely, we construct a natural unitary map\footnote{Although this map is
  natural from the point of view of wave function renormalization, it is not
  natural when considering the algebra of field observables and their
  representations, see the ensuing discussion in the main body of the text.}
\begin{equation*}
  U_{\mathsf{ren}}: \mathscr{H}_{\mathsf{ren}}\longrightarrow \mathscr{H}_0
\end{equation*}
that maps the renormalized space to the bare space; we call such map a
\emph{singular dressing}, since it plays the role of the unitary dressing
$U(0,\infty)$, whenever the latter stops to be defined as an automorphism on
$\mathscr{H}_0$. If, in defining the renormalized space $\mathscr{H}_{\mathsf{ren}}$, an infinite
wave function renormalization is involved, then we prove that, as expected,
$\mathscr{H}_{\mathsf{ren}}$ and $\mathscr{H}_0$ carry \emph{inequivalent representations} of the
field's CCR, \emph{i.e.},\ the singular dressing $U_{\mathsf{ren}}$ \emph{does
  not map} the interacting representation to the free one, but solely the
renormalized to the bare Hilbert space.

The singular dressing is particularly convenient to study self-adjointness of
the renormalized Hamiltonian, overcoming the issues encountered in the
literature. In order to do that, consider the regularized model $H(\sigma_0,\sigma)$ on
the bare space $\mathscr{H}_0$; also in presence of the cutoffs it is possible to
define a ``renormalized space'' $\mathscr{H}_{\sigma_0,\sigma}$, as the Hilbert completion of $\mathscr{D}$
with respect to the inner product
\begin{equation*}
  \langle \Psi  , \Phi \rangle_{\sigma_0,\sigma}= \frac{\langle T(\sigma_0,\sigma) \Psi , T(\sigma_0,\sigma)\Phi \rangle_{\mathscr{H}_0}}{\lVert T(\sigma_0,\sigma)\Omega  \rVert_{\mathscr{H}_0}^2}\;,
\end{equation*}
and a dressing $U_{\sigma_0,\sigma}: \mathscr{H}_{\sigma_0,\sigma}\to \mathscr{H}_0$ (not to be confused with the
$\mathscr{H}_0$-automorphism $U(\sigma_0,\sigma)$ above). Also, it is possible to rewrite the
Hamiltonian $H_{\mathsf{ren}}(\sigma_0,\sigma)$ (where the additive divergences have
been subtracted) as a self-adjoint operator $\mathcal{H}_{\mathsf{ren}}(\sigma_0,\sigma)$ on the
renormalized space $\mathscr{H}_{\sigma_0,\sigma}$, as follows:
\begin{equation*}
  \langle \Psi  ,\mathcal{H}_{\mathsf{ren}}(\sigma_0,\sigma) \Psi \rangle_{\sigma_0,\sigma}= \frac{\langle T(\sigma_0,\sigma) \Psi , H_{\mathsf{ren}}(\sigma_0,\sigma)T(\sigma_0,\sigma)\Phi \rangle_{\mathscr{H}_0}}{\lVert T(\sigma_0,\sigma)\Omega  \rVert_{\mathscr{H}_0}^2}\;.
\end{equation*}
The operator $\mathcal{H}_{\mathsf{ren}}(\sigma_0,\sigma)$ \emph{converges} to the self-adjoint
renormalized Hamiltonian $\mathcal{H}_{\mathsf{ren}}(0,\infty)$, but in what sense? Thanks
to the dressing maps $U_{\sigma_0,\sigma}$ and $U_{\mathsf{ren}}$, we can map all the
operators to the same \emph{parent space}, the bare space $\mathscr{H}_0$: we define
the dressed Hamiltonians on $\mathscr{H}_0$ as
\begin{gather*}
  \hat{H}_{\mathsf{ren}}(\sigma_0,\sigma)= U_{\sigma_0,\sigma}\mathcal{H}_{\mathsf{ren}}(\sigma_0,\sigma)U_{\sigma_0,\sigma}^{-1}\;,\\
  \hat{H}_{\mathsf{ren}}(0,\infty)= U_{\mathsf{ren}}\mathcal{H}_{\mathsf{ren}}(0,\infty)U_{\mathsf{ren}}^{-1}\;.
\end{gather*}
We prove that $\hat{H}_{\mathsf{ren}}(\sigma_0,\sigma)$ converges, in resolvent sense,
to $\hat{H}_{\mathsf{ren}}(0,\infty)$ as $\sigma_0\to 0$ and $\sigma\to \infty$; such a convergence
yields the convergence of $\mathcal{H}_{\mathsf{ren}}(\sigma_0,\sigma)$ to
$\mathcal{H}_{\mathsf{ren}}(0,\infty)$ in \emph{generalized resolvent sense}, a notion of
convergence precisely developed to study limits of operators acting on
different Hilbert spaces, but with a common parent space, see
\cref{sec:textbfg-norm-resolv} for further details and references.

The use of this topology strengthens the notion of convergence in the existing
literature, where only convergence of the quadratic form on a common domain
can be established. Even more so, as it was the case for the Nelson model,
without cutoffs it is the \emph{dressed Hamiltonian}
$\hat{H}_{\mathsf{ren}}(0,\infty)$ that is naturally defined and amenable to
study: the undressed one, $\mathcal{H}_{\mathsf{ren}}(0,\infty)$, is obtained \emph{a
  posteriori} by unitary conjugation. As already remarked, this is precisely
what allows us to overcome the difficulties of proving self-adjointness
encountered in the previous works. Finally, the dressed Hamiltonian is often
even explicit, so its spectral properties can be (or have already been)
studied -- most prominently the existence of a ground state -- thus allowing to
pass those properties to the implicit $\mathcal{H}_{\mathsf{ren}}(0,\infty)$.

In the rest of the paper we develop and apply this renormalization scheme to
three models, in increasing order of difficulty: the van Hove-Miyatake model
(\cref{sec:van-hove-miyatake}), the spin-boson
(\cref{sec:wave-funct-renorm-1}), and the Nelson model
(\cref{sec:wave-funct-renorm-2}). A detailed description of the models themselves, their appearance in the literature as well as a discussion of the respective results, their novelty and relevance, is given in each section in full detail; let us nevertheless
summarize them here for quick reference:

\vspace{2mm}

\begin{itemize}
  \setlength{\itemsep}{2mm}
\item In \cref{sec:van-hove-miyatake} we study the van Hove-Miyatake model:
  we can renormalize its Hamiltonian for any distributional source of
  interaction, \emph{i.e.},\ for (essentially) all kinds of IR and UV
  singularities. The renormalized Hamiltonian always has a non-degenerate
  ground state, and the dressed Hamiltonian is diagonalized, being the free
  field operator. The ground state is a generalized coherent state, centered
  around a possibly singular one-particle configuration.
\item In \cref{sec:wave-funct-renorm-1} we study the spin-boson, that is
  probably the most relevant of the three models from a physical
  standpoint. We are again able to renormalize the Hamiltonian for any
  distributional source, however the interplay between the singular dressing
  (that contains the interaction spin matrix) and the free spin matrix poses
  some technical challenges, for which we have to rely on a double operator
  integral representation, and $\Gamma$--convergence: \cref{app:DOI} focuses on the
  former in a novel weak formulation, which could be of interest on its own.
\item \cref{sec:wave-funct-renorm-2} is finally devoted to the IR problem in
  the Nelson model. In fact, the UV singularity of the Nelson Hamiltonian is
  cured by an additive self-energy renormalization alone
  \citep{nelson1964jmp}. Through the wave function renormalization, however,
  we are able to renormalize in the IR as well, and thus obtain the massless
  Nelson model without cutoffs in its ground state representation. We prove
  indeed that this renormalized model has a ground state whenever there is a
  binding particle potential, or, in the translation-invariant case, if we
  restrict to the fiber with zero total momentum.
\end{itemize}

\subsection*{Acknowledgments}
\label{sec:acknowledgements}

M.F.\ acknowledges the support of PNRR Italia Domani and Next Generation EU
through the ICSC National Research Centre for High Performance Computing, Big
Data and Quantum Computing; he also acknowledges the support by the Italian
Ministry of University and Research (MUR) through the grant “Dipartimento di
Eccellenza 2023-2027” of Dipartimento di Matematica, Politecnico di Milano,
and the PRIN 2022 grant ``OpeN and Effective quantum Systems (ONES)''.  B.H.\
and J.V.M.\ acknowledge support by the Ministry of Culture and Science of the
State of North Rhine-Westphalia within the project `PhoQC' (Grant
Nr. PROFILNRW-2020-067).

\subsection*{Notation}
Let us fix here some notation and conventions that will be used throughout
the text.
\begin{itemize}
\item Complex inner products are always anti-linear in the first and linear
  in the second argument.
\item Given any sesquilinear form $q(\psi,\ph)$, we will as usual denote by
  $q(\psi) = q(\psi,\psi)$ the associated quadratic form.
\item We denote algebraic tensor products, \emph{i.e.}, linear combinations
  of pure tensors, by $\hat\otimes$ and Hilbert space tensor products by
  $\otimes$.
\item The operator domains of a linear operator in a Hilbert space is denoted
  by the symbol $\sD(\cdot)$.  Since we also work with more general
  topological vector spaces, we emphasize that $\sD(\cdot)$ is always
  considered as subspace of the Hilbert space.
 \item For a Hilbert space operator $\omega$, we write $\omega\ge0$ in the sense of quadratic forms, i.e., $\braket{\psi,\omega\psi}\ge 0$ for all $\psi\in\sD(\omega)$.
 	In this case, we define $\omega^{-1}:\ran\omega\smallsetminus\{0\}\to\sD(\omega)$ as the canonical inverse and use the notation $v/\omega\coloneqq \omega^{-1}v$ for $v\in\ran\omega\smallsetminus\{0\}$.
\end{itemize}

\section{Wave Function Renormalization of the van Hove--Miyatake Model}
\label{sec:van-hove-miyatake}

The van Hove--Miyatake model was introduced independently in
\citep{vanhove1952physica,miyatake1952osaka} to describe a quantum scalar
field interacting with an unmovable source (\emph{e.g.}, a particle of very
large mass with no internal degrees of freedom). Despite it being not fully
relativistic (due to the immovable source), it still retains the most salient
features of field theories; it is therefore the perfect laboratory to study
problems in QFT.

Both infrared and ultraviolet divergences of the vHM model have been studied
in the past \citep[see][]{derezinski2003ahp,arai2020mps}; however, it was
only very recently that two of the authors of this paper could deal with such
divergences in full generality \citep{falconi2025prims}. Here we use the vHM
model to introduce in the simplest way the wave function renormalization
scheme in the Hamiltonian formalism, and to illustrate its most salient
features: \cref{sec:wave-funct-renorm-1,sec:wave-funct-renorm-2} rely heavily
on the concepts and notations introduced in this section.

\subsection{The vHM model in the Fock representation}
\label{sec:van-hove-model}

The vHM field theory can be seen, for sources that are sufficiently regular
both in the infrared and ultraviolet sectors, as a small perturbation of the
free scalar field, and thus it is well defined in the free field (Fock)
representation.

To set up the Fock representation, let us consider a Gel'fand triple
$(\mathscr{T},\mathfrak{h},\mathscr{T}')$, where:
\begin{itemize}
\item $\mathscr{T}$ is a TVS, $\mathscr{T}'$ is its continuous anti-dual, and
  $\mathfrak{h}$ is a Hilbert space;
\item $\mathscr{T}$ is continuously embedded in $\mathfrak{h}$ with dense image.
\item This also implies that $\sT$ and thus $\mathfrak h$ are continuously embedded in $\mathscr{T}'$ with dense image w.r.t. the  weak-$*$ topology ($\sigma(\sT',\sT)$ topology).\footnote{
		{\em Proof:} $(\sT',\sigma(\sT',\sT))$ is locally convex, so by
                the Hahn--Banach theorem and the fact that the weak-$*$ dual
                of $\sT'$ is $\sT$, density follows if for all $f\in\sT$ we can
                prove the implication $\forall g\in\sT:\braket{f,g}_\fh=0\Rightarrow f=0$. This
                is however obvious, since $\sT$ is dense in $\fh$. \qed}
\end{itemize}
We interpret $\mathscr{T}$ as the space of \emph{test functions} for the
field, $\mathfrak{h}$ as its \emph{one-particle space}, and $\mathscr{T}'$ as
the space of (classical) \emph{field configurations}. In view of the above
embeddings, we denote, for any $\varphi\in \mathscr{T}'$ and $f\in
\mathscr{T}$,
\begin{equation*}
  \langle \varphi,f \rangle_{}\equiv \overline{\varphi(f)}
\end{equation*}
the duality bracket between $\mathscr{T}$ and $\mathscr{T}'$, observing that by definition
whenever $\varphi\in \mathfrak{h}$,
\begin{equation*}
  \langle \varphi , f  \rangle_{}= \langle \varphi,f \rangle_{\mathfrak{h}}\;. 
\end{equation*}
We might write, with a slight abuse of notation, $\langle \varphi , f
\rangle_{\mathfrak{h}}$ for the duality pairing also if $\varphi\in
\mathscr{T}'\smallsetminus \mathfrak{h}$.
Whereas $\sT$ is dense in $\sT'$, it is not necessarily {\em sequentially} dense in $\sT'$. When using approximation arguments, we will thus formulate statements for approximating nets and call a net $\{g_\alpha\}_{\alpha\in \mathcal{A}}\subset \sT$ and $\{g_\alpha\}_{\alpha\in \mathcal{A}}\subset \fh$ converging to $g$ in the $\sigma(\sT',\sT)$-topology a $\sT$-{\em mollification} and $\fh$-{\em mollification} of $g$, respectively.
\begin{rem*}
  A sufficient criterion for $\sT$ to be sequentially weak-$*$ dense in
  $\sT'$ is that $\sT$ is locally convex, separable and metrizable,
  \emph{e.g.}, when choosing $\sT$ to be the Schwartz space $\cS(\mathbb{R}^d)$
  embedded in $L^2(\IR^d)$, or in other cases where an explicit mollification
  procedure is known, \emph{e.g.}, for the space $\cD(\IR^d\setminus\{0\})$ of smooth
  compactly supported functions with support in $\IR^d\setminus\{0\}$.
\end{rem*}
The Fock representation of the free scalar field over the one-particle space $\fh$ has the symmetric Fock space
\begin{equation*}
  \mathcal{F}(\mathfrak{h})= \bigoplus_{n\in \mathbb{N}} \mathfrak{h}_n
\end{equation*}
as ambient space, where $\mathfrak{h}_0=\mathbb{C}$, and
\begin{equation*}
  \mathfrak{h}_n= \underbrace{\mathfrak{h}\otimes_{\mathsf{s}}\dotsm\otimes_{\mathsf{s}}\mathfrak{h}  }_n\;.
\end{equation*}
A vector $\Psi\in \mathcal{F}(\mathfrak{h})$ will be denoted, fiber-wise, by
$\Psi=(\Psi_0,\Psi_1,\dotsc,\Psi_n,\dotsc)$, with inner product
\begin{equation*}
  \langle \Psi  , \Phi \rangle_{\mathcal{F}(\mathfrak{h})}= \sum_{n\in \mathbb{N}}^{}\langle \Psi_n  , \Phi_n \rangle_{\mathfrak{h}_n}\;.
\end{equation*}
The Fock space is the GNS representation of the C*-algebra of CCR
corresponding to the \emph{Fock vacuum state}, whose cyclic vector is
$\Omega=(1,0,\dotsc,0,\dotsc)$. A very useful dense subset of $\mathcal{F}(\mathfrak{h})$ is given by the
\emph{regular finite particle vectors}: given any dense subspace $\mathscr{U}\subseteq \mathfrak{h}$, the
finite particle vectors with $\mathscr{U}$ regularity are the dense subspace
\begin{multline*}
  \FSfin(\mathscr{U})= \Bigl\{\Psi\in \mathcal{F}(\mathfrak{h})\;,\; \exists n\in \mathbb{N}\;,\; \Psi=(\Psi_0,\Psi_1,\dotsc,\Psi_n,0,\dotsc,0,\dotsc)\\\text{ and } \forall j\in \{1,\dotsc,n\}\;,\;\Psi_j\in \underbrace{\mathscr{U}\hat{\otimes}_{\mathsf{s}} \dotsm\hat{\otimes}_{\mathsf{s}} \mathscr{U}}_j\Bigr\}\;.
\end{multline*}
The CCR-algebra itself is implemented through the \emph{creation and
  annihilation operators $a^{\dagger}(f)$ and $a(f)$}: for any $f\in
\mathfrak{h}$, the closed operators $a(f)$ and $a^{\dagger}(f)$ are defined
on their common core $\FSfin(\mathfrak{h})$ by their action on fiber-wise
symmetric tensor products
$\Psi_n=\psi_1^{(n)}\otimes_{\mathsf{s}}\dotsm\otimes_{\mathsf{s}}\psi_n^{(n)},
$
\begin{equation}
  \label{eq:1}
  \bigl(a(f)\Psi\bigr)_n = \frac{\sqrt{n+1}}{(n+1)!} \sum_{\sigma\in \Sigma_{n+1}}^{}\langle f  , \psi^{(n+1)}_{\sigma(1)} \rangle_{\mathfrak{h}}\:\psi^{(n+1)}_{\sigma(2)}\otimes \dotsm\otimes \psi^{(n+1)}_{\sigma(n+1)}\;;  
\end{equation}
and
\begin{equation}
  \label{eq:2}
  \bigl(a^{\dagger}(f)\Psi\bigr)_n = \sqrt{n} \bigl(f \otimes_{\mathsf{s}}\psi_1^{(n-1)}\otimes_{\mathsf{s}} \dotsm\otimes_{\mathsf{s}} \psi^{(n-1)}_{n-1}\bigr)\;. 
\end{equation}
The creation and annihilation operators are adjoint to one-another,
$a(f)^{*}=a^{\dagger}(f)$, and $a(f)+a^{\dagger}(f)$ is the Segal field operator,
essentially self-adjoint on $\FSfin(\mathfrak{h})$; the Fock representation of the
CCR-algebra now being generated by the unitary Weyl operators $\{W(f)\;,\; f\in
\mathscr{T}\}$, with
\begin{equation*}
  W(f)=\e^{i\pi \bigl(a^{\dagger}(f)+a(f)\bigr)}\;.
\end{equation*}
Free field Hamiltonians in the Fock representation take the form of
\emph{second quantization operators}: given a single-boson self-adjoint
operator $T$ with domain $\mathscr{D}(T)\subseteq \mathfrak{h}$, its \emph{second quantization} $\dG(T)$
is the self-adjoint operator with core $\FSfin\bigl(\mathscr{D}(T)\bigr)$ and
fiber-wise action
\begin{equation*}
  \bigl(\dG(T)\Psi\bigr)_n= \Bigl(\sum_{j=1}^n \mathds{1}\otimes \dotsm\otimes \underbrace{T}_j\otimes \dotsm\otimes \mathds{1}\Bigr) \Psi_n\;.
\end{equation*}
In particular, the second quantization $\dG(\dsone)$ of $\mathds{1}$ gives rise to the so-called \emph{number
  operator}, that counts how many Fock excitations are on each fiber, and the free scalar field Hamiltonian $H_0$ is given by
\begin{equation*}
  H_0=\dG(\omega)\;,
\end{equation*}
where $\omega\geq 0$ is the \emph{field's dispersion relation}.

\begin{defn}[Fock vHM Hamiltonian]
  \label[definition]{def:4}
  For any (regular) form factor $v\in \mathfrak{h}$, and dispersion relation $\omega\geq 0$
  (self-adjoint on $\mathscr{D}(\omega)\subseteq \mathfrak{h}$), the \emph{Fock vHM Hamiltonian} $H(v)$ is
  defined as
  \begin{equation*}
    H(v)= \dG(\omega)+ a^{\dagger}(v)+a(v)\;.
  \end{equation*}
\end{defn}
\begin{ex*}[Concrete vHM model]
  The concrete vHM model describes a scalar field in $\mathbb{R}^d$, created and
  annihilated by a charged particle with charge distribution
  $\check{v}(x)$. It is then natural to set $\mathfrak{h}=L^2(\mathbb{R}^d)$, whence the field's
  dispersion relation is $\omega(k)=\sqrt{k^2+\mu^2}$, $\mu\geq 0$ being the field's
  mass -- seen as a multiplication operator on $L^2(\mathbb{R}^d)$. Observe that
  $\omega^{-\alpha}$, $\alpha>0$, is bounded if $\mu>0$, while it is unbounded if $\mu=0$.

  Finally, let us remark that if $\mu>0$, $v\in \mathfrak{h}$ implies $v/\omega^{\alpha}\in \mathfrak{h}$ for any
  $\alpha>0$, while for $\mu=0$ there are $v\in \mathfrak{h}$ such that $v/\omega^{\alpha}\notin \mathfrak{h}$ for some
  $\alpha$, and conversely $v\notin \mathfrak{h}$ such that $v/\omega^{\alpha}\in \mathfrak{h}$ for some $\alpha>0$. This is
  expression of the well-known fact that massless fields can have a singular
  behavior in the IR, as well as in the UV.
\end{ex*}

The spectral properties of the vHM Hamiltonian depend on the regularity of
$v$, as illustrated in the following proposition.
\begin{prop}
  \label[proposition]{prop:5}
  Let $v\in \mathfrak{h}$, and $\omega\geq 0$. Then:
  \begin{itemize}
    \setlength{\itemsep}{3mm}
  \item $H(v)$ is essentially self-adjoint on $\FSfin\bigl(\mathscr{D}(\omega)\bigr)$ and
    \begin{equation*}
      \inf \sigma\bigl(H(v)\bigr) = -\lVert v/\sqrt{\omega}  \rVert_{\mathfrak{h}}^2
    \end{equation*}
    (the latter being $-\infty$ whenever $v/\sqrt{\omega}\notin \mathfrak{h}$);
  \item If $v/\sqrt{\omega}\in \mathfrak{h}$, then $H(v)$ is self-adjoint on $\mathscr{D}(\dG(\omega))$ and
    bounded from below, and it has no Fock ground state whenever $v/\omega \notin \mathfrak{h}$;
  \item If $g\coloneqq -v/\omega \in \mathfrak{h}$, then
    \begin{equation*}
      \e^{a(g)-a^{\dagger}(g)}H(v)\e^{a^{\dagger}(g)-a(g)}= \mathsf{d}\Gamma(\omega) - \lVert v/\sqrt{\omega}  \rVert_{\mathfrak{h}}^2\;,
    \end{equation*}
    and thus the Fock coherent state $c_g= \e^{a^{\dagger}(g)-a(g)}\Omega$ is the unique
    ground state of $H(v)$.
  \end{itemize}
\end{prop}
\begin{proof}
  The first result is an application of Nelson's commutator theorem
  \citep[][Thm.\ X.36]{ReedSimon.1975} (see also \citep{Falconi.2015}) for the
  self-adjointness part; the infimum of the spectrum is proved to be $-\infty$
  whenever $v/\sqrt{\omega}\notin \mathfrak{h}$ by observing that the coherent state $c_{g_n}$ is
  such that
  \begin{equation*}
    \lim_{n\to \infty}\langle c_{g_n}  , H(v)c_{g_n} \rangle_{\mathcal{F}(\mathfrak{h})} = -\infty
  \end{equation*}
  whenever $g_n=-v_n/\omega$, with $g_n\in \mathfrak{h}$ and $v_n\to v$ in $\mathfrak{h}$. If $v/\sqrt{\omega}\in
  \mathfrak{h}$, the fact that
  \begin{equation*}
    \inf_{}\sigma\bigl(H(v)\bigr)= -\lVert v/\sqrt{\omega}  \rVert_{\mathfrak{h}}^2
  \end{equation*}
  is proved in \citep{derezinski2003ahp}, as well as the second and third
  points in the proposition.
\end{proof}
The case $v/\sqrt{\omega}\in \mathfrak{h}$ with $v/\omega\notin \mathfrak{h}$ is called ``the infrared
catastrophe'': the ground state does not belong to the Fock representation
since it has an \emph{infinite} number of excitations of very low energy,
called ``soft photons'', and Fock vectors have smaller and smaller
probability of having a larger and larger number of excitations. More
precisely, the following lemma is proved in \citep{arai2020mps}.
\begin{lem}
  \label[lemma]{lemma:7}
  Let $v\in \mathfrak{h}$ such that $v/\sqrt{\omega}\in \mathfrak{h}$ and $v/\omega\notin \mathfrak{h}$, and let $v_n\to v$ be an
  $\mathfrak{h}$-convergent sequence such that $v_n/\sqrt{\omega},v_n/\omega\in \mathfrak{h}$. Then the
  sequence of ground states $c_{g_n}$ of $H(v_n)$ satisfies:
  \begin{itemize}
    \setlength{\itemsep}{2mm}
  \item In the weak topology of $\mathcal{F}(\mathfrak{h})$,
    \begin{equation*}
      c_{g_n}\rightharpoonup 0\;;
    \end{equation*}
  \item
    \begin{equation*}
      \lim_{n\to \infty}\langle c_{g_n}  , \dG(\mathds{1})c_{g_n} \rangle_{\mathcal{F}(\mathfrak{h})}=+\infty\;.
    \end{equation*}
  \end{itemize}
\end{lem}
The unitary operator $\e^{a^{\dagger}(g)-a(g)}$ that diagonalizes the vHM
Hamiltonian $H(v)$ is the \emph{unitary dressing transformation} originally
introduced by Gross in the polaron model \citep{gross1962ap,nelson1964jmp}.

\subsection{The vHM model with singular interactions}
\label{sec:vhm-model-with}

Recalling the Gel'fand triple $(\mathscr{T}, \mathfrak{h}, \mathscr{T}')$ that we used as a starting point
to construct the Fock representation, we would now like to extend the van
Hove model to any distributional form factor $v\in \mathscr{T}'$.

The first key observation is that, exploiting the $\mathscr{T},\mathscr{T}'$ duality bracket and
the continuous embeddings $\mathscr{T} \hookrightarrow \mathfrak{h} \hookrightarrow \mathscr{T}'$, annihilation makes sense also for
singular one-particle vectors.
\begin{defn}[Singular annihilation]
  \label[definition]{def:5}
  The annihilation operator extends uniquely to generic field configurations
  in $\mathscr{T}'$: for any $g\in \mathscr{T}'$, $a(g):\FSfin(\mathscr{T})\to \FSfin(\mathscr{T})$ is defined by its
  action on fiber-wise products $\Psi_n=
  \psi_1^{(n)}\otimes_{\mathsf{s}}\dotsm\otimes_{\mathsf{s}}\psi_n^{(n)}$,
  \begin{equation*}
    \bigl(a(g)\Psi\bigr)_n= \frac{\sqrt{n+1}}{(n+1)!} \sum_{\sigma\in \Sigma_{n+1}}^{}\langle g  , \psi_{\sigma(1)}^{(n+1)} \rangle_{} \: \psi_{\sigma(2)}^{(n+1)}\otimes\dotsm\otimes\psi^{(n+1)}_{\sigma(n+1)}\;.  
  \end{equation*}
  For any $g\in \mathfrak{h}$, this definition coincides with the standard one \cref{eq:1}.
\end{defn}
Let us stress the (perhaps obvious) fact that such a generalization is never
possible for the creation operator, cf.\ \cref{eq:2}: we cannot create an excitation outside of
$\mathfrak{h}$ without leaving the Fock space $\mathcal{F}(\mathfrak{h})$ altogether.
This reflects in the fact that $\sD(a(g)^*)=\{0\}$ if $g\in\sT'\smallsetminus\fh$.

Using generalized
annihilation, we can nevertheless define the vHM model for singular sources, at least as a
densely defined quadratic form in the Fock representation.

\begin{defn}[vHM form]
  \label[definition]{def:6}
  Let $(\mathscr{T},\mathfrak{h},\mathscr{T}')$ be the Gel'fand triple of the field, with the dispersion
  $\omega:\mathscr{T}\to \mathscr{T}$ being a linear homeomorphism that extends to a positive
  operator\footnote{We denote the extension of $\omega\bigr\rvert_{\mathscr{T}}$ to a
    self-adjoint operator on $\mathfrak{h}$ with the same symbol: it is understood that
    it is the very same extension that appears in the vHM
    Hamiltonian/quadratic form. Also, we assume implicitly that $\mathscr{T}$
    is a core for the operator $\omega$ on $\mathfrak{h}$, even though this
    is mostly unnecessary.} $\omega\geq 0$ on $\mathfrak{h}$, and $v\in
  \mathscr{T}'$. Then the \emph{vHM sesquilinear form} $q_v:
  \FSfin(\mathscr{T})\times \FSfin(\mathscr{T})\to \mathbb{C}$ is defined by
  \begin{equation*}
    q_v(\Psi,\Phi)= \langle \Psi  , \mathsf{d}\Gamma(\omega)\Phi \rangle_{\mathcal{F}(\mathfrak{h})} + \langle a(v)\Psi  , \Phi \rangle_{\mathcal{F}(\mathfrak{h})}+ \langle \Psi  , a(v)\Phi \rangle_{\mathcal{F}(\mathfrak{h})}\;.
  \end{equation*}
  For all $v\in \mathfrak{h}$, $q_v(\Psi,\Phi)= \langle \Psi , H(v)\Phi
  \rangle_{\mathcal{F}(\mathfrak{h})}$.
\end{defn}
The quadratic form associated to $q_v$ is, for generic $v\in \mathscr{T}'$,
neither closed nor bounded from below: it requires \emph{renormalization}. Two
types of divergences might plague $q_v$:
\begin{itemize}
  \setlength{\itemsep}{3mm}
\item \emph{Energy divergences}. Divergent terms in the Hamiltonian (in vHM,
  a constant) make the form unbounded from below; this is typically a UV
  behavior in particle-field models, but might happen in the IR sector as
  well.
\item \emph{Wave function divergences}. The characteristic states of the
  system, starting from the ground state itself, are singular: either they
  have too many excitations (like in the IR catastrophe), or they have a
  number of singular excitations (this happens typically for strong UV
  divergences).
\end{itemize}
Before discussing such renormalization procedures, let us refine
\cref{prop:5}. Observe that for $v\notin \mathfrak{h}$ such that
$v/\sqrt{\omega}\in \mathfrak{h}$, the quadratic form $q_v(\Psi,\Psi)$ is the
free form $\langle\Psi, \mathsf{d}\Gamma(\omega)\Psi \rangle_{}$ plus a small
form perturbation of it -- $2\mathsf{Re}\langle \Psi , a(f)\Psi \rangle_{}$
-- and thus by KLMN theorem \citep[][Thm.\ X.17]{ReedSimon.1975}, it is the form
of a unique self-adjoint operator $H(v)= \dG(\omega)\overset{\cdot}{+}
\bigl(a^{*}(v)+a(v)\bigr)$, where $\overset{\cdot}{+}$ denotes the form sum.
\begin{prop}
  \label[proposition]{prop:6}
  Let $v/\sqrt{\omega}\in \mathfrak{h}$, and $\omega\geq 0$. Then $H(v)=
  \dG(\omega)\overset{\cdot}{+} \bigl(a^{*}(v)+a(v)\bigr)$ is self-adjoint on
  $\mathscr{D}\bigl(H(v)\bigr)\subseteq
  \mathscr{D}\bigl(\dG(\omega)^{1/2}\bigr)$ and bounded from below.
\end{prop}

\subsection{Energy renormalization}
\label{sec:energy-renorm}
Energy divergences are renormalized by subtraction, and they are implemented
through the dressing transformation, that in this solvable model
\emph{diagonalizes the Hamiltonian exactly} (we will see that it is not
always the case for the non-diagonalizable models). Consider the case in which
$v\in \mathscr{T}'\smallsetminus \mathfrak{h}$ be such that
$v/\sqrt{\omega}\notin \mathfrak{h}$ and $v/\omega\in \mathfrak{h}$; and let
$\{v_n\}_{n\in \mathbb{N}}\subset \mathscr{T}$ be a mollification of $v$:
$v_n\to v$ in the $\sigma(\mathscr{T}',\mathscr{T})$-topology of
$\mathscr{T}'$. We suppose that $\sqrt{\omega}\bigr\rvert_{\mathscr{T}}$ is
also an homeomorphism, so that $v_n/\sqrt{\omega}\in \mathscr{T}$.

Then we have the following result of \emph{energy renormalization}
\citep[see, \emph{e.g.},][]{derezinski2003ahp}.
\begin{prop}
  \label[proposition]{prop:7}
  Let $v_\alpha\xrightarrow{\alpha} v$ be an $\fh$-mollification of $v\in\sT'$ such that $v_\alpha/\sqrt\omega\in\fh$ for all $\alpha$ and $\fh\ni v_\alpha/\omega\xrightarrow[\alpha]{\fh}v/\omega\in\fh$. Furthermore, define
  \begin{equation*}
    H_{\mathsf{ren}}(v_\alpha)= H(v_\alpha) + \lVert v_\alpha/\sqrt{\omega}  \rVert_{\mathfrak{h}}^2\;.
  \end{equation*}
  Then, as self-adjoint operators on $\mathcal{F}(\mathfrak{h})$,
  \begin{equation*}
    H_{\mathsf{ren}}(v_\alpha) \xrightarrow[\alpha]{\lVert \,\cdot \,  \rVert_{}^{}\mathsf{-res}} H_{\mathsf{ren}}(v)\;,
  \end{equation*}
  the latter being defined by
  \begin{equation*}
    H_{\mathsf{ren}}(v)= \e^{a^{\dagger}(g)-a(g)}\,\dG(\omega)\, \e^{a(g)-a^{\dagger}(g)}\;, 
  \end{equation*}
  with $g=-\frac{v}{\omega}$.
\end{prop}
\begin{proof}
  Firstly, remark that $g,g_\alpha\in \mathfrak{h}$, so $\e^{a(g)-a^{\dagger}(g)}$
  and $\e^{a(g_\alpha)-a^{\dagger}(g_\alpha)}$ are unitary operators on
  $\mathcal{F}(\mathfrak{h})$. Now, let $g_\alpha=-\frac{v_\alpha}{\omega}$, and define
  \begin{equation*}
    \hat{H}_{\mathsf{ren}}(v_\alpha)= \e^{a(g_\alpha)-a^{\dagger}(g_\alpha)}H_{\mathsf{ren}}(v_\alpha)\e^{a^{\dagger}(g_\alpha)-a(g_\alpha)}\;,
  \end{equation*}
  and
  \begin{equation*}
    \hat{H}_{\mathsf{ren}}(v)= \e^{a(g)-a^{\dagger}(g)}H_{\mathsf{ren}}(v)\e^{a^{\dagger}(g)-a(g)}\;.
  \end{equation*}
  It follows that for every $\alpha$,
  \begin{equation*}
    \hat{H}_{\mathsf{ren}}(v_\alpha)= \hat{H}_{\mathsf{ren}}(v)=\dG(\omega)\;.
  \end{equation*}
  Therefore, $\hat{H}_{\mathsf{ren}}(v_\alpha)\xrightarrow[\alpha]{\lVert
    \,\cdot \, \rVert_{}^{} \mathsf{-res}}\hat{H}_{\mathsf{ren}}(v)$
  trivially. On the other hand,
  $(\dG(\omega)+1)^{-1/2}\e^{a(g_\alpha)-a^{\dagger}(g_\alpha)}\to
  (\dG(\omega)+1)^{-1/2}\e^{a(g)-a^{\dagger}(g)}$ in norm operator topology;
  the result then follows from \citep[][Thm.\ A.1]{GriesemerWuensch.2018}.
\end{proof}

\subsection{Wave function renormalization I: renormalized Hilbert space}
\label{sec:wave-funct-renorm}

If energy renormalization can be considered a classical problem for
particle-field Hamiltonians, and it has been carried out successfully in all
models excluding the Pauli-Fierz model of quantum electrodynamics, wave
function renormalization has been somewhat overlooked so far: in the
literature, it is always remarked that, whenever $g\notin \mathfrak{h}$ (see
\cref{prop:5}), the dressing transformation cannot be defined as a unitary
operator and thus renormalization cannot be implemented (in the Fock
representation); perhaps surprisingly, until our recent paper
\citep{falconi2025prims} it was not realized that it is possible to resort to
a \emph{non-unitary dressing}, through which the wave function
renormalization is implemented by constructing a new representation,
inequivalent to the Fock one.

Since wave function renormalization must get rid of the singularity of the
ground state (and, in a cascade, of all other characteristic states of the
system), a new inner product -- and by completion, a new Hilbert space --
must be constructed. The wave function singularity is reflected in the fact
that the ``new'' Hilbert space carries a representation of the CCR that is
\emph{inequivalent to the Fock representation}.

The starting point is to observe that the singular annihilation defined in
\cref{def:5} satisfies the following properties.
\begin{lem}
  \label[lemma]{lemma:2}
  Let $g\in \mathscr{T}'$. Then:
  \begin{itemize}
  \item For any $k\in \mathbb{N}$,
    $\bigl(a(g)\bigr)^k\bigl(\FSfin(\mathscr{T})\bigr)\subset
    \FSfin(\mathscr{T})$;
  \item For any $\Psi\in \FSfin(\mathscr{T})$, there exists $k_{\Psi}\in
    \mathbb{N}$ such that $\bigl(a(g)\bigr)^k\Psi=0$ for any $k> k_{\Psi}$;
  \item For all $\Psi\in\FSfin(\sT)$ and $k\in\IN$, the map $\sT'\to\FS(\fh),\,h\mapsto a(h)^k\Psi$ is continuous w.r.t. $\FS(\fh)$--$\sigma(\sT',\sT)$-topology.
  \end{itemize}
\end{lem}
\begin{proof}
  The first result is clear from the definition of $a(g)$ given in
  \eqref{eq:3}: on any fiber, $\bigl(a(g)\Psi\bigr)_n\in
  \mathscr{T}\hat{\otimes}_{\mathsf{s}}\dotsm\hat{\otimes}_{\mathsf{s}}\mathscr{T}
  $; the second result follows from the fact that for any $\Psi\in
  \FSfin(\mathscr{T})$, there exists $n_{\psi}\in \mathbb{N}$ such that
  $\Psi_n=0$ for any $n> n_{\psi}$, and each action of $a(g)$ makes the
  component on each fiber to ``jump one fiber down'' -- in other words,
  $\bigl(a(g)\Psi\bigr)_n$ depends only on $\Psi_{n+1}$; finally, the
  continuity result follows from the continuity of $h\mapsto \braket{h,f}$ for any $f\in\sT$.
\end{proof}
Thanks to \cref{lemma:2}, we can define for any $g\in \mathscr{T}'$ the
\emph{singular dressing transformation} $\e^{a(g)}$ on $\FSfin(\mathscr{T})$:
\begin{equation*}
  \e^{a(g)}\Psi = \sum_{k\in \mathbb{N}}^{} \frac{\bigl(a(g)\bigr)^k}{k!} \Psi\;,
\end{equation*}
and the series converges (it is actually truncated) since
$\bigl(a(g)\bigr)^{k}\Psi=0$ for any $k>k_{\Psi}$.
\begin{prop}
  \label[proposition]{prop:1}
  For any $g\in \mathscr{T}'$, the map $\e^{a(g)}:\FSfin(\mathscr{T})\to
  \FSfin(\mathscr{T})$ is injective.
\end{prop}
\begin{proof}
  Suppose that $\Psi\in \mathsf{ker}(\e^{a(g)})$. Then, for all $n\in
  \mathbb{N}$,
  \begin{equation*}
    \bigl(\e^{a(g)}\Psi\bigr)_n=0\;.
  \end{equation*}
  Suppose that $k^{\mathsf{min}}_{\Psi}\neq 0$ ($k^{\mathsf{min}}_{\Psi}$
  being the smallest number such that $\Psi_k=0$ for all $k>
  k^{\mathsf{min}}_{\Psi}$). Then it follows that
  \begin{equation*}
    \bigl(\e^{a(g)}\Psi\bigr)_{k^{\mathsf{min}}_{\Psi}}= \Psi_{k^{\mathsf{min}}_{\Psi}}\;,
  \end{equation*}
  since the left hand side can depend only on the fibers $\Psi_{k}$, $k\geq
  k^{\mathsf{min}}_{\Psi}$. Therefore,
  \begin{equation}
    \label{eq:4}
    \Psi_{k^{\mathsf{min}}_{\Psi}}= \bigl(\e^{a(g)}\Psi\bigr)_{k^{\mathsf{min}}_{\Psi}}=0\;.
  \end{equation}
  Hence, $k^{\mathsf{min}}_{\Psi}-1$ also satisfies the condition $\Psi_k=0$
  for all $k> k^{\mathsf{min}}_{\Psi}-1$, contradicting the minimality of
  $k_{\Psi}^{\mathsf{min}}$. Thus, $k_{\Psi}^{\mathsf{min}}=0$, and
  \eqref{eq:4} yields that, indeed, $\Psi=0$.
\end{proof}

The renormalized inner product on $\FSfin(\mathscr{T})$ stems from the
singular dressing $\e^{a(g)}$: for any $\Psi,\Phi\in \FSfin(\mathscr{T})$,
\begin{equation*}
  \langle \Psi   , \Phi  \rangle_g \coloneqq  \langle \e^{a(g)}\Psi  , \e^{a(g)}\Phi \rangle_{\mathcal{F}(\mathfrak{h})}\;.
\end{equation*}
The sesquilinearity and positive-definiteness are straightforwardly inherited
by the sesquilinearity and positive-definiteness of the Fock space inner
product, and the non-degeneracy comes from the injectivity of
$\e^{a(g)}$. For any $g\in \mathscr{T}'$, the inner product space
\begin{equation*}
  \bigl(\FSfin(\mathscr{T}), \langle \cdot   , \cdot  \rangle_g\bigr)
\end{equation*}
is not complete, and is thus a pre-Hilbert space.
\begin{defn}[Renormalized Hilbert space]
  \label[definition]{def:2}
  For any $g\in \mathscr{T}'$, the $g$-\emph{renormalized Hilbert space} is
  defined as
  \begin{equation*}
    \mathcal{F}_g =\overline{\FSfin(\mathscr{T})}^{\lVert\, \cdot \,  \rVert_g}\;, 
  \end{equation*}
  \emph{i.e.}\ the closure of the pre-Hilbert space
  $\bigl(\FSfin(\mathscr{T}), \langle \cdot , \cdot \rangle_g\bigr)$ with
  respect to the topology induced by the norm $\lVert \cdot \rVert_g^{}=
  \sqrt{\langle \cdot , \cdot \rangle_g}$.
\end{defn}
The renormalized Hilbert space can be seen as a space of ``singular Fock
vectors'' as follows. Let us define, for any \emph{regular} $g\in \mathfrak{h}$, the
operator $\e^{a^{\dagger}(g)}$ on $\FSfin(\mathfrak{h})$, again by series expansion. The
convergence in $\lVert \cdot \rVert_{\mathcal{F}(\mathfrak{h})}^{}$-norm of the series $\sum_{k\in
  \mathbb{N}}^{}\frac{\bigl(a^{\dagger}(g)\bigr)^k}{k!}\Psi$ can be easily proved by induction
from the definition \eqref{eq:2} of creation operators, observing that
\begin{gather*}
  \e^{a^{\dagger}(g)}\Omega = \sum_{n\in \mathbb{N}}^{}\frac{1}{\sqrt{n!}} g^{\otimes n}\;;\\
  \bigl\lVert \e^{a^{\dagger}(g)}\Omega  \bigr\rVert_{\mathcal{F}(\mathfrak{h})}^2=  \e^{\lVert g  \rVert_{\mathfrak{h}}^2}\;.
\end{gather*}
From the definition of $\e^{a^{\dagger}(g)}$, canonical commutation relations, and
the Baker-Campbell-Hausdorff formula for unbounded operators that commute
with their commutator, we obtain the following result.
\begin{lem}
  \label[lemma]{lemma:3}
  For any $g\in \mathfrak{h}$, and $\Psi,\Phi\in \FSfin(\mathfrak{h})$, we have that
  \begin{equation*}
    \langle \Psi  , \Phi \rangle_g= \frac{\langle \e^{a^{\dagger}(g)}\Psi  , \e^{a^{\dagger}(g)}\Phi \rangle_{\mathcal{F}(\mathfrak{h})}}{\langle \e^{a^{\dagger}(g)}\Omega  , \e^{a^{\dagger}(g)}\Omega \rangle_{\mathcal{F}(\mathfrak{h})}}\;.
  \end{equation*}
\end{lem}
Combining \cref{lemma:3} with the third property given in \cref{lemma:2}, we
finally prove the proposition below.
\begin{prop}
  \label[proposition]{prop:2}
  Let $g\in \mathscr{T}'$, $\Psi,\Phi\in \FSfin(\mathscr{T})$. Then for any $\fh$-mollification $g_\alpha\xrightarrow{\alpha} g$, we have that
  \begin{equation}
    \label{eq:5}
    \langle \Psi  , \Phi \rangle_g = \lim_{\alpha} \frac{\langle \e^{a^{\dagger}(g_\alpha)}\Psi  , \e^{a^{\dagger}(g_\alpha)}\Phi \rangle_{\mathcal{F}(\mathfrak{h})}}{\langle \e^{a^{\dagger}(g_\alpha)}\Omega  , \e^{a^{\dagger}(g_\alpha)}\Omega \rangle_{\mathcal{F}(\mathfrak{h})}}\;.
  \end{equation}
\end{prop}
\begin{rem*}
  The vectors $\e^{a^{\dagger}(g_\alpha)}\Psi$, $\Psi\in \FSfin(\mathscr{T})$ have divergent Fock-norms when taking the limit; thus making the elements of the renormalized space ``singular''
  if they were to be viewed as Fock vectors. However, since Fock vectors are
  generated from the vacuum by the action of an arbitrary number of creation
  operators, it might not be surprising that it is enough to divide by the
  Fock vacuum contribution to obtain a well-defined Hilbert space.
\end{rem*}

\subsection*{The dressing maps}
\label{sec:textbfthe-dress-maps}

Let us now study the singular dressings $\e^{a(g)}$ as maps between Hilbert
spaces. Let us denote by $u_g: \FSfin(\mathscr{T})_g\to \FSfin(\mathscr{T})_0$ the action of the
dressing
\begin{equation*}
  u_g\Psi = \e^{a(g)}\Psi\;,
\end{equation*}
where for any $g\in \mathscr{T}'$, $\FSfin(\mathscr{T})_g$ is $\FSfin(\mathscr{T})$ identified as a $\lVert \cdot
\rVert_g^{}$-dense subspace of $\mathcal{F}_g$ (and therefore $u_g$ is a map from $\mathcal{F}_g$ to
$\mathcal{F}_0$).
\begin{lem}
  \label[lemma]{lemma:4}
  For any $g\in \mathscr{T}'$, the map $u_g: \mathcal{F}_g\to \mathcal{F}_0$ is a densely defined linear
  isometry.
\end{lem}
\begin{proof}
  By construction of $\mathcal{F}_g$, $\FSfin(\mathscr{T})$ is dense, and for all $\Psi\in \FSfin(\mathscr{T})$,
  \begin{equation*}
    \lVert \Psi  \rVert_g^{}= \lVert \e^{a(g)}\Psi  \rVert_{\mathcal{F}(\mathfrak{h})}^{}= \lVert u_g\Psi  \rVert_{\mathcal{F}_0}^{}\;.
    \qedhere
  \end{equation*}
\end{proof}
We can actually see the $u_g=u_{g,0}$ as members of a two-parameter group of
densely defined linear isometries: for any $g,g'\in \mathscr{T}'$, define
$u_{g,g'}: \FSfin(\mathscr{T})_g\to \FSfin(\mathscr{T})_{g'}$ by
\begin{equation*}
  u_{g,g'}\Psi = \e^{a(g-g')}\Psi\;.
\end{equation*}
Now, extend them uniquely, by density, to everywhere defined isometries
$U_{g,g'}: \mathcal{F}_g\to \mathcal{F}_{g'}$. The collection of these
isometries is actually a \emph{two-parameter unitary group}.
\begin{prop}
  \label[proposition]{prop:unitarygroup}
  The collection $(U_{g,g'})_{g,g'\in \mathscr{T}'}$ forms a two-parameter
  unitary group, namely the following properties hold true: for any
  $g,g',g''\in \mathscr{T}'$,
  \begin{itemize}
  \item $U_{g,g'}: \mathcal{F}_g\to \mathcal{F}_{g'}$ is unitary;
  \item $U_{g,g}= \mathds{1}$;
  \item $U_{g',g}= U^{-1}_{g,g'}$;
  \item $U_{g',g''}U_{g,g'}= U_{g,g''}$
  \end{itemize}
\end{prop}
To prove the result, we need to introduce exponential vectors on
$\mathcal{F}_g$. Let us fix $f\in \mathscr{T}$, and for any $n\in
\mathbb{N}$, define the vector $\varepsilon_n(f)\in \FSfin(\mathscr{T})$ by
its fiber value on each $\mathfrak{h}_m$:
\begin{equation}
  \label{eq:7}
  \bigl(\varepsilon_n(f)\bigr)_m=
  \begin{cases}
    \; 0&\text{if }m\neq n,\\
    \; \tfrac{1}{\sqrt{n!}} f^{\otimes n}&\text{if }m=n.
  \end{cases}
\end{equation}
\begin{lem}
  \label[lemma]{lemma:6}
  For any $g\in \mathscr{T}'$ and $f\in \mathcal{T}$, the exponential vector
  \begin{equation*}
    \varepsilon(f)= \sum_{n\in \mathbb{N}}^{} \varepsilon_n(f) 
  \end{equation*}
  is an absolutely convergent series in $\mathcal{F}_g$. Furthermore, in
  $\mathcal{F}_0$,
  \begin{equation}
    \label{eq:6}
    \e^{a(g)}\varepsilon(f)= \e^{\langle g  , f \rangle_{}}\varepsilon(f)\;.
  \end{equation}
\end{lem}
\begin{proof}
  The action $\e^{a(g)}\varepsilon_n(f)$ can be computed explicitly:
  \begin{equation*}
    \e^{a(g)}\varepsilon_n(f)= \sum_{k=0}^n \tfrac{1}{k!}\langle g  , f \rangle_{}^k\varepsilon_{n-k}(f)\;.
  \end{equation*}
  A resummation shows that
  \begin{equation*}
    \e^{a(g)}\sum_{n=0}^N \varepsilon_n(f) = E_N(g,f) \sum_{n=0}^N\varepsilon_n(f)\;,
  \end{equation*}
  where
  \begin{equation*}
    \mathbb{C}\ni E_N(g,f)= \sum_{k=0}^N \frac{\bigl(\langle g  , f \rangle_{}\bigr)^k}{k!}
  \end{equation*}
  is a truncated exponential, that converges absolutely as $N\to
  \infty$. Therefore, $\sum_{n=0}^N\varepsilon_n(f)$ converges absolutely in
  $\mathcal{F}_g$ if and only if it converges absolutely in
  $\mathcal{F}_0$. In $\mathcal{F}_0$, $\langle \varepsilon_n(f) ,
  \varepsilon_m(f) \rangle_{\mathcal{F}(\mathfrak{h})}=
  \frac{\delta_{nm}}{n!} \lVert f \rVert_{\mathfrak{h}}^{2n}$, and therefore
  $\varepsilon(f)$ is an absolutely convergent series and
  \begin{equation*}
    \lVert \varepsilon(f)  \rVert_{\mathcal{F}(\mathfrak{h})}^2= \e^{\lVert f  \rVert_{\mathfrak{h}}^2}\;.
    \qedhere
  \end{equation*}
\end{proof}
To avoid confusion, we might refer to the vector $\varepsilon(f)$ on
$\mathcal{F}_g$ as $\varepsilon_g(f)$.
\begin{proof}[Proof of \cref{prop:unitarygroup}]
  Once unitarity is proved, the group properties follow straightforwardly
  from the fact that on $\FSfin(\mathscr{T})$, $u_{g,g'}\Psi=
  \e^{a(g-g')}\Psi$.

  To prove unitarity, it suffices to prove that the maps $U_{g,g'}$ have
  dense range; let us start with the maps $U_{g,0}$. By \cref{lemma:6}, we
  have that
  \begin{equation*}
    U_{g,0}\Bigl(\mathsf{span} \bigl\{\varepsilon_g(f)\;,\; f\in \mathscr{T}\bigr\}\Bigr)= \mathsf{span}\bigl\{\varepsilon_0(f)\;,\; f\in \mathscr{T}\bigr\}\;.
  \end{equation*}
  Since $\bigl\{\varepsilon_0(f)\;,\; f\in \mathscr{T}\bigr\}$ is total
  \citep[see, \emph{e.g.},][Thm.\ 5.37]{Arai.2018}, $U_{g,0}$ is unitary, and
  $\bigl\{\varepsilon_g(f)\;,\; f\in \mathscr{T}\bigr\}$ is total in
  $\mathcal{F}_g$.

  Now, observe that since $\varepsilon(f)$ is defined on each $\mathcal{F}_g$
  by absolute convergence of the truncated series -- and the latter belongs
  to $\FSfin(\mathscr{T})$ -- the relation \eqref{eq:6} can be extended to
  yield
  \begin{equation*}
    U_{g,g'}\varepsilon_g(f)= \e^{\langle g-g'  , f \rangle_{}}\varepsilon_{g'}(f)\;.
  \end{equation*}
  Since as proved above, $\bigl\{\varepsilon_{g'}(f)\;,\; f\in
  \mathscr{T}\bigr\}$ is total in $\mathcal{F}_{g'}$, it follows that
  $U_{g,g'}$ is unitary for any $g,g'\in \mathscr{T}'$.
\end{proof}
\begin{cor}
  \label[corollary]{cor:2}
  For any $g,g'\in \mathscr{T}'$ and $f\in \mathscr{T}$,
  \begin{equation*}
    U_{g,g'}\varepsilon_g(f)= \e^{\langle g-g'  , f \rangle_{}}\varepsilon_{g'}(f)\;.
  \end{equation*}
\end{cor}
We can use the unitary maps $U_{g,g'}$ defined above to study the equivalence
or inequivalence of renormalized Fock spaces as representations of the Weyl
C*-algebra of canonical commutation relations. In fact, while all separable
Hilbert spaces are canonically isomorphic, it is well-known that they might
carry inequivalent representations of the Weyl C*-algebra describing quantum
fields, and that such inequivalent representations play a crucial role in
physics, for example distinguishing between free and interacting theories in
high-energy physics, or between the representation of observables before and
after the thermodynamic limit in statistical mechanics \citep[see,
\emph{e.g.},][and references
therein]{haag1992tmp,bratteli1997tmp2,streater2000plp}.

Let us recall that the (smeared) Segal field operator $\phi_0(f)$ is the self-adjoint
operator on $\mathcal{F}_0$ defined for all $f\in \mathscr{T}$ by
\begin{equation*}
  \phi_0(f)= \tfrac{1}{\sqrt{2}}\bigl(a^{\dagger}(f)+a(f)\bigr)\;.
\end{equation*}
Its conjugate momentum $\pi_0(f)$ is the self-adjoint operator
\begin{equation*}
  \pi_0(f)=\tfrac{1}{i \sqrt{2}} \bigl(a^{\dagger}(f)-a(f)\bigr)\;.
\end{equation*}
For regular $g\in\fh$, we obtain \emph{interacting} fields and conjugate
momenta: we define the self-adjoint operators $\phi_g(f)$ and $\pi_g(f)$ on $\mathcal{F}_g$
as
\begin{gather*}
  \Braket{\Psi,\phi_g(f)\Phi}_g = \frac{\langle \e^{\ad(g)}\Psi  , \phi_0(f)\, \e^{\ad(g)}\phi \rangle_{\mathcal{F}(\mathfrak{h})}}{\Braket{\e^{\ad(g)}\Omega,\e^{\ad(g)}\Omega}_{\mathcal{F}(\mathfrak{h})}}= \langle \Psi  , U_{0,g}\bigl(\phi_0(f)+ \sqrt{2}\Re \langle g  , f \rangle_{\mathfrak{h}}\bigr)U_{g,0}\Phi  \rangle_g\;,\\
  \Braket{\Psi,\pi_g(f)\Phi}_g = \frac{\langle \e^{\ad(g)}\Psi  , \pi_0(f)\, \e^{\ad(g)}\phi \rangle_{\mathcal{F}(\mathfrak{h})}}{\Braket{\e^{\ad(g)}\Omega,\e^{\ad(g)}\Omega}_{\mathcal{F}(\mathfrak{h})}}= \langle \Psi  , U_{0,g}\bigl(\pi_0(f)+ \sqrt{2}\Im \langle g  , f \rangle_{\mathfrak{h}}\bigr)U_{g,0}\Phi  \rangle_g\;.
\end{gather*}
By introducing an $\mathfrak{h}$-mollification $g_{\alpha}$ of $g\in \mathscr{T}'$ and taking limits as
in \cref{prop:2}, we see that $\mathcal{F}_g$ carries a representation of the
interacting fields for all $g\in \mathscr{T}'$, given by
\begin{gather*}
  \phi_g(f)= U_{0,g}\Bigl(\phi_0(f)+ \sqrt{2}\Re \langle g  , f \rangle\Bigr)U_{g,0}\;,\\
  \pi_g(f)=U_{0,g}\Bigl(\pi_0(f)+ \sqrt{2}\Im \langle g  , f \rangle\Bigr)U_{g,0}\;.
\end{gather*}
The equivalence of ``translated regular representations'' of the CCR on the
bare space $\mathcal{F}_0$ has been systematically studied in \citep[][Thm.\
8.12]{arai2020mps}, leading to the following result.

\begin{prop}
  \label[proposition]{prop:4}
  Let $g,g'\in \mathscr{T}'$; then the interacting fields and conjugate momenta on $\mathcal{F}_g$
  and $\mathcal{F}_{g'}$ are related by:
  \begin{gather*}
    \phi_g(f)= U_{g',g}\Bigl(\phi_{g'}(f)+ \sqrt{2}\Re \langle g-g'  , f \rangle\Bigr)U_{g,g'}\;,\\
    \pi_g(f)=U_{g',g}\Bigl(\pi_{g'}(f)+ \sqrt{2}\Im \langle g-g'  , f \rangle\Bigr)U_{g,g'}\;.
  \end{gather*}

  It follows that the CCR representations on $\mathcal{F}_g$ and $\mathcal{F}_{g'}$ are unitarily
  equivalent if and only if $g-g'\in \mathfrak{h}$ (with $\mathfrak{h}$ seen as a subspace of
  $\mathscr{T}'$). In particular, the representation of interacting fields carried by
  $\mathcal{F}_g$ is disjoint from the Fock representation of free fields in $\mathcal{F}_0$
  whenever $g\in \mathscr{T}'\smallsetminus \mathfrak{h}$.
\end{prop}
\begin{rem*}[Physical interpretation of the dressing map]
  \label{rem:1}
  In view of the above, we provide the following interpretation: the dressing
  $\e^{a(g)}$ always maps the physical renormalized space $\mathcal{F}_g$ to the
  standard Fock space $\mathcal{F}_0$ unitarily. This is mathematically very
  convenient, since, as we will see below, it allows to write physical
  observables (belonging to $\mathcal{F}_g$) as the unitary conjugation of their
  dressed counterparts, that will usually be rather explicit. On the other
  hand, the dressing map \emph{does not preserve} the canonical commutation
  relations, whenever $g\in \mathscr{T}'\smallsetminus \mathfrak{h}$ is singular; in other words, for singular
  sources \emph{the interacting field theory is disjoint from the free
    theory}, even though its dressed observables can be written in a Fock
  space. Although this toy model is not fully relativistic and Haag's theorem
  does not hold, it is remarkable that we can see nonetheless an expression
  of the fact that free and interacting vacua are disjoint; and this will
  also be the case for the spin-boson and Nelson models studied below.
\end{rem*}

Let us conclude this section by addressing the arbitrariness in the choice of
test functions $\mathscr{T}$ in the single-boson Gel'fand triple construction. Let us
suppose that we have two topological vector spaces $\mathscr{T}_1,\mathscr{T}_2\subset \mathfrak{h}$ continuously
embedded in the same Hilbert space, and such that $\mathscr{T}_1\cap \mathscr{T}_2$ is dense in
$\mathfrak{h}$. We can then identify the renormalized Fock spaces
corresponding to $g_1\in \mathscr{T}_1'$ and $g_2\in \mathscr{T}_2'$,
whenever $g_1$ and $g_2$ agree on $\mathscr{T}_1\cap \mathscr{T}_2$.
\begin{lem}
  \label[lemma]{lemma:5}
  Let $g_1\in \mathscr{T}_1'$ and $g_2\in \mathscr{T}_2'$ be such that for
  all $f\in \mathscr{T}_1\cap \mathscr{T}_2$, $g_1(f)=g_2(f)$. Then, there
  exists a natural unitary map $\Omega_{g_1,g_2}: \mathcal{F}_{g_1}\to
  \mathcal{F}_{g_2}$.
\end{lem}
\begin{proof}
  Since $\mathscr{T}_1\cap \mathscr{T}_2$ is dense, it follows that $\FSfin(\mathscr{T}_1\cap \mathscr{T}_2)$ is dense in
  $\mathcal{F}_{g_1}$ and $\mathcal{F}_{g_2}$ for any $g_1\in \mathscr{T}_1'$ and $g_2\in \mathscr{T}_2'$. Furthermore,
  since $g_1$ and $g_2$ agree on $\mathscr{T}_1\cap \mathscr{T}_2$, we have that for all $\Psi,\Phi\in
  \FSfin(\mathscr{T}_1\cap \mathscr{T}_2)$,
  \begin{equation*}
    \langle \Psi  , \Phi \rangle_{g_1}=\langle \Psi  , \Phi \rangle_{g_2}\;.
  \end{equation*}
  The identity map $\omega_{g_1,g_2}: \FSfin(\mathscr{T}_1\cap
  \mathscr{T}_2)_{g_1}\to \FSfin(\mathscr{T}_1\cap \mathscr{T}_2)_{g_2}$ on
  $\FSfin(\mathscr{T}_1\cap \mathscr{T}_2)$ (seen respectively as a subspace
  of $\mathcal{F}_{g_1}$ and $\mathcal{F}_{g_2}$) thus extends uniquely to
  the unitary map $\Omega_{g_1,g_2}$.
\end{proof}

\subsection{Wave function renormalization II: renormalized Hamiltonian}
\label{sec:wave-funct-renorm}

Thanks to the singular dressing map, we can diagonalize and define in the
renormalized Hilbert space the vHM model with any singular source $g\in
\mathscr{T}'$; this also has  the advantage of taking care at once of both UV
and IR singularities: the renormalized Hamiltonian always exhibits a unique ground
state in the renormalized Hilbert space.

Beforehand, let us quickly recall the notion of \emph{generalized norm
  resolvent convergence} originally due to Weidmann
\citep[][\textsection9.3]{weidmann2000Ide} (see also
\citep{bogli2017ieot,bogli2018jst}) and then refined by Post \citep[see][and
references therein]{post2022jst}, that is perfectly suited to study
convergence of the vHM Hamiltonian as the cutoffs are removed.

\subsection*{Generalized norm resolvent convergence}
\label{sec:textbfg-norm-resolv}

Norm and strong resolvent convergence of operators acting on the same Hilbert
space are very natural notions introduced to extend, to the case of unbounded
but self-adjoint operators, the definition given for bounded operators of
uniform and strong operator topology \citep[see,
\emph{e.g.},][\textsection{VIII.7}]{reed1972I}. In the wave function
renormalization procedure, however, we must study the convergence of a
sequence of self-adjoint operators each defined on a different Hilbert space
(the renormalized spaces). Nonetheless, thanks to the dressing maps $U_{g,0}:
\mathcal{F}_g\to \mathcal{F}_0$, all operators can be embedded on the same
Hilbert space. The generalized norm resolvent convergence addresses exactly
this situation.

\begin{defn}[Generalized norm/strong resolvent convergence]
  \label[definition]{def:4}
  Let $\{\mathscr{H}_\alpha:\alpha\in \mathcal{A}\}$ be a net of Hilbert
  spaces, as well as $\mathscr{H}_{\infty}$. We say that a family of
  self-adjoint operators $T_\alpha$ on $\mathscr{H}_\alpha$ converges to $T_{\infty}$
  on $\mathscr{H}_{\infty}$ in \emph{generalized norm(strong) resolvent
    sense} if and only if:
  \begin{itemize}
  \item There exists a parent Hilbert space $\mathscr{H}$ and isometries
    $\imath_\alpha: \mathscr{H}_\alpha\to \mathscr{H}$ for any $\alpha\in \mathcal{A}\cup
    \{\infty\}$;
  \item for any $z\in \mathbb{C}$, $\mathsf{Im}z\neq0$,
    \begin{equation*}
      \imath_\alpha (T_\alpha+z)^{-1} \imath^{*}_\alpha \xrightarrow[]{\alpha} \imath_{\infty}(T_{\infty}+z)^{-1}\imath^{*}_{\infty}\;,
    \end{equation*}
    where the convergence is intended in the norm(strong) operator topology
    on $\mathscr{H}$.
  \end{itemize}
\end{defn}
The spectral properties of generalized norm and strong resolvent convergence
are studied in detail in the aforementioned papers by Bögli
\citep{bogli2017ieot,bogli2018jst}. Let us also mention that if the maps
$\imath_\alpha$ are unitary, then the definition above has a more operative
version.
\begin{lem}
  \label[lemma]{lemma:7}
  Let the maps $\imath_\alpha$ in \cref{def:4} above be unitary. Then $T_\alpha$
  converges to $T_{\infty}$ in generalized norm/strong resolvent sense if and
  only if $\hat{T}_\alpha$ converges to $\hat{T}_{\infty}$ in norm/strong
  resolvent sense on $\mathscr{H}$, where for any $\alpha\in \mathcal{A}\cup
  \{\infty\}$,
  \begin{equation*}
    \hat{T}_\alpha= \imath_\alpha T_\alpha \imath_\alpha^{*}
  \end{equation*}
  is the unitary conjugation of $T_\alpha$.
\end{lem}
\begin{cor}
  \label[corollary]{cor:1}
  Suppose that $\hat{T}_\alpha\equiv \hat{T}_\infty$ on $\mathscr{H}$ for any $\alpha\in
  \mathcal{A}\cup \{\infty\}$; then $T_\alpha \to T_{\infty}$ in generalized norm
  resolvent sense.
\end{cor}

\subsection*{Renormalizing the vHM Hamiltonian}
\label{sec:textbfr-vhm-hamilt}

Let us now consider a field with dispersion relation $\omega$, and take a
Gel'fand triple $(\mathscr{T},\mathfrak{h},\mathscr{T}')$ for such field so
that $\omega: \mathscr{T}\to \mathscr{T}$ is a linear homeomorphism that
extends to a positive operator $\omega\geq 0$ on $\mathfrak{h}$, with
$\sqrt{\omega}\bigr\rvert_{\mathscr{T}}$ also an homeomorphism. Let us then
fix the vHM source $v\in \mathscr{T}'$ and an $\fh$-mollifying net $v_\alpha$ such that $v_\alpha/\sqrt\omega,v_\alpha/\omega\in\fh$ for all $\alpha$. Denote by $E(v_\alpha)\coloneqq  -\lVert
v_\alpha/\sqrt{\omega} \rVert_{\mathfrak{h}}^2$ the ground state energy of the vHM
form, and by $g=-v/\omega \in \mathscr{T}'$ and $g_\alpha=-v_\alpha/\omega\in
\fh$ the singular and mollified one-particle vHM ground
state configurations, respectively.

The vHM quadratic form $q_v$ is renormalized as follows. Define
\begin{equation*}
  H_{\mathsf{ren}}(v_\alpha) \coloneqq H(v_\alpha) - E(v_\alpha)\;;
\end{equation*}
and its associated quadratic form on the renormalized Hilbert space
$\mathcal{F}_{g_\alpha}$ as
\begin{equation*}
  q^{\mathsf{ren}}_{v_\alpha}(\Psi,\Phi)\coloneqq  \frac{\langle \e^{a^{\dagger}(g_\alpha)}\Psi  , H_{\mathsf{ren}}(v_\alpha)\e^{a^{\dagger}(g_\alpha)}\Phi \rangle_{\mathcal{F}(\mathfrak{h})}}{\langle \e^{a^{\dagger}(g_\alpha)}\Omega  , \e^{a^{\dagger}(g_\alpha)}\Omega \rangle_{\mathcal{F}(\mathfrak{h})}}\;,
\end{equation*}
densely defined on $\FSfin(\mathscr{T})_{g_\alpha}$ (seen as a subspace of
$\mathcal{F}_{g_\alpha}$). Let us remark that for regularized sources the action
of $H_{\mathsf{ren}}(v_\alpha)$ on the renormalized Hilbert space amounts to a
``rotation'' of its action in the Fock representation: it is the same as
applying the unitary dressing transformation.
\begin{lem}
  \label[lemma]{lemma:8}
  For any $\Psi,\Phi\in \FSfin(\mathscr{T})_{g_\alpha}$, we have that
  \begin{equation*}
    q^{\mathsf{ren}}_{v_\alpha}(\Psi,\Phi)= \langle \e^{a(g_\alpha)}\Psi  , \mathsf{d}\Gamma(\omega) \e^{a(g_\alpha)}\Phi  \rangle_{\mathcal{F}(\mathfrak{h})}\;,
  \end{equation*}
  and therefore it is the quadratic form associated to the self-adjoint
  operator $H_{\mathsf{ren}}(v_\alpha)_{g_\alpha}$ on $\mathcal{F}_{g_\alpha}$ defined by
  the unitary CCR-equivalent conjugation
  \begin{equation*}
    H_{\mathsf{ren}}(v_\alpha)_{g_\alpha} \coloneqq  U_{0,g_\alpha} \mathsf{d}\Gamma(\omega) U_{g_\alpha,0}\;.
  \end{equation*}
\end{lem}
Remarkably, the same holds true also in the limit $\alpha\in \mathcal{A}$ in which the
cutoff is removed, even though now the unitary conjugation does not preserve
the canonical commutation relations: the wave function renormalization
defines (and, at the same time diagonalizes through the singular dressing)
the renormalized, non-Fock, vHM Hamiltonian to which the regularized one
converges in generalized norm resolvent sense.
\begin{thm}[Renormalization of the vHM model]
  \label[theorem]{thm:3}
  For any $v\in \mathscr{T}'$, the renormalized vHM Hamiltonian
  $H_{\mathsf{ren}}(v)_g$ has the following properties:

  \vspace{1mm}
  
  \begin{itemize}
    \setlength{\itemsep}{2mm}
  \item For any net $\{v_\alpha\}\subset\sT'$ such that $v_\alpha\xrightarrow[\alpha]{\sigma(\sT',\sT)}v$, the corresponding Hamiltonians $H_{\mathsf{ren}}(v_\alpha)_{g_\alpha}$ converge to
    $H_{\mathsf{ren}}(v)_g$ in generalized norm resolvent sense.
  \item $H_{\mathsf{ren}}(v)_g\coloneqq  U_{0,g} \mathsf{d}\Gamma(\omega)U_{g,0}$
    with $g=-v/\omega$, and therefore:
    \begin{itemize}
      \setlength{\itemsep}{1mm}
    \item $H_{\mathsf{ren}}(v)_g$ is self-adjoint on $\mathcal{F}_g$ with
      domain $\mathscr{D}\bigl(H_{\mathsf{ren}}(v)_g\bigr)=
      U_{0,g}\mathscr{D}\bigl(\mathsf{d}\Gamma(\omega)\bigr)$;
    \item $\inf \sigma\bigl(H_{\mathsf{ren}}(v)_g\bigr)=0$;
    \item $U_{g,0}\Omega$ is its unique ground state.
    \end{itemize}
  \end{itemize}
\end{thm}
\begin{proof}
  By \cref{lemma:8,prop:2} it follows that for any $\Psi,\Phi\in
  \FSfin(\mathscr{T})$,
  \begin{equation*}
    \lim_{\alpha}q^{\mathsf{ren}}_{v_\alpha}(\Psi,\Phi)= q^{\mathsf{ren}}_v(\Psi,\Phi)= \langle \Psi  , U_{0,g}\mathsf{d}\Gamma(\omega)U_{g,0}\Phi \rangle_g\;.
  \end{equation*}
  The convergence is extended to generalized norm resolvent sense by
  \cref{cor:1}, observing that uniformly with respect to $\alpha$,
  \begin{equation*}
    U_{g_\alpha,0}H_{\mathsf{ren}}(v_\alpha)_{g_\alpha}U_{0,g_\alpha} = \mathsf{d}\Gamma(\omega)\;.
  \end{equation*}
  The remaining properties are a straightforward consequence of the fact that
  $H_{\mathsf{ren}}(v)_g= U_{0,g}\mathsf{d}\Gamma(\omega)U_{g,0}$ is the
  unitary conjugation of $\mathsf{d}\Gamma(\omega)$.
\end{proof}

\section{Wave Function Renormalization of the Spin-Boson Model}
\label{sec:wave-funct-renorm-1}

\subsection*{The spin-boson in physics}
\label{sec:spin-boson-physics}

The spin-boson model builds on the vHM model by giving the unmovable source
some internal degrees of freedom, that interact with the scalar field. This
apparently simple modification gives rise to a much richer physical
structure, as well as a much wider range of applicability: the spin-boson
model is considered ``ubiquitous'' in physics. In particular, it is used
either to describe ``intrinsic'' two-level systems (such as polarization in
photons, spin $\frac{1}{2}$ projection in atoms, strangeness of a neutral
meson), effective two-level systems (such as in some types of chemical
reactions, amorphous materials, or the motion of defects in crystals) or
interactions of two-level systems with a surrounding bath (as an explicit
model for an open quantum system) \citep[see, \emph{e.g.},][and references
therein]{Leggettetal.1987,jaksic1995aihppt,jaksic1996cmp,jaksic1996cmp2,jaksic2002cmp,frohlich2004cmp251,merkli2007prl}.

In the aforementioned classical review article \citep{Leggettetal.1987}, the
emergence of form factors $v$ (there called spectral functions $J(\omega)$),
and their regularity -- or lack thereof -- is discussed in great detail, and
some attention is already devoted to the complicated physical behavior that
arises when renormalization is needed. Whereas imposing an ultraviolet cutoff
seems to be physically reasonable -- and it is the choice ultimately made in
many concrete applications -- neglecting high frequency photons makes the
theory inconsistent at large coupling. In recent years, both the
investigation of large coupling regimes and of the role of the ultraviolet
cutoff have garnered considerable interest in the physical community, both
from a theoretical \citep[see, \emph{e.g.},][and references
therein]{OtterpohlNalbachThorwart.2022,YanShao.2023,goulko2025prl,guo2025pla}
and an experimental point of view \citep[see,
\emph{e.g.},][]{magazzu2018nc,sun2025nc}.

The most striking physical application of a spin-boson Hamiltonian with
strong singularities -- and whose renormalization is a long-standing open
problem we solve hereafter -- is in the theory of spontaneous emission:
almost a century ago, Weisskopf and Wigner \citep{WeisskopfWigner.1930}
proposed a simple theory of spontaneous emission within the framework of
quantum electrodynamics; despite its simplicity, such model is still
considered a golden standard in describing spontaneous emission of atoms and
molecules \cite[see, \emph{e.g.}][and references
therein]{Sharafievetal.2021}. Coming from quantum electrodynamics, the field
form factor has both IR and UV divergences, that can be cured with wave
function renormalization: most notably, the simplest renormalized spin-boson
model for spontaneous emission can be \emph{diagonalized exactly} (see the
discussion below).

To conclude this brief overview, it is now establshed that the spin-boson
model possesses an inherent duality to the Ising model from statistical
mechanics \cite{EmeryLuther.1974,SpohnDuemcke.1985,FannesNachtergaele.1988}.
In the presence of an ultraviolet cutoff, the links between the infrared
behavior of the spin-boson model and long-range order in the corresponding
Ising model have been studied; in particular, the existence or absence of an
Ising phase transition is established, in dependence of the singularity of
the form factor \cite{Spohn.1989,BetzHinrichsKraftPolzer.2025,HinrichsPolzer.2025}.
However,
all of these works crucially rely on the presence of an ultraviolet cutoff in
the interaction. The renormalization of the spin-boson model can thus, in
this perspective, be considered as a problem dual to the renormalization
problem for critical Ising models.

\subsection*{Renormalization of the spin-boson: an overview of known results}
\label{sec:renorm-spin-boson}

The renormalization of the spin-boson model has garnered a lot of attention
in recent years; in the existing literature the problem of renormalization
draws heavy inspiration from the Nelson model, and the techniques used
insofar are essentially limited to a unitary dressing transformation of Gross
type, allowing to cure energy but not wave function divergences. Using the
same notation as in the vHM case, three regimes of source regularity have been
studied for the UV regime:
\begin{itemize}
\item \emph{Subcritical case}: $\boxed{v/\sqrt{1+\omega}\in \mathfrak{h}}$. The
  spin-boson Hamiltonian is self-adjoint with form domain coinciding with the
  free form domain if $v\notin \mathfrak{h}$, and domain $D(H_0)$ if $v\in
  \mathfrak{h}$ (by KLMN and Kato-Rellich's theorems, respectively
  \citep[see, \emph{e.g.}][]{ReedSimon.1975}; or see \cite{Lonigro.2022} for
  an operator based technique); $H$ is also \emph{bounded from below}.
\item \emph{Critical case}: $\boxed{v/(1+\omega)\in L^2,v/\sqrt{1+\omega}\notin L^2}$.
	The Hamiltonian is a self-adjoint and bounded from below operator
  only after a suitable \emph{self-energy renormalization}, implemented by a
  \emph{unitary dressing transformation}
  \citep{AlvarezLillLonigroValentinMartin.2025,HinrichsLampartValentinMartin.2024},
  and removing the cutoff yields norm resolvent convergence.
\item \emph{Supercritical case}: $\boxed{v,v/\sqrt{1+\omega},v/(1+\omega)\notin\mathfrak{h}}$.
	It was proved in \citep[][Corollary 4.4]{DamMoller.2018b}
  that a self energy renormalization \emph{cannot converge} in strong
  resolvent sense (see also
  \citep{AlvarezLillLonigroValentinMartin.2025,HinrichsLampartValentinMartin.2024}).
\end{itemize}
Similar regimes of source regularity give rise to analogous problems in the infrared sector:
\begin{itemize}
	\item \emph{Subcritical case:} $\boxed{v,v/\sqrt{\omega},v/\omega\in\mathfrak{h}}$.
	The spin-boson Hamiltonian is self-adjoint on $D(H_0)$, bounded from below and possesses a unique ground state \cite{AraiHirokawa.1997}.
	\item \emph{Critical case}: $\boxed{v,v/\sqrt\omega\in	\mathfrak{h},v/\omega\notin \mathfrak{h}}$.
		The spin-boson Hamiltonian $H$ is self-adjoint on $D(H_0)$ and bounded from below, but
		the model only exhibits a ground state at weak coupling \cite{Spohn.1989,HaslerHerbst.2010,BetzHinrichsKraftPolzer.2025}.
		In fact, any symmetry-breaking term, will cause absence of the ground state immediately.
	\item \emph{Supercritical case}: $\boxed{v\in \mathfrak{h},v/\sqrt{\omega}\notin \mathfrak{h}}$.
	The spin-boson Hamiltonian $H$ is essentially self-adjoint on $D(H_0)\cap
	D_{\mathrm{fin}}$, with $H_0= A\otimes \mathds{1} + \mathds{1}\otimes
	\dG(\varpi)$ and $D_{\mathrm{fin}}$ the set of vectors with a finite number
	of Fock excitations \citep{Falconi.2015}. However, $H$ is \emph{unbounded
		from below}. This instability is infrared in nature, and typical of
	massless fields.
\end{itemize}

\subsection{The spin-boson model in the Fock representation}
\label{sec:spin-boson-model}

Let us use the same setting and notations of \cref{sec:van-hove-model},
concerning the ``boson subsystem'': most notably, let us consider $(\mathscr{T}, \mathfrak{h},
\mathscr{T}')$ to be the test/one-particle/field configurations Gel'fand triple, $\omega$
the field's dispersion, $\mathcal{F}(\mathfrak{h})$ the Fock space and $a^{\sharp}(f)$,
$\mathsf{d}\Gamma(T)$ the creation, annihilation and second quantization
operators. The ``spin subsystem'' is set in a Hilbert space $\mathcal{S}$, that is
typically finite dimensional, but we here do not require to have any particular properties.
The Fock representation is thus set on the Hilbert space
\begin{equation*}
  \mathcal{B}_0 \coloneqq  \mathcal{S}\otimes \mathcal{F}(\mathfrak{h})\;.
\end{equation*}
The free spin-boson Hamiltonian is
\begin{equation*}
  H_0= A\otimes \mathds{1} + \mathds{1}\otimes \mathsf{d}\Gamma(\omega)\;,
\end{equation*}
with $A$ self-adjoint and bounded on $\mathcal{S}$. Therefore, $H_0$ is always bounded
from below, with $\inf \sigma(H_0)= \inf \sigma(A)$; furthermore, given any ground
state $\psi\in \mathcal{S}$ of $A$, $\psi\otimes \Omega\in \mathcal{B}_0$ is the associated ground state of $H_0$: the
degeneracy of the ground state of $H_0$ coincides with the degeneracy of the
ground state of $A$. The spin-boson (SB) Hamiltonian is obtained adding a
linear interaction between the field and the spin degrees of freedom:
\begin{defn}[Fock SB Hamiltonian]
  \label[definition]{def:7}
  Let $\omega\geq 0$ on $\mathfrak{h}$, $A$ bounded and self-adjoint on $\mathcal{S}$, $v\in \mathfrak{h}$, and $B$
  normal and bounded on $\mathcal{S}$. Then, the Fock SB Hamiltonian $H(B,v)$ is
  defined as
  \begin{equation*}
    H(B,v)= H_0+ B\otimes a^{\dagger}(v)+B^{*}\otimes a(v)= A\otimes\mathds{1}+\mathds{1}\otimes \mathsf{d}\Gamma(\omega)+ B\otimes a^{\dagger}(v)+B^{*}\otimes a(v)\;. 
  \end{equation*}
\end{defn}
The operator $B$ takes the name of \emph{interaction spin-matrix}.
\begin{rem*}
  \label{rem:2}
  The condition on $B$ being normal is used heavily throughout the paper,
  however it excludes the so-called \emph{rotating wave approximation}, in
  which $B$ is 2-nilpotent. We are convinced that the wave function
  renormalization scheme given here can be adapted with no major
  complications to a nilpotent $B$, and probably to any generic bounded
  operator as well. We refrain from doing that here because it would add an
  unnecessary layer of technicality.
\end{rem*}
\begin{exs*}[Standard, Weisskopf-Wigner, Energy-preserving SB]
  $\phantom{i}$
  \begin{itemize}
    \setlength{\itemsep}{3mm}
  \item \emph{Standard SB}. The standard spin-boson is the most widespread
    from a physical perspective \citep[see,
    \emph{e.g.},][]{Leggettetal.1987}. In the standard spin-boson, $\mathcal{S}=\mathbb{C}^2$,
    $A=\sigma_z=
    \begin{pmatrix}
      1&0\\0&-1
    \end{pmatrix}
    $, $\mathfrak{h}=L^2(\mathbb{R}^3)$, $\omega(k)= \sqrt{k^2+\mu^2}$ with $\mu\geq 0$, $B=\sigma_x=
    \begin{pmatrix}
      0&1\\1&0
    \end{pmatrix}
    $, and $v(k)\in L^2(\mathbb{R}^3)$ with compact support (ultraviolet cutoff). Let us
    remark that while $H_0$ has in this case a non-degenerate ground state,
    $H(B,v)$ might have no ground state at all for $\mu=0$ (in the Fock
    representation).
  \item \emph{Weisskopf-Wigner SB}. The Weisskopf-Wigner SB is a standard SB
    in which, however, $v(k)= 1/\sqrt{\omega(k)}\notin L^2(\mathbb{R}^3)$. Therefore, it cannot
    be defined nor renormalized as a self-adjoint operator in the Fock
    representation: a non-trivial wave function renormalization will be
    necessary, \citep[see, {\em e.g.},][]{RomanRocheSanchezZueco}.
  \item \emph{Energy-preserving SB}. An energy preserving SB has $B=A$; thus
    it can be diagonalized in the spin degrees of freedom: on each spin fiber
    it acts as a Fock vHM Hamiltonian.
  \end{itemize}
\end{exs*}

The spectral properties of the SB Hamiltonian depend on both the regularity
of the form factor $v$ and the type of interaction spin matrix $B$: their
interplay generates a more intricate behavior, compared to the vHM case. In
particular, the existence of a Fock ground state extends beyond the vHM
threshold, but only for suitable interaction spin matrices and coupling
strengths. Furthermore, contrarily to the vHM case, the SB Hamiltonian is not
exactly diagonalizable.

\begin{prop}
  \label[proposition]{prop:8}
  Let $v\in \mathfrak{h}$, $B$ bounded and normal and $A$ bounded and self-adjoint on $\mathcal{S}$,
  and $\omega\geq 0$. Then:
  \begin{itemize}
    \setlength{\itemsep}{3mm}
  \item $H(B,v)$ is essentially self-adjoint on $\mathcal{S}\hat{\otimes}
    \FSfin\bigl(\mathscr{D}(\omega)\bigr)$ and $ \inf \sigma\bigl(H(v)\bigr) = - \infty $, whenever
    $v/\sqrt{\omega}\notin \mathfrak{h}$;
  \item If $v/\sqrt{\omega}\in \mathfrak{h}$, then $H(B,v)$ is self-adjoint on $\mathcal{S}\otimes \mathscr{D}(\dG(\omega))$
    and bounded from below;
  \item If, in addition, $v/\omega\in \mathfrak{h}$, $\dim\cS<\infty$ and $\braket{\ph, A\psi}\le 0$ for
    eigenvectors $\ph\perp\psi$ of $B$, then $H(B,v)$ also has a Fock ground state.
  \item If $H(B,v)$ has a Fock ground state $\Psi_{\mathsf{gs}}$ such that $\langle
    \Psi_{\mathsf{gs}} , B\otimes \mathds{1} \Psi_{\mathsf{gs}} \rangle_{\mathcal{B}_0}\neq0$, then $v/\omega\in \mathfrak{h}$.
  \end{itemize}
\end{prop}
\begin{proof}
  The first claim is proven similar to \cref{prop:5}, by Nelson's commutator
  theorem and inserting trial states of the form $x\otimes c_{g}$.  The second
  claim is a direct consequence of the Kato--Rellich theorem
  \citep[Thm.~X.12]{ReedSimon.1975} and a standard relative bound
  \citep[Thm.~5.16]{Arai.2018}.  The existence of a Fock ground state in the
  case that $v/\omega\in\fh$ has been first observed in \cite{AraiHirokawa.1997} for
  small coupling $\|v/\omega\|_\fh$ and later proven non-perturbatively under
  additional regularity assumptions and with different methods \citep[see,
  \emph{e.g.},][]{Gerard.2000,HirokawaHiroshimaLorinczi.2014,DamMoller.2018a,Hinrichs.2022b}.
  The precise statement given here is proven in
  \citep[Cor.~5.5]{HinrichsPolzer.2025} for the case that $B$ is self-adjoint, but the proof carries over {\em ad verbum} to normal $B$.
	
  It thus remains to consider the case $v/\omega\notin\fh$. Absence of ground states
  with $\langle \Psi_{\mathsf{gs}} , B\otimes \mathds{1} \Psi_{\mathsf{gs}} \rangle_{\mathcal{B}_0}\neq0$ and
  $\|(\dsone\otimes \dG(\dsone))^{1/2}\Psi_{\mathsf{gs}}\|<\infty$ was studied in
  \cite{AraiHirokawaHiroshima.1999}.  We here prove a mild generalization of
  their result. By the spectral theorem, we can assume without loss of
  generality that $\fh=L^2(\sM,\mu)$ for some measure space $(\sM,\mu)$ and that
  $\omega$ is a multiplication operator on this space. Then, for almost every
  $k\in\sM$, the so-called pull-through formula
  \begin{align*}
    a_k\Psi_{\mathsf{gs}} = - \frac{v(k)}{H(B,v)-\inf\sigma(H(B,v)) + \omega(k)}B\otimes\Id\Psi_{\mathsf{gs}}, \quad \mbox{$\mu$-a.e. $k\in\sM$}
  \end{align*}
  holds, \citep[see e.g.][Thm.~D.16]{DamMoller.2018a}, where $a_k$ denotes
  the pointwise annihilation operator on $\FS(L^2(\sM,\mu))$.  Given $\eps>0$,
  we now pick $\eta_\eps\in\sD(\dG(\Id)^{1/2})$ such that
  $\|\Psi_{\mathsf{gs}}-\eta_\eps\|<\eps$, which is possible by
  density. Then, we can estimate
  \begin{align*}
    \big|\!\Braket{\eta_\eps,a_k\Psi_{\mathsf{gs}}}\!\big|
    &\ge \big|\!\Braket{\Psi_{\mathsf{gs}},a_k\Psi_{\mathsf{gs}}}\!\big| - \big|\Braket{\eta_\eps-\Psi_{\mathsf{gs}},a_k\Psi_{\mathsf{gs}}}\big|
    \\& \ge \frac{|v(k)|}{\omega(k)}\big|\!\Braket{\Psi_{\mathsf{gs}},(B\otimes\Id)\Psi_{\mathsf{gs}}}\!\big| - \frac{|v(k)|}{\omega(k)}\|\eta_{\eps}-\Psi_{\mathsf{gs}}\|\|(B\otimes \Id)\Psi_{\mathsf{gs}}\|
    \\&\ge \Big(\big|\!\Braket{\Psi_{\mathsf{gs}},(B\otimes\Id)\Psi_{\mathsf{gs}}}\!\big|-\eps\|B\|\Big)\frac{|v(k)|}{\omega(k)}.
  \end{align*}
  On the opposite, we have
  \begin{align*}
    \int_{\sM}\big|\!\Braket{\eta_\eps,a_k\Psi_{\mathsf{gs}}}\!\big|^2 \sfd k \le \|(\dG(\Id)+1)^{1/2}\eta_\eps\|^2\|\Psi_{\mathsf{gs}}\|^2.
  \end{align*}
  Thus the function $k\mapsto
  \Big(\big|\!\Braket{\Psi_{\mathsf{gs}},(B\otimes\Id)\Psi_{\mathsf{gs}}}\!\big|-\eps\|B\|\Big)\frac{|v(k)|}{\omega(k)}$
  must be square-integrable for any $\eps>0$, which if
  $\Braket{\Psi_{\mathsf{gs}},(B\otimes\Id)\Psi_{\mathsf{gs}}}\ne 0$ is only possible if
  $v/\omega\in\fh$.
\end{proof}
\begin{rems*}
  $\phantom{i}$
  \begin{itemize}
  \item The existence result has also been shown in infinite dimensions for
    $A$ being an arbitrary operator with compact resolvent under additional
    regularity assumptions on $\omega$ and $v$ \citep[cf.][]{Gerard.2000}.  The
    compactness method from \citep{Hinrichs.2022b,HaslerHinrichsSiebert.2023}
    in principle also allows to remove the assumption of a compact resolvent.
    We assume that any restriction on $\dim\cS$, $A$, $B$ or $\omega$ is merely of
    technical nature, as long as $v/\omega\in\fh$.
  \item The condition $\langle \Psi_{\mathsf{gs}} , (B\otimes \mathds{1}) \Psi_{\mathsf{gs}} \rangle_{\mathcal{B}_0}\neq0$
    on the contrary is not technical, but actually of physical relevance: if
    $\langle \Psi_{\mathsf{gs}} , (B\otimes \mathds{1}) \Psi_{\mathsf{gs}} \rangle_{\mathcal{B}_0}=0$, then
    cancellations of infrared divergences can occur, which would otherwise
    lead to the absence of ground states as in \cref{prop:5}.  Note that $\langle
    \Psi_{\mathsf{gs}} , (B\otimes \mathds{1}) \Psi_{\mathsf{gs}} \rangle_{\mathcal{B}_0}=0$ follows whenever
    there exists a unitary $U$ on $\cS$, which commutes with $A$ and
    anti-commutes with $B$
    \citep[][Proof of Lemma~2.11]{BachBallesterosKoenenbergMenrath.2017}.
  \item Specifically, in the case of the standard spin-boson model $A=\sigma_z$ and
    $B=\sigma_x$, such a unitary is given by $U=\sigma_z$ and existence of a Fock
    ground state for $v/\omega\notin \mathfrak{h}$ and small coupling $\|v/\sqrt\omega\|$ has been proved
    in \citep{HaslerHerbst.2010,BachBallesterosKoenenbergMenrath.2017} using
    a perturbative method and in the article series
    \citep{HaslerHinrichsSiebert.2021a,HaslerHinrichsSiebert.2021b,HaslerHinrichsSiebert.2021c}
    using a non-perturbative approach, also see \cite{HinrichsPolzer.2025} for a simpler proof.  The absence of a Fock ground state
    for large coupling as indicated by algebraic results in
    \citep{Spohn.1989} has recently been proved in
    \citep{BetzHinrichsKraftPolzer.2025}.
  \item The interplay between $A$, $\omega$, $v$ and $B$ in the SB model is thus
    quite subtle: the massless SB Hamiltonian can still have a catastrophic
    behavior due to soft photons in the infrared, however the interaction
    with the spin degrees of freedom could be sufficient to take care of such
    divergence itself, at least if the coupling is sufficiently small.
  \end{itemize}
\end{rems*}

\subsection{The SB model with singular interactions and renormalization in
  the Fock representation}
\label{sec:sb-model-with}

Using the singular annihilation of \cref{def:5}, it is possible to define the
SB form on $\mathcal{S}\hat{\otimes} \FSfin(\mathscr{T})$ for any $v\in \mathscr{T}'$, analogously to the vHM
case. The types of divergences that might plague such form are also
analogous, and the Fock space energy renormalization for the SB model has
been carried out in
\citep{HinrichsLampartValentinMartin.2024,AlvarezLillLonigroValentinMartin.2025}. Let
us also briefly review the negative results of Dam and Møller, concerning the
Fock renormalization of singular sources \citep{DamMoller.2018b}.

\begin{defn}[SB form]
  \label[definition]{def:8}
  Let $(\mathscr{T},\mathfrak{h},\mathscr{T}')$ be the Gel'fand triple of the field, with the dispersion
  $\omega:\mathscr{T}\to \mathscr{T}$ being a linear homeomorphism that extends to a positive operator
  $\omega\geq 0$ on $\mathfrak{h}$, and $v\in \mathscr{T}'$. In addition, let $\mathcal{S}$ be the spin Hilbert space,
  with free spin matrix $A$ bounded and self-adjoint, and spin interaction
  matrix $B$ bounded and normal on $\mathcal{S}$. Then the \emph{SB sesquilinear form}
  $q_{B,v}: (\mathcal{S}\hat{\otimes} \FSfin(\mathscr{T}))\times(\mathcal{S}\hat{\otimes} \FSfin(\mathscr{T}))\to \mathbb{C}$ is defined by
  \begin{equation*}
    q_{B,v}(\Psi,\Phi) = \langle \Psi  , \bigl(A\otimes \mathds{1}\bigr)\Phi \rangle_{\mathcal{B}_0}+\langle \Psi  , \bigl(\mathds{1}\otimes \mathsf{d}\Gamma(\omega)\bigr)\Phi \rangle_{\mathcal{B}_0}+ \langle \bigl(B^{*}\otimes a(v)\bigr)\Psi  , \Phi \rangle_{\mathcal{B}_0}+ \langle \Psi  , \bigl(B^{*}\otimes a(v)\bigr)\Phi \rangle_{\mathcal{B}_0}\;.
  \end{equation*}
\end{defn}
It is clear that for all $v\in \mathfrak{h}$, the quadratic form $q_{B,v}(\Psi,\Phi)= \langle \Psi ,
H(B,v)\Phi \rangle_{\mathcal{B}_0}$.
\begin{prop}
  \label[proposition]{prop:9}
  Let $v/\sqrt{\omega}\in \mathfrak{h}$, and $\omega\geq 0$. Then
  \begin{equation*}
    H(B,v)= \bigl(A\otimes \mathds{1}+\mathds{1}\otimes
    \dG(\omega)\bigr)\overset{\cdot}{+} \bigl(B\otimes a^{\dagger}(v)+ B^{*}\otimes a(v)\bigr)
  \end{equation*}
  is self-adjoint on $\mathscr{D}\bigl(H(B,v)\bigr)\subseteq \mathscr{D}\bigl((1\otimes
  \mathsf{d}\Gamma(\omega))^{1/2}\bigr)$ and bounded from below.
\end{prop}
\begin{proof}
  Whenever $v/\sqrt{\omega}\in \mathfrak{h}$, the form $q_{B,v}$ is closed on the free form
  domain, and bounded from below by KLMN theorem \citep[][Thm.\
  X.17]{ReedSimon.1975}, since it is the sum of the free form $\langle \,\cdot \, , H_0\,\cdot
  \, \rangle_{}$ and an arbitrarily small relative perturbation.
\end{proof}
Explicit operator-theoretic constructions of $H(B,v)$ in this
form-perturbation case were discussed in
\citep{Lonigro.2022,LillLonigro.2024}.

Similarly to \cref{prop:7}, with an analogous proof that makes use of a
unitary dressing transformation, the SB model also allows for an energy
renormalization scheme, when the spin-field interaction ceases to be form-bounded
with respect to $H_0$.
\begin{prop}
  Let $v_\alpha\xrightarrow{\alpha} v$ be a $\fh$-mollification of $v\in\sT'$ such that $v_\alpha/\sqrt\omega\in\fh$ for all $\alpha$ and $\fh\ni v_\alpha/\omega\xrightarrow[\alpha]{\fh} v/\omega\in\fh$. Further, let $B$ be
  a bounded normal operator on $\cS$.  Then there exists a unique
  self-adjoint operator $H_{\mathsf{ren}}(B,v)$ on $\cS\otimes\FS(\fh)$ such that
  \[ H(B,v_\alpha) + B^*B\|v_\alpha/\sqrt\omega\|_{\mathfrak{h}}^2 \xrightarrow[\alpha]{\lVert \,\cdot \,
      \rVert_{}^{}\mathsf{-res}} H_{\mathsf{ren}}(B,v).\]
\end{prop}
\begin{proof}
  A similar construction of $H_{\mathsf{ren}}(B,v)$ to that in \cref{prop:7}
  is given in \citep[][Thm.~2.7]{HinrichsLampartValentinMartin.2024}. A
  different construction was also achieved in
  \citep{AlvarezLillLonigroValentinMartin.2025} and in the special case of
  the standard SB model already in \citep{DamMoller.2018b}.
\end{proof}
In the super-critical case $v/\omega\notin\fh$, the existence of an energy-renormalized
self-adjoint and lower-bounded Hamiltonian on Fock space can be excluded
\citep{DamMoller.2018b,HinrichsLampartValentinMartin.2024}.
\begin{prop}
  If $\inf\sigma(\omega)>0$ and $v/\omega\notin\fh$, then for any $\sT$-mollification $v_\alpha$ of
  $v$ and any bounded normal operator $B$ on $\cS$, the sequence
  $H(B,v_\alpha)+B^*B\|v_\alpha/\sqrt\omega\|_\fh^2$ can not converge in strong resolvent
  sense.
\end{prop}
\begin{proof}
  A slightly more general version of this result is proven in
  \citep[][Thm.~2.8]{HinrichsLampartValentinMartin.2024}.
\end{proof}
For so-called permuting interactions, \emph{i.e.}, if for any eigenvector $\psi$
of $A$, then $B\psi\perp\psi$ is also an eigenvector of $A$, this result can be traced
back to the observation that the appropriately dressing-transformed model
converges to a trivial one without any spin-field interaction
\citep{DamMoller.2018b,AlvarezLillLonigroValentinMartin.2025}.

\subsection{SB wave function renormalization I: renormalized Hilbert space}
\label{sec:wave-funct-renorm-3}
The wave function renormalization of the SB model follows closely the one for the vHM
model, taking into account the additional degrees of freedom in a natural
fashion. To that extent, let us consider $\mathcal{S}\hat{\otimes}\FSfin(\mathscr{T})$, and define the
\emph{SB singular dressing} $\e^{B^{*}\otimes a(g)}:\mathcal{S}\hat{\otimes}\FSfin(\mathscr{T})\to
\mathcal{S}\hat{\otimes}\FSfin(\mathscr{T})$ by its action on pure tensors: given $\varphi\in \mathcal{S}$ and $\Xi\in
\FSfin(\mathscr{T})$, then
\begin{equation*}
  \e^{B^{*}\otimes a(g)}\bigl(\varphi\otimes \Xi\bigr)= \sum_{k\in \mathbb{N}}^{}\frac{1}{k!}\bigl((B^{*})^k\varphi\bigr)\otimes \bigl((a(g))^k\Xi\bigr)\;.
\end{equation*}
From this decomposition and \cref{prop:1}, the map $\e^{B^{*}\otimes a(g)}$ is also
injective.
\begin{prop}
  \label[proposition]{prop:10}
  For any $B$ normal and bounded on $\mathcal{S}$, and $g\in \mathscr{T}'$, the map $\e^{B^{*}\otimes
    a(g)}$ is injective on $\mathcal{S}\hat{\otimes}\FSfin(\mathscr{T})$. Therefore,
  \begin{equation*}
    \langle \Psi  , \Phi \rangle_{B,g}\coloneqq  \bigl\langle \e^{B^{*}\otimes a(g)}\Psi  , \e^{B^{*}\otimes a(g)}\Phi \rangle_{\mathcal{B}_0}
  \end{equation*}
  is a non-degenerate inner product on $\mathcal{S}\hat{\otimes}\FSfin(\mathscr{T})$.
\end{prop}
Thus, similar to \cref{def:2}, we can define:
\begin{defn}[Renormalized Spin-Boson Space]
  \label[definition]{def:3}
  For any $B$ normal and bounded on $\mathcal{S}$ and $g\in \mathscr{T}'$, the \emph{renormalized
    SB space} $\mathcal{B}_{B,g}$ is defined as
  \begin{equation*}
    \mathcal{B}_{B,g}= \overline{\mathcal{S}\hat{\otimes} \FSfin(\mathscr{T})}^{\lVert\, \cdot \,  \rVert_{B,g}^{}}\;.
  \end{equation*}
\end{defn}
As for the vHM model, the inner product $\langle \cdot , \cdot \rangle_{B,g}$ can be seen as the
outcome of a wave function renormalization procedure, in which, due to the
presence of the spin degrees of freedom, the divergent Fock vacuum partial
expectation is operator-valued in $\mathcal{S}$.
\begin{prop}
  \label[proposition]{prop:11}
  Let $B$ be bounded and normal on $\mathcal{S}$, let $g\in \mathscr{T}'$ and let $g_\alpha$ be a $\fh$-mollification of $g$. Then, for all $\Psi,\Phi\in
  \mathcal{S}\hat{\otimes}\FSfin(\mathscr{T})$,
  \begin{equation*}
    \langle \Psi  , \Phi \rangle_{B,g}= \lim_{\alpha} \bigl\langle \e^{-\frac{1}{2}B^{*}B \lVert g_\alpha  \rVert_{\mathfrak{h}}^2}\,\e^{B\otimes a^{\dagger}(g_\alpha)}\Psi  , \,\e^{-\frac{1}{2}B^{*}B \lVert g_\alpha  \rVert_{\mathfrak{h}}^2}\, \e^{B\otimes a^{\dagger}(g_\alpha)}\Phi \bigr\rangle_{\mathcal{B}_0}\;.
  \end{equation*}
\end{prop}

\subsection*{The SB dressing maps}
\label{sec:sb-dressing-maps}

Let us now focus on the SB singular dressings $\e^{B^{*}\otimes a(g)}$ as maps
between renormalized SB spaces. We can densely define, for a fixed normal $B\in
\mathcal{B}(\mathcal{S})$, the maps between SB spaces with different $g$s as follows: for any
$g,g'\in \mathscr{T}'$, define the isometries $U_{B,g,g'}: \mathcal{B}_{B,g}\to \mathcal{B}_{B,g'}$ as the
unique extension of
\begin{equation}
  \label{eq:ugB}
  u_{B,g,g'}\Psi= \e^{B^{*}\otimes a(g-g')}\Psi\;,
\end{equation}
with $\Psi\in \mathcal{S}\hat{\otimes} \FSfin(\mathscr{T})$.
\begin{prop}
  \label[proposition]{prop:unitarygroupspin}
  For any normal and bounded $B$ on $\mathcal{S}$, the collection $(U_{B,g,g'})_{g,g'\in
    \mathscr{T}'}$ forms a two-parameter unitary group, namely the following properties
  hold true: for any $g,g',g''\in \mathscr{T}'$,
  \begin{itemize}
  \item $U_{B,g,g'}: \mathcal{B}_{B,g} \to \mathcal{B}_{B,g'}$ is unitary;
  \item $U_{B,g,g}= \mathds{1}$;
  \item $U_{B,g',g}= U^{-1}_{B,g,g'}$;
  \item $U_{B,g',g''}U_{B,g,g'}= U_{B,g,g''}$.
  \end{itemize}
\end{prop}
\begin{proof}
  The proof is completely analogous to that of \cref{prop:unitarygroup} in
  the trivial spin case, let us only point out the modifications needed to
  prove unitarity. Considering vectors of the form $\{\psi\otimes \varepsilon_{g}(f)\;,\; \psi\in\cS\;,\; f\in \mathscr{T}\}$, it is not
  difficult to see that
  \begin{equation*}
    U_{B,g,g'}\bigl(\psi_k\otimes \varepsilon_g(f)\bigr) = \bigl(\e^{B^{*}\langle g-g'  , f \rangle_{}}\psi_k\bigr)\otimes \varepsilon_{g'}(f)\;.
  \end{equation*}
  Taking $g'=0$ one gets unitarity of $U_{B,g,0}$ since $\{\psi\otimes \varepsilon_0(f)\;,\; \psi\in\cS\;,\; f\in \mathscr{T}\}$ is total in $\mathcal{B}_0$, and $\e^{B^{*}\langle g , f \rangle_{}}$ is a bijection
  on $\mathcal{S}$. From this, it also follows that $\{\psi\otimes \varepsilon_{g'}(f)\;,\; \psi\in\cS\;,\; f\in \mathscr{T}\}$
  are total in $\mathcal{B}_{B,g'}$ for any $g'\in \mathscr{T}'$, and thus $U_{B,g,g'}$ is unitary
  as well.
\end{proof}
\cref{prop:4,lemma:5} are also provable for nontrivial spin-boson spaces
$\mathcal{B}_{B,g}$, let us write the statements explicitly for later
reference. Beforehand, observe that for any $g\neq0$, the interacting fields
$\phi_{B,g}(f)$ and $\pi_{B,g}(f)$ clearly depend on the interaction spin-matrix
$B$, however their form is perfectly analogous to the one obatined for the
vHM model.
\begin{prop}
  \label[proposition]{prop:4bis}
  Let $B$ be normal and bounded on $\mathcal{S}$, and let $g,g'\in \mathscr{T}$; then the
  interacting fields and conjugate momenta on $\mathcal{B}_{B,g}$ and $\mathcal{B}_{B,g'}$ are
  related by:
  \begin{gather*}
    \phi_{B,g}(f)= U_{B,g',g}\Bigl(\phi_{B,g'}(f)+ B^{*}B\sqrt{2}\Re \langle g-g'  , f \rangle\Bigr)U_{B,g,g'}\;,\\
    \pi_{B,g}(f)=U_{B,g',g}\Bigl(\pi_{B,g'}(f)+ B^{*}B\sqrt{2}\Im \langle g-g'  , f \rangle\Bigr)U_{B,g,g'}\;.
  \end{gather*}

  It follows that $\mathcal{B}_{B,g}$ and $\mathcal{B}_{B,g'}$ carry unitarily equivalent
  representations of the Weyl C*-algebra of canonical commutation relations
  if and only if $g-g'\in \mathfrak{h}$ (with $\mathfrak{h}$ seen as a subspace of $\mathscr{T}'$). In
  particular, the representation of interacting fields carried by $\mathcal{B}_{B,g}$
  is disjoint from the Fock representation of free fields on $\mathcal{B}_{0}$ whenever
  $g\in \mathscr{T}'\smallsetminus \mathfrak{h}$.
\end{prop}
\begin{lem}
  \label[lemma]{lemma:5bis}
  Let $g_1\in \mathscr{T}_1'$ and $g_2\in \mathscr{T}_2'$ be such that for all $f\in \mathscr{T}_1\cap \mathscr{T}_2$,
  $g_1(f)=g_2(f)$. Then, for any normal bounded $B$ on $\mathcal{S}$ there exists a
  natural unitary map $\Omega_{B,g_1,g_2}: \mathcal{B}_{B,g_1}\to \mathcal{B}_{B,g_2}$.
\end{lem}

\subsection{SB wave function renormalization II: renomalized Hamiltonian}
\label{sec:sb-wave-function}

In analogy with the vHM case, let us consider a field with dispersion $\omega$,
and take a Gel'fand triple $(\mathscr{T},\mathfrak{h},\mathscr{T}')$ such that $\omega:\mathscr{T}\to \mathscr{T}$ is a linear
homeomorphism that extends to a positive operator $\omega\geq 0$ on $\mathfrak{h}$, with
$\sqrt{\omega}\bigr\rvert_{\mathscr{T}}$ also being an homeomorphism. Let us then fix $v\in \mathscr{T}'$,
and an $\fh$-mollification $v_\alpha$ of $v$ such that $v_\alpha/\sqrt\omega,v_\alpha/\omega\in\fh$.
Also again denote by $E(v_\alpha)\coloneqq  -\lVert v_\alpha/\sqrt{\omega} \rVert_{\mathfrak{h}}^2$, and by $g=-v/\omega\in
\mathscr{T}'$, respectively $g_\alpha=-v_\alpha/\omega\in\fh$ the one-particle field configurations.

The SB quadratic form is renormalized as follows: define the renormalized
Hamiltonian as
\begin{equation*}
  H_{\mathsf{ren}}(B,v_\alpha)= H(B,v_\alpha)- E(v_\alpha)\bigl(B^{*}B\otimes \mathds{1}\bigr)\;,
\end{equation*}
and the renormalized quadratic form on $\mathcal{B}_{g_\alpha}$ as
\begin{equation*}
  q^{\mathsf{ren}}_{B,v_\alpha}(\Psi,\Phi)\coloneqq  \bigl\langle \e^{-\frac{1}{2}B^{*}B\lVert g_\alpha  \rVert_{\mathfrak{h}}^2}\e^{B\otimes a^{\dagger}(g_\alpha)}\Psi  , H_{\mathrm{ren}}(B,v_\alpha)\e^{-\frac{1}{2}B^{*}B\lVert g_\alpha  \rVert_{\mathfrak{h}}^2}\e^{B\otimes a^{\dagger}(g_\alpha)}\Phi \bigr\rangle_{\mathcal{B}_0}\;.
\end{equation*}
Let us first consider the SB Hamiltonian without the free spin matrix:
\begin{equation*}
  \tilde{H}_{\mathsf{ren}}(B,v_\alpha)\coloneqq   H_{\mathsf{ren}}(B,v_\alpha) - A\otimes \mathds{1}= H_{\mathsf{ren}}(B,v_\alpha)\bigr\rvert_{A=0}\;,
\end{equation*}
and its associated quadratic form
\begin{equation*}
  \tilde{q}^{\mathsf{ren}}_{B,v_\alpha}(\Psi,\Phi)\coloneqq  \bigl\langle \e^{-\frac{1}{2}B^{*}B\lVert g_\alpha  \rVert_{\mathfrak{h}}^2}\e^{B\otimes a^{\dagger}(g_\alpha)}\Psi  , \tilde{H}_{\mathrm{ren}}(B,v_\alpha)\e^{-\frac{1}{2}B^{*}B\lVert g_\alpha  \rVert_{\mathfrak{h}}^2}\e^{B\otimes a^{\dagger}(g_\alpha)}\Phi \bigr\rangle_{\mathcal{B}_0}\;.
\end{equation*}
Such an Hamiltonian is diagonalized exactly by the SB dressing $U_{B,g_\alpha,0}$,
in full analogy with the vHM Hamiltonian; the fact, that the
SB model is not exactly diagonalizable is due \emph{solely} to the dressing
of the free spin matrix.
\begin{lem}
  \label[lemma]{lemma:99}
  For any $\Psi,\Phi\in \FSfin(\mathscr{T})$, we have that
  \begin{equation*}
    \tilde{q}^{\mathsf{ren}}_{B,v_\alpha}(\Psi,\Phi)= \langle \e^{B^{*}\otimes a(g_\alpha)}\Psi  , \bigl(\mathds{1}\otimes \mathsf{d}\Gamma(\omega)\bigr) \e^{B^{*}\otimes a(g_\alpha)}\Phi \rangle_{\mathcal{B}_0}\;,
  \end{equation*}
  and therefore it is the quadratic form of the self-adjoint operator
  $\tilde{H}_{\mathsf{ren}}(B,v_\alpha)_{g_\alpha}$ on $\mathcal{B}_{B,g_\alpha}$ defined by the
  unitary CCR-equivalent conjugation
  \begin{equation*}
    \tilde{H}_{\mathsf{ren}}(B,v_\alpha)_{g_\alpha}\coloneqq  U_{B,0,g_\alpha}\bigl(\mathds{1}\otimes \mathsf{d}\Gamma(\omega)\bigr)U_{B,g_\alpha,0}\;. 
  \end{equation*}
  Its limit counterpart on $\mathcal{B}_{B,g}$ is the self-adjoint operator
  \begin{equation*}
    \tilde{H}_{\mathsf{ren}}(B,v)_g= U_{B,0,g}\bigl(\mathds{1}\otimes \mathsf{d}\Gamma(\omega)\bigr)U_{B,g,0}\;,
  \end{equation*}
  with domain of self-adjointness $U_{B,0,g}\mathscr{D}\bigl(\mathds{1}\otimes \mathsf{d}\Gamma(\omega)\bigr)$,
  and core $\mathcal{S}\hat{\otimes} \FSfin(\mathscr{T})$. Finally,
  $\tilde{H}_{\mathsf{ren}}(B,v_\alpha)_{g_\alpha}$ converges to
  $\tilde{H}_{\mathsf{ren}}(B,v)_g$ in generalized norm resolvent sense (see
  \cref{sec:wave-funct-renorm} for a definition of generalized resolvent
  convergence).
\end{lem}
\begin{proof}
  The first part of the lemma is proved using the following operator
  identities, proved by the use of canonical commutation relations, the
  normality of $B$, and the Baker--Campbell--Hausdorff formula for unbounded
  operators commuting with their commutator:
  \begin{align*}
    &(\Id\otimes \dG(\omega))\e^{B\otimes\ad(g_\alpha)}
    = \e^{B\otimes\ad(g_\alpha)} \big(\Id\otimes \dG(\omega) + B\otimes\ad(\omega g_\alpha)\big)\\
    &(B\otimes \ad(v_\alpha) + B^*\otimes a(v_\alpha))\e^{B\otimes\ad(g_\alpha)}
    = \e^{B\otimes\ad(g_\alpha)}\big( B\otimes \ad(v_\alpha) + B^*\otimes a(v_\alpha) + \braket{v_\alpha,g_\alpha}_\hs B^*B\otimes\Id \big)
    \\
    & \e^{B^*\otimes a(g_\alpha)}(\Id\otimes \dG(\omega))
    =  \big(\dG(\omega) + B^*\otimes a(\omega g_\alpha)\big) \e^{B^*\otimes a(g_\alpha)}
    \\
    &
    \e^{B^*\otimes a(g_\alpha)}\e^{B\otimes\ad(g_\alpha)}=\e^{B^*B\|g_\alpha\|^2_{\fh}}\e^{B\otimes \ad(g_\alpha)}\e^{B^*\otimes a(g_\alpha)}\,.
  \end{align*}
  The domain of self-adjointness of $\tilde{H}_{\mathsf{ren}}(B,v)_g$ is
  obvious, as well as the convergence in generalized norm resolvent sense.

  It remains to prove that $\mathcal{S}\hat{\otimes}\FSfin(\mathscr{T})$ is a core for
  $\tilde{H}_{\mathsf{ren}}(B,v)_g$. More precisely, since
  $\mathcal{S}\hat{\otimes}\FSfin(\mathscr{T})$ is a core\footnote{Here, we use the fact that $\mathscr{T}$ is a
    core for $\omega$, as an operator on $\mathfrak{h}$.} for $\mathds{1}\otimes \mathsf{d}\Gamma(\omega)$, it is
  sufficient to prove that
  $U_{B,0,g}\mathcal{S}\hat{\otimes}\FSfin(\mathscr{T})=\mathcal{S}\hat{\otimes}\FSfin(\mathscr{T})$. Firstly, recall that
  $U_{B,0,g}$ is the extension of $u_{B,0,g}:\mathcal{S}\hat{\otimes}\FSfin(\mathscr{T}) \to
  \mathcal{S}\hat{\otimes}\FSfin(\mathscr{T})$. The latter is surjective, since $u_{B,g,0}u_{B,0,g}=
  \mathds{1}_{\mathcal{S}\hat{\otimes}\FSfin(\mathscr{T})}$.
\end{proof}

\subsection*{Dressing of the free spin matrix}
\label{sec:dressing-free-spin}
The dressing of the free spin matrix $A\otimes \mathds{1}$ is, in general, non-trivial due
to the non-commutativity of $A$ and $B$, with the latter appearing as part of
the dressing map $\e^{B^{*}\otimes a(g)}$. We here identify explicit operator
expressions employing so-called double operator integrals, which thus appear
as the natural framework to construct the dressing of the operator $A$.
These types of integrals are of interest by themselves, whence we provide an
explicit treatment from the viewpoint of quadratic forms as well as a brief
literature review in \cref{app:DOI}. We stress that in the finite-dimensional
case $\dim\cS<\infty$ (or more general the case that $B$ has a compact resolvent),
integrals of this type are immediately well-defined as double sums over the
eigenspace projections of $B$.
From now on, we use the notation from \cref{def:DOI} and assume that $A$ is
$(E_B,E_B)$-decomposable, where $E_B$ as usual denotes the spectral measure
of $B$; keep in mind that any $A$ that is Hilbert--Schmidt, or such that it
commutes/anti-commutes with $B$, is $(E_B,E_B)$-decomposable.

For fixed $\alpha$, the renormalized free spin matrix quadratic form is then given by
\begin{align}
	\label{eq:renfree}
  q^{\mathsf{ren}}_{B,v_\alpha}(\Psi,\Phi)-\tilde{q}^{\mathsf{ren}}_{B,v_\alpha}(\Psi,\Phi)
  = \langle \e^{-\frac{1}{2}B^{*}B \lVert g_\alpha  \rVert_{\mathfrak{h}}^2} \e^{B\otimes a^{\dagger}(g_\alpha)}\Psi  , (A\otimes \mathds{1})\e^{-\frac{1}{2}B^{*}B \lVert g_\alpha  \rVert_{\mathfrak{h}}^2}\e^{B\otimes a^{\dagger}(g_\alpha)}\Phi \rangle_{\mathcal{B}_0}\;,
\end{align}
that, as we will see below, defines a bounded operator $A(B,v_\alpha)_{g_\alpha}$
acting on the renormalized space $\mathcal{B}_{B,g_\alpha}$. The analogous construction
shall hold when $v_\alpha$ and $g_\alpha$ are replaced by $v$ and $g$ respectively.  On
pure tensors $\psi\otimes \Xi$ and $\phi\otimes \Theta$ in $\mathcal{S}\hat{\otimes}\FSfin(\mathscr{T})$, we define the form
$\mathfrak{a}^{\mathsf{ren}}_{B,v}$ by the double operator integral
\begin{equation}
  \label{eq:formalspin}
  \mathfrak{a}^{\mathsf{ren}}_{B,v}(\psi\otimes \Xi,\phi\otimes \Theta)
  \coloneqq \iint_{\IC^2} \chi_g(\lambda,\mu)\Braket{\e^{\overline{\mu} a(g)}\Xi,\e^{\overline{\lambda} a(g)}\Theta}_{\FS(\fh)}\Braket{\d E_B(\lambda)\psi,A\, \d E_B(\mu)\varphi}_{\cS} \; ,
\end{equation}
where the integrand is defined
by
\begin{align}
  \chi_g(\lambda,\mu)
  =
  \begin{cases}
    \e^{(\overline\lambda\mu-\frac12(\abs\lambda^2+\abs\mu^2))\norm{g}_\fh^2} & \mbox{for}\ g\in\fh,\\
    \chr_{\{\lambda=\mu\}} & \mbox{for}\ g\in\sT'\smallsetminus\fh.
  \end{cases}
  \label{eq:defchi}
\end{align}
Extending to the algebraic tensor product $\cS\hat\otimes\FSfin(\sT)$ by sesquilinearity gives the correct quadratic form:
\begin{lem}
	\label[lemma]{lemma:10}
	Let $\Psi,\Phi\in\cS\hat\otimes\FSfin(\sT)$. Then, for all $\alpha\in \mathcal{A}$,
	\begin{equation*}
		\mathfrak{a}^{\mathsf{ren}}_{B,v_\alpha}(\Psi,\Phi) = q^{\mathsf{ren}}_{B,v_\alpha}(\Psi,\Phi)-\tilde{q}^{\mathsf{ren}}_{B,v_\alpha}(\Psi,\Phi).
	\end{equation*}
\end{lem}
\begin{proof}
  It suffices to prove the statement for vectors of the form $\Psi = \psi\otimes\Xi$ and $\Phi
  = \varphi\otimes\Theta$.  Then expanding the exponentials $\e^{B^*\otimes \ad(g_\alpha)}$ inside the
  form $q^{\mathsf{ren}}_{B,v_\alpha}-\tilde{q}^{\mathsf{ren}}_{B,v_\alpha}$, we obtain
  terms of the form
	\begin{align*}
		&\frac{1}{m!\,p!}\Braket{\e^{-\frac12B^*B\norm{g_\alpha}_\fh^2}B^m\psi,A\,\e^{-\frac12B^*B\norm{g_\alpha}_\fh^2}B^p\varphi}_\cS\Braket{\ad(g_\alpha)^m\Xi,\ad(g_\alpha)^p\Theta}_{\FS(\fh)}
		\\&\qquad
		=
		\frac{1}{m!\,p!}\iint_{\IC^2}\e^{-\frac12(\abs\lambda^2+\abs\mu^2)\norm{g_\alpha}_\fh^2}\Braket{\lambda^m\ad(g_\alpha)^m\Xi,\mu^p\ad(g_\alpha)^p\Theta}_{\FS(\fh)}\Braket{\d E_B(\lambda)\psi, A \d E_B(\mu)\ph}_\cS,
	\end{align*}
	by \cref{eq:PVMFubini}.  Summing over $m$ and $p$ and applying
	the dominated convergence theorem, we obtain
	\begin{align*}
		&\Braket{\e^{-\frac12B^*B\norm{g_\alpha}_\fh^2}\e^{B\otimes \ad(g_\alpha)}\Psi,(A\otimes\Id)\e^{-\frac12B^*B\norm{g_\alpha}_\fh^2}\e^{B\otimes \ad(g_\alpha)}\Phi}_{\cS\otimes\FS(\fh)}
		\\&\qquad
		= \iint_{\IC^2}\e^{-\frac12(\abs\lambda^2+\abs\mu^2)\norm{g_\alpha}_\fh^2} \Braket{\e^{\lambda\ad(g_\alpha)}\Xi,\e^{\mu\ad(g_\alpha)}\Theta}_{\FS(\fh)}\Braket{\d E_B(\lambda)\psi, A \d E_B(\mu)\ph}_\cS
		\\&\qquad
		= \iint_{\IC^2}\chi_{g_\alpha}(\lambda,\mu) \Braket{\e^{\overline\mu a(g_\alpha)}\Xi,\e^{\overline\lambda a(g_\alpha)}\Theta}_{\FS(\fh)}\Braket{\d E_B(\lambda)\psi, A \d E_B(\mu)\ph}_\cS
		,
	\end{align*}
	where the last step was again an application of the canonical commutation
	relations and the Baker--Campbell--Hausdorff formula.  Comparing with
	\cref{eq:formalspin}, we have proved the statement.
\end{proof}
This construction yields a bounded quadratic form on all of $\cB_{B,g}$, even for singular $v\in\sT'\smallsetminus\fh$.
\begin{prop}
  \label[proposition]{prop:spinform}
  Equations \cref{eq:formalspin,eq:defchi} define a unique bounded symmetric
  quadratic form $\mathfrak{a}^{\mathsf{ren}}_{B,v}$ on $\cB_{B,g}$. It is bounded by
  the operator norm $\|A\|_{\cS\to\cS}$.
\end{prop}
\begin{proof}
	By the dominated convergence theorem, \cref{eq:formalspin} and \cref{lemma:2}, on pure tensors and thus for all $\Psi,\Phi\in\cS\hat\otimes\FSfin(\sT)$,
	\[ \mathfrak{a}^{\mathsf{ren}}_{B,v_\alpha}(\Psi,\Phi) \overset{\alpha}{\longrightarrow} \mathfrak{a}^{\mathsf{ren}}_{B,v}(\Psi,\Phi). \]
	By \cref{lemma:10,eq:renfree,prop:11}, we further have
	\begin{align*}
		\abs{\mathfrak{a}^{\mathsf{ren}}_{B,v_\alpha}(\Psi,\Phi)}
		&\le \norm{A}_{\cS\to\cS}\norm{e^{-\frac12 B^*B\norm{g_\alpha}_\fh^2}e^{B\otimes\ad(g_\alpha)}\Psi}_{\cB_{0}}\norm{e^{-\frac12 B^*B\norm{g_\alpha}_\fh^2}e^{B\otimes\ad(g_\alpha)}\Phi}_{\cB_{0}}
		 \\&= \norm{A}_{\cS\to\cS} \norm{\Psi}_{\cB_{B,g_\alpha}} \norm{\Phi}_{\cB_{B,g_\alpha}}.
	\end{align*}
	Once more applying \cref{lemma:2} and the definition of the $\mathcal{B}_{B,g}$ inner product, we know that $\norm{\Psi}_{\mathcal{B}_{B,g_\alpha}}\xrightarrow{\alpha} \norm{\Psi}_{\mathcal{B}_{B,g}}$, so 
	\begin{equation*}
		\abs{\mathfrak{a}^{\mathsf{ren}}_{B,v}(\Psi,\Phi)} \le \norm{A}_{\cS\to\cS} \norm{\Psi}_{\cB_{B,g}}\norm{\Phi}_{\cB_{B,g}},
	\end{equation*}
	which proves boundedness on the dense subspace $\cS\hat\otimes\FSfin(\sT)$.
        Thus, we have shown the existence of a unique bounded sesquilinear
        extension to all of $\mathcal{B}_{B,g}$. Symmetry is immediate, since limits of
        symmetric forms are always symmetric.
\end{proof}
We are now in a position to define the renormalized spin Hamiltonian
$A(B,v)_g$.
\begin{defn}[$A(B,v)_g$]
  \label[definition]{def:Ar}
  Given $A,B\in\cL(\cS)$ with $A$ self-adjoint and $B$ normal, we define the
  renormalized spin matrix $A(B,v)_g$ to be the unique bounded self-adjoint
  operator on $\mathcal{B}_{B,g}$ yielded, by Riesz's representation theorem,
  by the quadratic form $\mathfrak{a}^{\mathsf{ren}}_{B,v}$ of
  \cref{prop:spinform}.
\end{defn}
Let us stress the fact that, due to the general
non-commutativity of $A$ and $B$, the operator $A(B,v)_g$ \emph{is not} of
the diagonal dressed form $U_{B,0,g} (T\otimes \mathds{1})U_{B,g,0}$; as a
matter of fact, such a dressed diagonal operator takes the -- subtly similar,
but \emph{different} -- form below.
\begin{lem}
  \label[lemma]{lemma:noncomm}
  If $T\in\cL(\cS)$, then
  \begin{align*}
    \Braket{\psi\otimes\Xi,U_{B,0,g}(T\otimes\Id)U_{B,g,0}\ph\otimes\Theta}_{B,g}
    =
    \iint_{\IC^2} \Braket{\e^{\overline\lambda a(g)}\Xi,\e^{\overline\mu a(g)}\Theta}_{\FS(\fh)}\Braket{\d E_B(\lambda)\psi,T\d E_B(\mu)\varphi}_\cS
  \end{align*}
  for all $\psi,\ph\in\cS$ and $\Xi,\Theta\in\FSfin(\sT)$.
\end{lem}
\begin{proof}
  By definition
  \[
    \Braket{\psi\otimes\Xi,U_{0,g,B}(T\otimes\Id)U_{g,0,B}\ph\otimes\Theta}_{B,g}
    = \Braket{\e^{B^*\otimes a(g)}\psi\otimes\Xi,(T\otimes\Id)\e^{B^*\otimes
        a(g)}\ph\otimes\Theta}_{\mathcal{B}_0}\,.
  \]
  By expanding the exponential into its truncated series and then applying
  \cref{eq:PVMFubini}, similarly to the above, the statement follows.
\end{proof}

\subsection*{The renormalized SB Hamiltonian and cutoff removal}
\label{sec:renorm-sb-hamilt}

Joining the two results together, we are in a position to prove the main
result concerning the renormalization in the SB model. Beforehand, let us
define
\begin{equation*}
  H_{\mathsf{ren}}(B,v)_g\coloneqq  A(B,v)_g+ \tilde{H}_{\mathsf{ren}}(B,v)_g\;.
\end{equation*}
\begin{thm}[Renormalization of the SB model]
  \label[theorem]{thm:4}
  Let $B$ be bounded and normal on $\mathcal{S}$, then for any $v\in
  \mathscr{T}'$ the renormalized Hamiltonian $H_{\mathsf{ren}}(B,v)_g$ has
  the following properties:

  \vspace{1.5mm}
  
  \begin{itemize}
    \setlength{\itemsep}{3mm}
  \item For any net $\{v_\alpha\}\subset \sT'$ such that $v_\alpha\xrightarrow[\alpha]{\sigma(\sT',\sT)} v$, the
    Hamiltonian $H_{\mathsf{ren}}(B,v_\alpha)_{g_\alpha}$ converges to
    $H_{\mathsf{ren}}(B,v)_g$ in generalized strong resolvent sense. Convergence is in generalized norm resolvent sense, if $g_\alpha,g\in\fh$ and $g_\alpha\xrightarrow[\alpha]{\fh}g$.
  \item $H_{\mathsf{ren}}(B,v)_g= A(B,v)_g+ U_{B,0,g}\bigl(\mathds{1}\otimes
    \mathsf{d}\Gamma(\omega)\bigr)U_{B,g,0}$ with $g=-v/\omega$, and
    therefore:

    \vspace{1mm}
    
    \begin{itemize}
      \setlength{\itemsep}{2mm}
    \item $H_{\mathsf{ren}}(B,v)_g$ is self-adjoint on $B_{B,g}$ with domain
      $\mathscr{D}\bigl(H_{\mathsf{ren}}(B,v)_g\bigr)=
      U_{B,0,g}\mathscr{D}\bigl(\mathds{1}\otimes \mathsf{d}\Gamma(\omega)
      \bigr)$ and core $\mathcal{S}\hat{\otimes}\FSfin(\mathscr{T})$;
    \item $\inf \sigma\bigl(H_{\mathsf{ren}}(B,v)_g\bigr)\ge-\|A\|_{\cS\to\cS}$.
    \end{itemize}
  \end{itemize}
\end{thm}
\begin{proof}
  Let us start by proving the second point, for it is rather straightforward:
  $A(B,v)_g$ is a bounded self-adjoint operator on $\mathcal{B}_{B,g}$ by
  \cref{prop:spinform,def:Ar}, while $U_{B,0,g}\bigl(\mathds{1}\otimes
  \mathsf{d}\Gamma(\omega)\bigr)U_{B,g,0}$ is self-adjoint on
  $U_{B,0,g}\mathscr{D}\bigl(\mathds{1}\otimes \mathsf{d}\Gamma(\omega)
  \bigr)$ with $\mathcal{S}\hat{\otimes}\FSfin(\mathscr{T})$ as core by
  \cref{lemma:99}, hence by Kato-Rellich's theorem it follows that
  $H_{\mathsf{ren}}(B,v)_g$ is self-adjoint on the same domain and core, and
  bounded from below.

  We now prove generalized strong resolvent convergence.  By \cref{lemma:7},
  since the maps $U_{B,g_\alpha,0}$ and $U_{B,g,0}$ are all unitary, defining for
  a generic $w\in\sT'$ with $h=-w/\omega$
  \begin{equation*}
  	\hat{H}_{\mathsf{ren}}(B,w)_{h}\coloneqq U_{B,h,0}H_{\mathsf{ren}}(B,w)_{h}U_{B,0,h}\;,
  \end{equation*}
   it suffices
  to prove strong resolvent convergence of $\hat{H}_{\mathsf{ren}}(B,v_\alpha)_{g_\alpha}$ to $\hat{H}_{\mathsf{ren}}(B,v)_{g}$  on $\mathcal{B}_0$.
  Recall that
  \begin{equation*}
  	\hat{H}_{\mathsf{ren}}(B,w)_h\coloneqq U_{B,h,0}H_{\mathsf{ren}}(B,w)_{h}U_{B,0,h}= \mathds{1}\otimes\mathsf{d}\Gamma(\omega) + \hat{A}(B,w)_h\;, 
  \end{equation*}
  where $\hat{A}(B,w)_h$ is a bounded self-adjoint operator on
  $\mathcal{B}_0$, given by
  \begin{equation*}
  	\hat{A}(B,w)_h= U_{B,h,0}A(B,w)_hU_{B,0,h}\;.
  \end{equation*}
  Its associated quadratic form is thus
  \begin{equation*}
  	\hat{q}^{\mathsf{ren}}_{B,w}(\Psi,\Phi)\coloneqq  q_0(\Psi,\Phi)+ \hat{\mathfrak{a}}^{\mathsf{ren}}_{B,w}(\Psi,\Phi)= \langle \Psi  , \bigl(\mathds{1}\otimes \mathsf{d}\Gamma(\omega)\bigr)\Phi \rangle_{\mathcal{B}_0}+ \langle \Psi  , \hat{A}(B,w)_h \Phi\rangle_{\mathcal{B}_0}\;.
  \end{equation*}
  Now by \cite[Thm.~13.12]{dalmaso1993pnde} the claimed strong resolvent convergence follows, if for all $\Psi\in \cB_0$ and all $\Psi_\alpha\xrightarrow{\alpha}\Psi$, we can prove
  \begin{align}
  		\label{eq:upperbound}
  		&\limsup_\alpha\hat{q}^{\mathsf{ren}}_{B,v_\alpha}(\Psi) \le \hat{q}^{\mathsf{ren}}_{B,v}(\Psi)\;,\\
  		\label{eq:lowerbound}
  		&\liminf_\alpha\hat{q}^{\mathsf{ren}}_{B,v_\alpha}(\Psi_\alpha) \ge \hat{q}^{\mathsf{ren}}_{B,v}(\Psi)\;.
  \end{align}
  The first claim \cref{eq:upperbound} is apparent from the fact that $\fa_{B,v_\alpha}^{\mathsf{ren}}(\Psi)\xrightarrow{\alpha}\fa_{B,v}^{\mathsf{ren}}(\Psi)$,
  by the definition \cref{eq:formalspin}, the dominated convergence theorem, the continuity statement \cref{lemma:2} and the totality of pure tensors.
  To prove \cref{eq:lowerbound}, we note that lower-semicontinuity of closed quadratic forms \cite[Thm.~2]{Simon.1978} implies
  \begin{align*}
  	\liminf_\alpha\langle \Psi_\alpha  , \bigl(\mathds{1}\otimes \mathsf{d}\Gamma(\omega)\bigr)\Psi_\alpha \rangle_{\mathcal{B}_0} \ge \langle \Psi  , \bigl(\mathds{1}\otimes \mathsf{d}\Gamma(\omega)\bigr)\Psi \rangle_{\mathcal{B}_0}\;.
  \end{align*}
  Furthermore, $\fa_{B,v_\alpha}^{\mathsf{ren}}(\Psi_\alpha)\xrightarrow{\alpha} \fa_{B,v}^{\mathsf{ren}}(\Psi)$, by the estimate
  \begin{align*}
  		\big|\fa_{B,v_\alpha}^{\mathsf{ren}}(\Psi_\alpha)-\fa_{B,v}^{\mathsf{ren}}(\Psi)\big|
  		&\le
  		\big|\fa_{B,v_\alpha}^{\mathsf{ren}}(\Psi) - \fa_{B,v}^{\mathsf{ren}}(\Psi)\big| + \big|\fa_{B,v_\alpha}^{\mathsf{ren}}(\Psi_\alpha) - \fa_{B,v_\alpha}^{\mathsf{ren}}(\Psi)\big|,
  \end{align*}
  where the first term again converges to zero, by the above consideration,
  and the second converges to zero, by the uniform continuity of
  $\fa_{B,v_\alpha}$ with respect to $\alpha\in \mathcal{A}$, \cref{prop:spinform}. Finally, generalized
  norm resolvent convergence with $g_\alpha,g\in\fh$ and $g_\alpha\xrightarrow[\alpha]{\fh}g$ is a rephrasing in this setting of the analogous
  result for the unitary dressing \citep{GriesemerWuensch.2018}.

\end{proof}

\subsection{(Un)expected exact diagonalization: the standard and
  energy-preserving SB models}
\label{sec:unexp-exact-diag}

As anticipated above, and then proved in \cref{thm:4}, the renormalized SB
Hamiltonian is, in general, not diagonalizable exactly by the renormalized SB
dressing, due to the non commutativity between the free and interaction spin
matrix. There are at least two notable exceptions: the \emph{standard SB},
and the \emph{energy-preserving SB}. While the diagonalizability of the
latter holds also for regular sources and it was thus completely expected,
the diagonalizability of the standard SB is true \emph{only for singular
  sources} and is somewhat unexpected (and also rather special, as discussed
below).

\subsection*{The energy-preserving SB}
\label{sec:energy-preserving-sb}

In the energy-preserving SB, we have $B=A$ (the same result would hold true,
more generally, for any $B$ commuting with $A$). Therefore, in that case, the
renormalized spin matrix is diagonalized by the renormalized SB dressing:
\begin{equation*}
  A(A,v)_g= U_{A,0,g}\bigl(A\otimes \mathds{1}\bigr)U_{A,g,0}\;,
\end{equation*}
as well as
\begin{equation*}
  A(A,v_\alpha)_{g_\alpha}= U_{A,0,g_\alpha}\bigl(A\otimes \mathds{1}\bigr)U_{A,g_\alpha,0}\;.
\end{equation*}
This shows that both for regular and for singular sources, the dressed SB
Hamiltonian is diagonal:
\begin{gather*}
  H_{\mathsf{ren}}(A,v_\alpha)_{g_\alpha}= U_{A,0,g_\alpha}\bigl(A\otimes \mathds{1} + \mathds{1}\otimes \mathsf{d}\Gamma(\omega)\bigr)U_{A,g_\alpha,0}\;;\\
  H_{\mathsf{ren}}(A,v)_{g}= U_{A,0,g}\bigl(A\otimes \mathds{1} + \mathds{1}\otimes \mathsf{d}\Gamma(\omega)\bigr)U_{A,g,0}\;.
\end{gather*}
It follows that the spectral properties of $H_{\mathsf{ren}}(A,v)_g$ and
$H_0$ coincide: if the latter has a non-degenerate ground state
$\Psi_{\mathsf{gs}}$, so it happens for $H_{\mathsf{ren}}(A,v)_g$ (the ground
state being $U_{A,0,g}\Psi_{\mathsf{gs}}$).

\subsection*{The standard SB}
\label{sec:standard-sb}

The standard SB model has $\mathcal{S}=\mathbb{C}^2$, $A=\sigma_z$ and $B=\sigma_x$, so in particular
$\{\sigma_z,\sigma_x\}=0$, \emph{i.e.}\ the free and interaction spin matrices
\emph{anticommute} (what follows is in fact true for any $A$ and $B$ that
anticommute with $AB\neq0$). Consider now the function $\chi_h(\lambda,\mu)$ defined in
\eqref{eq:defchi}: whenever $h\in \mathscr{T}'\smallsetminus \mathfrak{h}$, it is supported on the diagonal
$\lambda=\mu$; while it is supported everywhere, when $h\in \mathfrak{h}$. For anticommuting $A$
and $B$, the support of $\chi_g$ on the diagonal entails that the double
integral \emph{vanishes}, and thus $\forall g\in \mathscr{T}'\smallsetminus \mathfrak{h}$ and $\{A,B\}=0$,
\begin{equation*}
  A(B,v)_g=0 \;.
\end{equation*}
Hence, for the standard spin-boson with $g\in \mathscr{T}'\smallsetminus
\mathfrak{h}$,
\begin{equation*}
  H_{\mathsf{ren}}(\sigma_x,v)_g = U_{\sigma_x,0,g} \bigl(\mathds{1}\otimes \mathsf{d}\Gamma(\omega)\bigr) U_{\sigma_x,g,0}\;.
\end{equation*}
Therefore, while the free standard SB has a non-degenerate ground state
$(0,1)\otimes \Omega$ on $\mathcal{B}_0$, the renormalized interacting standard SB with singular
form factor (\emph{including} the one of Weisskopf-Wigner type) has a
two-fold degenerate ground state, with the ground eigenspace being
$U_{\sigma_x,0,g}\bigl(\mathbb{C}\otimes \Omega\bigr)$. We subsume the above discussion in the theorem
below.
\begin{thm}[Renormalization of the standard SB with a singular form factor]
  \label[theorem]{thm:5}
  Let $\mathcal{S}=\mathbb{C}^2$, $A=\sigma_z$ and $B=\sigma_x$. Then for any $v\in \mathscr{T}'$ such that $g=-v/\omega\in
  \mathscr{T}'\smallsetminus \mathfrak{h}$, the renormalized Hamiltonian $H_{\mathsf{ren}}(\sigma_x,v)_g$ has the
  following properties:

  \vspace{1.5mm}
  
  \begin{itemize}
    \setlength{\itemsep}{3mm}
  \item For any net $\{v_\alpha\}_{\alpha\in \mathbb{N}}\subset \mathscr{T}'$ such that
    $v_\alpha\xrightarrow[\alpha]{\sigma(\mathscr{T}',\mathscr{T})} v$ in the $\sigma(\mathscr{T}',\mathscr{T})$-topology, the
    regularized SB Hamiltonian $H_{\mathsf{ren}}(\sigma_x,v_\alpha)_{g_\alpha}=
    A(\sigma_x,v_\alpha)_{g_\alpha}+ U_{\sigma_x,0,g_\alpha}\bigl(\mathds{1}\otimes
    \mathsf{d}\Gamma(\omega)\bigr)U_{\sigma_x,g_\alpha,0}$ converges to
    $H_{\mathsf{ren}}(\sigma_x,v)_g$ in generalized strong resolvent sense.
  \item $H_{\mathsf{ren}}(\sigma_x,v)_g= U_{\sigma_x,0,g}\bigl(\mathds{1}\otimes
    \mathsf{d}\Gamma(\omega)\bigr)U_{\sigma_x,g,0}$, and therefore:

    \vspace{1mm}
    
    \begin{itemize}
      \setlength{\itemsep}{2mm}
    \item $H_{\mathsf{ren}}(\sigma_x,v)_g$ is self-adjoint on $B_{\sigma_x,g}$ with
      domain $\mathscr{D}\bigl(H_{\mathsf{ren}}(\sigma_x,v)_g\bigr)= U_{\sigma_x,0,g}\mathscr{D}\bigl(\mathds{1}\otimes
      \mathsf{d}\Gamma(\omega) \bigr)$ and core $\mathbb{C}^2\hat{\otimes}\FSfin(\mathscr{T})$;
    \item $\inf \sigma\bigl(H_{\mathsf{ren}}(\sigma_x,v)_g\bigr)=0$;
    \item $H_{\mathsf{ren}}(\sigma_x,v)_g$ has a two-fold degenerate ground state,
      with ground eigenspace $U_{\sigma_x,0,g}\bigl(\mathbb{C}^2\otimes \Omega\bigr)$.
    \end{itemize}
  \end{itemize}
\end{thm}

\section{Wave Function Renormalization of the Nelson Model}
\label{sec:wave-funct-renorm-2}

\subsection*{The Nelson model in mathematical physics}
\label{sec:nelson-math-phys}

The Nelson model is the next possible generalization of the spin-boson models
covered in the previous section.  In it, the two-state system is replaced by
a Schr\"odinger particle, which is again linearly coupled to the bosonic
field.  Since its introduction due to E.\ Nelson \citep{nelson1964jmp}, this
model -- also named after him -- has been widely studied by the community of
mathematical physics. Its most prominent feature, directly covered in
Nelson's original paper and later refined in
\cite{Cannon.1971,Ammari.2000,GriesemerWuensch.2018}, is that it allows for
self-energy renormalization similar to the one presented in
\cref{sec:energy-renorm} for the vHM model.  It has also served as a valuable
testing ground for other approaches than unitary dressings, \emph{e.g.}, the
interior boundary condition/contact interaction method
\cite{TeufelTumulka.2016,LampartSchmidt.2019,posilicano2020mpag} and
functional integration through Feynman--Kac formulas
\cite{GubinelliHiroshimaLorinczi.2014,MatteMoller.2018}.  For the so-called
pseudo-relativistic model, in which the particles have relativistic
dispersion, an additive renormalization scheme has been successfully carried
out in two space dimensions
\cite{Sloan.1974,Schmidt.2019,HinrichsMatte.2022}, but not in the physical
three-dimensional case \cite{Gross.1973,DeckertPizzo.2014}.

The infrared behavior of the Nelson model has received similar attention,
since its natural infrared divergences cause the absence of a ground state in
the Fock representation, with and without UV cutoff
\cite{LorincziMinlosSpohn.2002,DerezinskiGerard.2004,hirokawa2006prims,Panati.2009,hiroshima2022rmp}.
To overcome this, one usually constructs the {\em infraparticle}, {\em i.e.,}
a dressed Hamiltonian incorporating an infrared divergent cloud of bosons,
which then possesses a ground state, in a so-called {\em non-Fock
  representation} of the model
\cite{Arai.2001,Sasaki.2005,Panati.2009,hiroshima2022rmp}.  We here prove
that the wave function renormalization procedure for the Nelson model
provides a full operator-theoretic construction of this non-Fock
representation without IR or UV cutoffs, substantiating the algebraic result
of A.\ Arai \cite{Arai.2001} for the regularized model.

Concerning the IR catastrophe, of special interest is the
translation-invariant model, in which the Schr\"odinger particle is not
subjected to an external potential.  The so-obtained operator is
translation-invariant in the sense that it commutes with the total momentum
operator, {\em i.e.}, the sum of particle's and field's momenta.  It thus
decomposes with respect to the total momentum into the direct sum of {\em
  fiber operators}.  J.\ Fr\"ohlich observed \cite{Frohlich.1973} that these
fibers do not posses a ground state for any total momentum, due to the
infrared catastrophe; see \cite{Dam.2018} and \cite{DamHinrichs.2021} for
full proofs with and without ultraviolet cutoff, respectively. The non-Fock
representation of the translation invariant model is only partially infrared
renormalized in the sense that the zero-momentum fiber exhibits a ground
state, whereas the non-zero momentum fibers do not. Non-Fock representations
of the infraparticle fibers were also proposed by J.\ Fr\"ohlich
\cite{Frohlich.1974}, and later constructed by A.\ Pizzo \cite{Pizzo.2003}
perturbatively, in \cite{HaslerHinrichsSiebert.2023} non-perturbatively with
UV cutoff and in \cite{BachmannDeckertPizzo.2012} without UV cutoff,
respectively. They were further utilized in the construction of scattering
states for the infrared-singular but ultraviolet-regular model in
\cite{Pizzo.2005,DybalskiPizzo.2014,DybalskiPizzo.2018,DybalskiPizzo.2019,BeaudDybalskiGraf.2021,DybalskiPizzo.2022,GriesemerKussmaul.2026}.
In principle, the framework of wave function renormalization should allow for
such constructions with neither infrared nor ultraviolet cutoff in place. We
refrain from doing so in this article, because it would require introducing
further technical tools.

\subsection{The Nelson model in the Fock representation}
\label{sec:nelson-model-fock}

We still carry on the same setting and notation as in
\cref{sec:van-hove-model,sec:spin-boson-model} concerning the scalar
field. Therefore, let $(\mathscr{T},\mathfrak{h},\mathscr{T}')$ be the test/one-particle/field
configurations Gel'fand triple, $\omega$ the field's dispersion, $\mathcal{F}(\mathfrak{h})$ the Fock
space and $a^{\sharp}(f)$, $\mathsf{d}\Gamma(T)$ the creation, annihilation and second
quantization operators. As we will see, contrarily to the vHM and SB case,
the interplay between wave function renormalization and the particle's
degrees of freedom will not allow for a renormalization of generic sources,
luckily still providing new information in the physically relevant
case. Therefore, let us restrict directly to the specific physical model,
without loosing too much generality: let $\mathscr{T}= \mathscr{C}_0^{\infty}(\mathbb{R}^3_{*})\equiv\mathscr{C}_0^{\infty}(\mathbb{R}^3\smallsetminus
\{0\})$ be the space of smooth compactly supported functions away from the
origin, $\mathfrak{h}= L^2_k(\mathbb{R}^3)$ the space of square-integrable functions of momentum
$k$, and $\omega(k)= \sqrt{k^2+\mu^2}$, $\mu\geq 0$, the dispersion of the scalar
Klein--Gordon field, defined as a multiplication operator.
We remark that this specific choice of $\sT$ is sequentially dense in $\sT'$,
whence we from now on work with approximating sequences instead of nets.

 The particle's
degrees of freedom are represented by the Hilbert space $L^2_x(\mathbb{R}^3)$ -- $x$
being the particle's position -- thus yielding the total Hilbert space
\begin{equation*}
  \mathcal{N}_0\coloneqq  L^2_x(\mathbb{R}^3)\otimes \mathcal{F}\bigl(L^2_k(\mathbb{R}^3)\bigr)\;.
\end{equation*}
The free Nelson Hamiltonian is given by
\begin{equation*}
  H_0= \bigl(-\Delta_x\overset{_{\cdot}}{+}V(x)\bigr)\otimes \mathds{1} + \mathds{1}\otimes \mathsf{d}\Gamma(\omega)\;,
\end{equation*}
with $V$ an external potential acting on the particle, such that $-\Delta
\overset{_{\cdot}}{+}V$ is self-adjoint and bounded from below at least in the sense
of quadratic forms (we will be more specific in the assumptions about $V$ in
\cref{def:9}; we mostly follow \citep{hiroshima2022rmp}). As for the SB
model, $H_0$ is self-adjoint and bounded from below on $\mathcal{N}_0$, its ground
states -- if any -- being of the form $\psi\otimes \Omega$, $\psi\in L^2_x$ being a ground state
for $-\Delta\overset{_{\cdot}}{+}V$. The full Nelson Hamiltonian is obtained by coupling the particle
and the field through a first-order interaction with the scalar field
(smeared by the particle's charge distribution), evaluated at the particle's
position.
\begin{defn}[Fock Nelson Hamiltonian]
  \label[definition]{def:9}
  Let $\mu\geq 0$ be the field's mass, and let $V=V_+-V_-$ with $V_{\pm}\geq 0$ and
  $V_+\in K_3^{\mathrm{loc}}$ of local Kato class three, and $V_-\in K_3$ of Kato
  class three\footnote{These regularity assumptions on $V$ are fairly
    general, still guaranteeing the self-adjointness and boundedness from
    below of $-\Delta \overset{_\cdot}{+}V$, as well as of all interacting Nelson
    Hamiltonians we will define, including the renormalized ones. For a
    detailed definition of Kato classes refer, \emph{e.g.}, to
    \citep{AizenmanSimon.1982}.}. Then the Fock Nelson Hamiltonian $H(v)$ is
  defined as
  \begin{equation*}
    H(v)= H_0 + a^{\dagger}(v_x)+ a(v_x)= \bigl(-\Delta_x\overset{_{\cdot}}{+}V(x)\bigr)\otimes \mathds{1} + \mathds{1}\otimes \mathsf{d}\Gamma(\omega)+ a^{\dagger}(v_x)+a(v_x)\;,
  \end{equation*}
  where $v_x(k)=\e^{2\pi i k\cdot x}v(k)$, $v\in L^2_k$ being the Fourier transform of
  the particle's charge distribution (divided by $\sqrt{2\omega}$).
\end{defn}
\begin{rem*}
  The extension of Fock space creation and annihilation operators $a^{\sharp}$ to
  closed operators on $\mathcal{N}_0$ for all arguments of type $v_x$ above is fairly
  standard \citep[see, \emph{e.g.},][for a general
  definition]{GriesemerWuensch.2018}.
\end{rem*}
\begin{prop}
  \label[proposition]{prop:3}
  Let $H(v)$ with $v\in L^2_k$ be the Fock Nelson Hamiltonian. Then:
  \begin{itemize}
    \setlength{\itemsep}{2mm}
  \item $H(v)$ is essentially self-adjoint on
    $\mathscr{D}\bigl(-\Delta\overset{_{\cdot}}{+}V\bigr)\hat{\otimes}\FSfin\bigl(\mathscr{D}(\omega)\bigr)$;
  \item If $v/\sqrt{\omega}\in L^2_k$, then $H(v)$ is also bounded from below.
  \end{itemize}
\end{prop}
\begin{proof}
  Essential self-adjointness is proved in \citep{Falconi.2015}, exploiting
  the fibered structure of the Fock space, and the linearity of the
  interaction in the field operator. Boundedness from below is proved by KLMN
  theorem \citep[see, \emph{e.g.},][]{GriesemerWuensch.2018,hiroshima2022rmp}.
\end{proof}

\subsection{IR singularities and non-Fock representation (in the UV regular
  case)}
\label{sec:ir-singularities-non}

The Nelson Hamiltonian is well-defined and bounded from below in the Fock
representation whenever $v/\sqrt{\omega}\in L^2_k$ (even if $v\notin L^2$ itself, by an
application of KLMN theorem completely analogous to \cref{prop:6,prop:9} for
the vHM and SB models, respectively). This is the case, \emph{e.g.}, for the
standard form factor
\begin{equation}
  \label{eq:3}
  v_{\sigma}(k)= \frac{\chr_{\{\lvert \cdot   \rvert_{}^{}< \sigma\}}(k)}{\sqrt{2\omega(k)}}
\end{equation}
with a (sharp) UV cutoff $\sigma<+\infty$. Without loss of generality, let us focus,
from now on, upon this specific case, that coincides with the original
problem studied by Nelson \citep{nelson1964jmp}. For a massive field $\mu>0$,
and a potential $V$ satisfying the \emph{Binding Condition} (abbreviated to
BC from now on \citep[see, \emph{e.g.},][we omit here the
details]{hiroshima2022rmp}), $H(v_{\sigma})$ admits a non-degenerate ground state.
\begin{lem}
  \label[lemma]{lemma:1}
  Let $\mu>0$, let the form factor $v_{\sigma}$, $\sigma<+\infty$, be given by \eqref{eq:3},
  and let $V$ satisfy BC. Then $H(v_{\sigma})$ has a non-degenerate ground state
  $\Psi_{\mathsf{gs}}\in \mathcal{N}_0$.
\end{lem}
\begin{proof}
  The proof is originally given for a class of confining potentials in
  \citep{derezinski1999rmp,Ammari.2000}, see \citep[][Thm.\ 2.15 \&
  3.19]{hiroshima2022rmp} for a more general statement that includes the one
  given here.
\end{proof}
On the contrary, if the field is massless, $\mu=0$, no Fock ground state exists.
\begin{lem}
  \label[lemma]{lemma:12}
  Let $\mu=0$, and $v_{\sigma}$, $\sigma<+\infty$, be given by \eqref{eq:3}. Then for all $V$
  of \cref{def:9}, $\inf \sigma\bigl(H(v_{\sigma})\bigr)$ is not an eigenvalue.
\end{lem}
\begin{proof}
  The bulk of the proof originates from an argument given in
  \citep{derezinski1999rmp}; a full proof can be found in \citep[][Thm.\
  4.1]{hiroshima2022rmp}.
\end{proof}

A non-Fock representation of the massless Nelson Hamiltonian that possesses a
ground state has been constructed algebraically by A.\ Arai
\citep{Arai.2001}: the idea is to characterize the Nelson model through the
abstract Lie algebra generated by the Nelson Hamiltonian, the particle's
position and momentum operators, and the field's field and conjugate
field-momentum operators (taking the commutator as Lie bracket). Such Lie
algebra can be represented in $\mathcal{N}_0$ (yielding $H(v_{\sigma})$ as the Hamiltonian),
or in a different representation $\mathcal{N}_g$, yielding an Hamiltonian $H(v_{\sigma})_g$
with ground state.

This algebraic construction has one major drawback: it is well-known that the
Lie algebra of canonical commutation relations, even in finite dimensions,
has non-physical representations that are inequivalent to the usual
Schrödinger representation \citep[see, \emph{e.g.},][\textsection VIII.5]{reed1972I};
therefore, it is unclear why the non-Fock representation constructed in
\citep{Arai.2001} from the Lie algebra shall be the right one. We will
prove that such representation emerges naturally through our wave function
renormalization procedure, thus providing a constructive realization of
Arai's abstract argument. Furthermore, we are able to extend the construction to
the model without UV cutoff.

\subsection{UV singular Nelson model: self-energy renormalization and the
  interplay with an infrared cutoff}
\label{sec:stand-nels-model}

The removal of the UV cutoff -- \emph{i.e.}, taking the limit $\sigma\to \infty$ of
$H(v_{\sigma})$ -- has been studied since Nelson's original paper
\citep{nelson1964jmp}, and constitutes the first rigorous example of energy
renormalization. The interplay between the UV renormalization and the
infrared behavior in the massless case $\mu=0$ has been clarified, however,
only very recently in a crucial paper by F.\ Hiroshima and O.\ Matte
\citep{hiroshima2022rmp}. Nelson's original renormalization procedure
involves a unitary dressing transformation $U(\sigma,\sigma_0)$ introduced by E.P.\
Gross \citep{gross1962ap}, and it was performed for a massive field $\mu>0$; it
was later generalized to massless fields in
\citep{hirokawa2005am,GriesemerWuensch.2018}. Let us describe it schematically:
\begin{itemize}
  \setlength{\itemsep}{2mm}
\item The dressing operator $U(\sigma,\sigma_0)\coloneqq U(g_{\sigma}-g_{\sigma_0})$, with $g_{\sigma}(k)= -
  \frac{v_{\sigma}(k)}{\omega(k)+\lvert k \rvert_{}^2}$ depends on both a superior ($\sigma$) and
  inferior ($\sigma_0$) cutoff: $0\leq \sigma_0<\sigma\leq +\infty$. The dressing map implements the
  renormalization of momenta satisfying $\sigma_0\leq \lvert k \rvert_{}^{}\leq \sigma$.
\item Conjugation with the dressing operator defines the \emph{renormalized
    dressed Hamiltonian} $\hat{H}_{\mathsf{ren}}(v_{\sigma})$, once a divergent
  constant, called \emph{self-energy} has been subtracted:
  \begin{equation*}
    \hat{H}_{\mathsf{ren}}(v_{\sigma})= U(\sigma,\sigma_0) \bigl(H(v_{\sigma})-E_{\sigma,\sigma_0}\bigr)U^{*}(\sigma,\sigma_0)\;,
  \end{equation*}
  where
  \begin{equation}
  	\label{eq:selfenergy}
    E_{\sigma,\sigma_0} \coloneqq -\int_{\sigma_0\leq \lvert k  \rvert_{}^{}\leq \sigma}^{}\frac{1}{2\omega(k)(\omega(k)+\lvert k  \rvert^2)}  \mathrm{d}k\;.
  \end{equation}
  Observe that for all $\mu\geq 0$, $E_{\sigma,\sigma_0}$ is finite for all $0\leq \sigma_0<\sigma<+\infty$
  and $\lim_{\sigma\to \infty}E_{\sigma,\sigma_0}=-\infty$.
\item The operators $\hat{H}_{\mathrm{ren}}(v_{\sigma})$ converge in norm
  resolvent sense as $\sigma\to \infty$ on $\mathcal{N}_0$, thus defining the UV renormalized
  dressed Fock Nelson Hamiltonian $\hat{H}_{\mathsf{ren}}(v_{\infty})$.
\item The (undressed) renormalized Fock Nelson Hamiltonian is finally
  obtained by unitary conjugation $H_{\mathsf{ren}}(v_{\infty})= U^{*}(\infty,\sigma_0)
  \hat{H}_{\mathsf{ren}}(v_{\infty})U(\infty,\sigma_0)$.
\end{itemize}

In its original \citep{nelson1964jmp} and refined
\citep{hirokawa2005am,GriesemerWuensch.2018} forms, the above renormalization
scheme fully works for any mass $\mu\geq 0$, \emph{provided that $\sigma_0$ is taken
  big enough}\footnote{Alternatively, it is possible to trade off a large
  $\sigma_0$ with a small enough coupling constant in front of the interaction.}:
a large $\sigma_0$ is instrumental in proving self-adjointness of
$\hat{H}_{\mathsf{ren}}(v_{\infty})$ by KLMN's theorem; furthermore, for $\mu=0$ an
inferior cutoff $\sigma_0>0$ is necessary for $U(\sigma,\sigma_0)$ to be well-defined as a
unitary on $\mathcal{N}_0$. The first requirement has been recently proved in
\citep{hiroshima2022rmp} to be unnecessary, using functional integration
techniques \citep[see also][]{MatteMoller.2018}: $\hat{H}_{\mathsf{ren}}(v_{\infty})$
is self-adjoint and bounded from below for any $\sigma_0\geq 0$, and it possesses a
unique ground state if the potential $V$ satisfies the BC (the Hamiltonian
dressed with a large $\sigma_0$ has no ground state instead, since there is no
renormalization of small momenta). Nonetheless, the second assertion remains
true: the dressing $U(\infty,0)$ is not a well-defined operator on $\mathcal{N}_0$, and
therefore the resulting dressed Hamiltonian $\hat{H}_{\mathsf{ren}}(v_{\infty})$
is ``disconnected'' from the undressed one\footnote{The authors of
  \citep{hiroshima2022rmp} call $\hat{H}_{\mathsf{ren}}(v_{\infty})$ the
  ``renormalized Nelson Hamiltonian in the non-Fock representation'', since
  it is not unitary equivalent to the undressed operator
  $H_{\mathsf{ren}}(v_{\infty})$ anymore. Indeed, through our wave function
  renormalization procedure, this germinal intuition will come into full
  expression.}. We synthesize all the above contributions in the following
proposition.
\begin{prop}
  \label[proposition]{prop:12}
  Let $\mu\geq 0$, $v_{\sigma}$, $\sigma\leq +\infty$, the form factor given by \eqref{eq:3}, and
  $V=V_+-V_-$ of Kato class as in \cref{def:9}. Furthermore, let $E_{\sigma,\sigma_0}$
  be the Nelson self-energy as defined in \cref{eq:selfenergy},
  and let, for all $\sigma<+\infty$, $H_{\mathsf{ren}}(v_{\sigma})= H(v_{\sigma})- E_{\sigma,\sigma_0}$ be
  the self-adjoint \emph{renormalized Nelson Hamiltonian}. Then for all
  $0<\sigma_0<\sigma< +\infty$, there exists a unitary operator $U(\sigma,\sigma_0)$ on $\mathcal{N}_0$ such
  that
  \begin{equation*}
    H_{\mathsf{ren}}(v_{\sigma})= U^{*}(\sigma,\sigma_0)\hat{H}_{\mathsf{ren}}(v_{\sigma})U(\sigma,\sigma_0)\;,
  \end{equation*}
  where $\hat{H}_{\mathsf{ren}}(v_{\sigma})$ is a self-adjoint operator on $\mathcal{N}_0$
  that depends on both $\sigma$ and $\sigma_0$, and is called the \emph{dressed
    renormalized Nelson Hamiltonian}\footnote{Both $U(\sigma,\sigma_0)$ and
    $\hat{H}_{\mathsf{ren}}(v_{\sigma})$ could be written explicitly, we postpone
    further details to the forthcoming sections.}. The following statements
  are true:

  \vspace{1.5mm}
  
  \begin{itemize}
    \setlength{\itemsep}{2mm}
  \item $\hat{H}_{\mathsf{ren}}(v_{\sigma})$ is self-adjoint on $\mathcal{N}_0$ and bounded
    from below uniformly with respect to both $0<\sigma\leq +\infty$ and $ \sigma_0\geq 0$, and it
    has a ground state, if either $\mu>0$ or $\mu=0$ and $\sigma_0=0$, whenever $V$
    satisfies BC. In particular, for $\mu= 0$, $\sigma_0=0$, and $V$ satisfying BC,
    $\hat{H}_{\mathsf{ren}}(v_{\infty})$ has a non-degenerate ground state.
  \item $\hat{H}_{\mathsf{ren}}(v_{\sigma})\xrightarrow[\sigma\to \infty]{\lVert \cdot
      \rVert_{}^{}-\mathsf{res}} \hat{H}_{\mathsf{ren}}(v_{\infty})$.
  \item If either $\mu>0$ or $\mu=0$ and $\sigma_0>0$, we define the UV renormalized
    Nelson Hamiltonian as the self-adjoint operator
    \begin{equation*}
      H_{\mathsf{ren}}(v_{\infty})\coloneqq  U^{*}(\infty,\sigma_0) \hat{H}_{\mathsf{ren}}(v_{\infty}) U(\infty,\sigma_0)\;.
    \end{equation*}
    In this case, we also have $H_{\mathsf{ren}}(v_{\sigma})\xrightarrow[\sigma\to \infty]{\lVert \cdot
      \rVert_{}^{}-\mathsf{res}} H_{\mathsf{ren}}(v_{\infty})$.
  \end{itemize}
\end{prop}

\subsection{Nelson wave function renormalization I: renormalized Hilbert
  space}
\label{sec:nelson-wave-function}

The wave function renormalization for the Nelson model is necessary only in
the massless case $\mu=0$, to properly define the non-Fock ground state
representation, curing the IR divergence; nonetheless, it can be carried out
always as a generalization of the dressing transformation $U(\sigma,\sigma_0)$.

To set up the wave function renormalization, we need to define the
annihilation operators $a(g_x)$, when $g\in \mathscr{T}'\smallsetminus L^2_k$ (recall that we have set
$\mathscr{T}= \mathscr{C}_0^{\infty}(\mathbb{R}^3_{*})$). Let us take advantage of the following well-known
isomorphism\footnote{The isomorphism $\mathcal{N}_0\cong L^2_x\bigl(\mathbb{R}^3,\mathcal{F}(L^2_k)\bigr)$ is
  implemented by a unitary map $\mathcal{U}: \mathcal{N}_0\to L^2_x\bigl(\mathbb{R}^3,\mathcal{F}(L^2_k)\bigr)$. We
  will use implicitly such map throughout the text, if no confusion arises.}:
\begin{equation*}
  \mathcal{N}_0=L^2_x\otimes \mathcal{F}(L^2_k)\cong L^2_x\bigl(\mathbb{R}^3,\mathcal{F}(L^2_k)\bigr)\;,
\end{equation*}
under which vectors $\Psi\in \mathcal{N}_0$ are identified with almost everywhere defined
square integrable functions $\mathbb{R}^3\ni x\mapsto \Psi(x)\in \mathcal{F}(L^2_k)$. Let us observe that for
any self-adjoint operator $T$ on $L^2_x$, the domain $\mathscr{D}(T\otimes \mathds{1})\subset \mathcal{N}_0$ is mapped
by the above isomorphism to a dense subset $\mathcal{U}\bigl(\mathscr{D}(T\otimes \mathds{1})\bigr)$ that we
often still denote by $\mathscr{D}(T\otimes \mathds{1})$ with a slight abuse of
notation. 
 Keeping that in mind, we are naturally led to
the following definition.
\begin{defn}[Singular annihilation evaluated at $x$]
  \label[definition]{def:5biss}
  Let $g\in \mathscr{T}'$ be a field configuration. Then define, for any $x\in \mathbb{R}^3$, the
  field configuration evaluated at $x$ as $\mathscr{T}'\ni g_x\coloneqq \e^{2\pi i k\cdot x}g$. The
  associated \emph{singular annihilation evaluated at $x$}, $a(g_x):
  L^2_x\bigl(\mathbb{R}^3,\FSfin(\mathscr{T})\bigr)\to L^2_x\bigl(\mathbb{R}^3,\FSfin(\mathscr{T})\bigr)$, is then
  defined by its action on 
  $\Psi_n(x)=
  \psi_1^{(n)}(x)\otimes_{\mathsf{s}}\dotsm\otimes_{\mathsf{s}}\psi_n^{(n)}(x)$, with $\psi_j^{(n)}(x)\in
  \mathscr{T}$ for a.e.\ $x\in \mathbb{R}^3$ and all $j=1,\dotsc,n$:
  \begin{equation*}
    \bigl(a(g_x)\Psi(x)\bigr)_n= \frac{\sqrt{n+1}}{(n+1)!} \sum_{\sigma\in \Sigma_{n+1}}^{}\langle \e^{2\pi i k\cdot x}g  , \psi_{\sigma(1)}^{(n+1)}(x) \rangle_{} \: \psi_{\sigma(2)}^{(n+1)}(x)\otimes\dotsm\otimes\psi^{(n+1)}_{\sigma(n+1)}(x)\;.
  \end{equation*}
\end{defn}
\begin{lem}
  \label[lemma]{lemma:13}
  Let $g\in \mathscr{T}'$. Then:
  \begin{itemize}
  \item  We have $a(g_x)\bigl( \sC_0^\infty(\IR^3,\FSfin(\sT))\bigr)\subset \sC_0^\infty(\IR^3,\FSfin(\sT))$;
  \item For any $\Psi\in L^2_x\bigl(\mathbb{R}^3,\FSfin(\mathscr{T})\bigr)$, there exists $k_{\Psi}\in
    \mathbb{N}$ such that $\bigl(a(g_x)\bigr)^k\Psi=0$ for any $k> k_{\Psi}$;
  \item The maps $g\mapsto\bigl(a(g_x)\bigr)^k\Psi$ are continuous w.r.t.\ the
    $\sigma(\sT',\sT)$ to Hilbert topology (on $\cN_0$) for all $\Psi\in
    L_x^2(\IR^3,\FSfin(\sT))$.
  \end{itemize}
\end{lem}
\begin{proof}
  The proof of this lemma follows directly from \cref{lemma:2}, observing
  that $x$-a.e., the above statements reduce to statements on $\FSfin(\mathscr{T})\subset
  \mathcal{F}(L^2_k)$.
\end{proof}
Thanks to the above, we can define the \emph{Nelson singular dressing} map
$\e^{a(g_x)}$ by (truncated) series expansion on
$L^2_x\bigl(\mathbb{R}^3,\FSfin(\mathscr{T})\bigr)$:
\begin{equation*}
  \e^{a(g_x)}\Psi= \sum_{k\in \mathbb{N}}^{}\frac{\bigl(a(g_x)\bigr)^k}{k!}\Psi\;.
\end{equation*}
By $x$-a.e.\ inspection of the fiber-wise components
$\bigl(\e^{a(g_x)}\Psi(x)\bigr)_n$, it is proved by \cref{prop:1} that such
dressing map is injective.
\begin{prop}
  \label[proposition]{prop:13}
  For any $g\in \mathscr{T}'$, the map $\e^{a(g_x)}: L^2_x\bigl(\mathbb{R}^3,\FSfin(\mathscr{T})\bigr)\to
  L^2_x\bigl(\mathbb{R}^3,\FSfin(\mathscr{T})\bigr)$ is injective.
\end{prop}
In complete analogy with the vHM and SB cases, we are thus led to the
definition of the renormalized Nelson inner product
\begin{equation*}
  \langle \Psi  , \Phi \rangle_{g_x}\coloneqq  \langle \e^{a(g_x)}\Psi  , \e^{a(g_x)}\Phi \rangle_{\mathcal{N}_0}
\end{equation*}
on $\mathcal{U}^{*}\bigl(L^2_x\bigl(\mathbb{R}^3,\FSfin(\mathscr{T})\bigr)\bigr)$, and the latter endowed with it
becomes a pre-Hilbert space.
\begin{defn}[Renormalized Nelson Space]
  \label[definition]{def:1}
  For any $g\in \mathscr{T}'$, the \emph{renormalized Nelson space} $\mathcal{N}_{g_x}$ is defined
  as
  \begin{equation*}
    \mathcal{N}_{g_x}= \overline{\mathcal{U}^{*}\bigl(L^2_x\bigl(\mathbb{R}^3,\FSfin(\mathscr{T})\bigr)\bigr)}^{\lVert \,\cdot \,  \rVert_{g_x}^{}}\;.
  \end{equation*}
\end{defn}
Once again, the renormalized inner product can be seen as the result of a
wave function renormalization procedure. Beforehand, let us remark that for
all $h\in L^2_k$, $\e^{a^{\dagger}(h_x)}$ is a well-defined operator on
$L^2_x\bigl(\mathbb{R}^3,\FSfin(\mathscr{T})\bigr)$.
\begin{prop}
  \label[proposition]{prop:14}
  Let $g\in \mathscr{T}'$ with $g_n\to g$ in the $\sigma(\mathscr{T}',\mathscr{T})$-topology (and $g_n\in \mathscr{T}$ for all
  $n\in \mathbb{N}$). Then for all $\Psi,\Phi\in
  \mathcal{U}^{*}\bigl(L^2_x\bigl(\mathbb{R}^3,\FSfin(\mathscr{T})\bigr)\bigr)$,
  \begin{equation*}
    \langle \Psi  , \Phi \rangle_{g_x}= \lim_{n\to \infty}\frac{\langle \e^{a^{\dagger}((g_n)_x)}\Psi  , \e^{a^{\dagger}((g_n)_x)}\Phi \rangle_{\mathcal{N}_0}}{\langle \e^{a^{\dagger}(g_n)}\Omega  , \e^{a^{\dagger}(g_n)}\Omega \rangle_{\mathcal{F}(L^2_k)}}\;.
  \end{equation*}
\end{prop}

\subsection*{The Nelson dressing maps}
\label{sec:nelson-dressing-maps}

In light of \cref{prop:12}, the singular dressing map $\e^{a(g_x)}$ -- seen as
a map between $\mathcal{N}_{g_x}$ and $\mathcal{N}_0$ -- is the only missing ingredient to
complete the IR sector of the Nelson renormalization picture: it plays the
exact role of the ill-defined $U(\infty,0)$ in the massless case $\mu=0$.

In order to prove that $\e^{a(g_x)}$ extends to a unitary map from $\mathcal{N}_{g_x}$
to $\mathcal{N}_0$, let us tensor the exponential vectors $\varepsilon(f)$, $f\in \mathscr{T}$, of
\cref{lemma:6} with an orthonormal basis of $L^2_x$. Let $\{e_j\}_{j\in \mathbb{N}}\subset
L^2_x$ be an orthonormal basis, and let $\varepsilon_n(f)\in \FSfin(\mathscr{T})$, $n\in \mathbb{N}$ and $f\in
\mathscr{T}$, be defined in \eqref{eq:7}; then we define $\mathbf{e}_{j,n,f}\in
L^2_x\bigl(\mathbb{R}^3,\FSfin(\mathscr{T})\bigr)$ by
\begin{equation*}
  \mathbf{e}_{j,n,f}(x)= e_j(x)\varepsilon_n(f)\;.
\end{equation*}

\begin{lem}
  \label[lemma]{lemma:6biss}
  For any $g\in \mathscr{T}'$ and $f\in \mathcal{T}$, the exponential vector\footnote{We are using
    implicitly the isomorphism $\mathcal{U}$ in this lemma: $\mathbf{e}_{j,f}= \sum_{n\in
      \mathbb{N}}^{} \mathcal{U}^{*}\mathbf{e}_{j,n,f}$; and \eqref{eq:6biss} holds in
    $\mathcal{U}(\mathcal{N}_0)=L^2_x\bigl(\mathbb{R}^3,\mathcal{F}(L^2_k)\bigr)$.}
  \begin{equation*}
    \mathbf{e}_{j,f}= \sum_{n\in \mathbb{N}}^{} \mathbf{e}_{j,n,f} 
  \end{equation*}
  is an absolutely convergent series in $\mathcal{N}_{g_x}$. Furthermore, in
  $\mathcal{N}_0$,
  \begin{equation}
    \label{eq:6biss}
    \bigl(\e^{a(g_x)}\mathbf{e}_{j,f}\bigr)(x)= \bigl(\e^{\langle g_x  , f \rangle_{}}e_j(x)\bigr)\varepsilon(f)\;,
  \end{equation}
  where $\e^{\langle g_x , f \rangle_{}}$ shall be seen as a bounded multiplication
  operator on $L^2_x$ (it is actually a bijection).
\end{lem}
\begin{proof}
  The action $\e^{a(g_x)}\mathbf{e}_{j,n,f}$ can be computed explicitly:
  \begin{equation*}
    \bigl(\e^{a(g_x)}\mathbf{e}_{j,n,f}\bigr)(x)= \sum_{k=0}^n \tfrac{1}{k!}\langle g_x  , f \rangle_{}^ke_j(x)\varepsilon_{n-k}(f)\;.
  \end{equation*}
  A resummation shows that
  \begin{equation*}
    \e^{a(g_x)}\sum_{n=0}^N \mathbf{e}_{j,n,f} = E_N(g_x,f) \sum_{n=0}^N\mathbf{e}_{j,n,f}\;,
  \end{equation*}
  where for all $x\in \mathbb{R}^3$,
  \begin{equation*}
    \mathbb{C}\ni E_N(g_x,f)= \sum_{k=0}^N \frac{\bigl(\langle g_x  , f \rangle_{}\bigr)^k}{k!}
  \end{equation*}
  is a truncated exponential, that converges absolutely as $N\to \infty$. Therefore,
  by monotone convergence, $\sum_{n=0}^N\mathbf{e}_{j,n,f}$ converges absolutely
  in $\mathcal{N}_{g_x}$ if and only if it converges absolutely in $\mathcal{N}_0$. In $\mathcal{N}_0$, $\langle
  \mathbf{e}_{j,n,f} , \mathbf{e}_{j,m,f} \rangle_{\mathcal{N}_0}= \frac{\delta_{nm}}{n!} \lVert f
  \rVert_2^{2n}$, and therefore $\mathbf{e}_{j,f}$ is an absolutely convergent
  series and
  \begin{equation*}
    \lVert \mathbf{e}_{j,f}  \rVert_{\mathcal{N}_0}^2= \e^{\lVert f  \rVert_2^2}\;.
    \qedhere
  \end{equation*}
\end{proof}
To avoid confusion, we might refer to the vector $\mathbf{e}_{j,f}$ on
$\mathcal{N}_{g_x}$ as $(\mathbf{e}_{j,f})_{g_x}$. We are now in the position to
densely define, the maps between Nelson spaces with different $g$s as
follows: for any $g,g'\in \mathscr{T}'$, define the isometries $U_{g_x,g'_x}: \mathcal{N}_{g_x}\to
\mathcal{N}_{g'_x}$ as the unique extension of
\begin{equation}
  \label{eq:ugNel}
  u_{g_x,g'_x}\Psi= \e^{a(g_x-g'_x)}\Psi\;,
\end{equation}
with $\Psi\in L^2_x(\mathbb{R}^3,\FSfin(\mathscr{T}))$.
\begin{prop}
  \label[proposition]{prop:unitarygroupnelson}
  The collection $(U_{g_x,g'_x})_{g,g'\in \mathscr{T}'}$ forms a two-parameter unitary
  group, namely the following properties hold true: for any $g,g',g''\in \mathscr{T}'$,
  \begin{itemize}
  \item $U_{g_x,g'_x}: \mathcal{N}_{g_x} \to \mathcal{N}_{g'_x}$ is unitary;
  \item $U_{g_x,g_x}= \mathds{1}$;
  \item $U_{g'_x,g_x}= U^{-1}_{g_x,g'_x}$;
  \item $U_{g'_x,g''_x}U_{g_x,g'_x}= U_{g_x,g''_x}$.
  \end{itemize}
\end{prop}
\begin{proof}
  In light of \cref{prop:unitarygroup,prop:unitarygroupspin}, we only
  discuss unitarity. By \cref{lemma:6biss},
  \begin{equation*}
    U_{g_x,g'_x}\mathbf{e}_{j,f} = \Bigl(\bigl(\e^{\langle g_x-g'_x  , f \rangle_{}}e_{j}(x)\bigr)\sum_{n\in \mathbb{N}}^{}\varepsilon_n(f)\Bigr)_{g'_x}\;;
  \end{equation*}
  where $(\,\cdot \,)_{g'_x}$ stands for the fact that we consider that as a vector
  in $\mathcal{N}_{g'_x}$. Taking $g'=0$ one gets unitarity of $U_{g_x,0}$, since
  $\{(\mathbf{e}_{j,f})_0 j\in \mathbb{N}\;,\; f\in \mathscr{T}\}$ is total in $\mathcal{N}_0$, and $\e^{\langle g_x ,
    f \rangle_{}}$ is a bijection on $L^2_x$. From this it also follows that
  $\{(\mathbf{e}_{j,f})_{g'_x}\;,\; j\in \mathbb{N}\;,\; f\in \mathscr{T}\}$ are total in $\mathcal{N}_{g'_x}$ for
  any $g'\in \mathscr{T}'$, and thus $U_{g_x,g'_x}$ is unitary as well.
\end{proof}
\cref{prop:4,lemma:5} are also provable, \emph{mutatis mutandis}, for
nontrivial Nelson spaces $\mathcal{N}_{g_x}$. Let us remark that
\cref{prop:4trisnelson} clarifies why the renormalized dressed Hamiltonian in
the massless case is commonly called ``non-Fock'' \citep[see,
\emph{e.g.}][]{MatteMoller.2018,hiroshima2022rmp}.
\begin{prop}
  \label[proposition]{prop:4trisnelson}
  Let $g,g'\in \mathscr{T}$; then the interacting fields and conjugate momenta on
  $\mathcal{N}_{g_x}$ and $\mathcal{N}_{g'_x}$ are related by:
  \begin{gather*}
    \phi_{g_x}(f)= U_{g'_x,g_x}\Bigl(\phi_{g'_x}(f)+ \sqrt{2}\Re \langle g_x-g'_x  , f \rangle\Bigr)U_{g_x,g'_x}\;,\\
    \pi_{g_x}(f)=U_{g'_x,g_x}\Bigl(\pi_{g'_x}(f)+ \sqrt{2}\Im \langle g_x-g'_x  , f \rangle\Bigr)U_{g_x,g'_x}\;.
  \end{gather*}
  
  It follows that $\mathcal{N}_{g_x}$ and $\mathcal{N}_{g'_x}$ carry unitarily equivalent
  representations of the interacting field's Weyl C*-algebra of canonical
  commutation relations if and only if $g-g'\in L^2_k$. In particular, the
  representation of interacting fields carried by $\mathcal{N}_{g_x}$ is disjoint from
  the Fock representation of free fields on $\mathcal{N}_0$ whenever $g\in \mathscr{T}'\smallsetminus L^2_k$.
\end{prop}
\begin{lem}
  \label[lemma]{lemma:5trisnelson}
  Let $g_1\in \mathscr{T}_1'$ and $g_2\in \mathscr{T}_2'$ be such that for all $f\in \mathscr{T}_1\cap \mathscr{T}_2$,
  $g_1(f)=g_2(f)$. Then there exists a natural unitary map $\Omega_{g_1,g_2}:
  \mathcal{N}_{(g_1)_x}\to \mathcal{N}_{(g_2)_x}$.
\end{lem}

\subsection{Nelson wave function renormalization II: renormalized
  Hamiltonian}
\label{sec:nelson-wave-function-1}

Let us now set $g_{\sigma}(k)\coloneqq  - \frac{v_{\sigma}(k)}{\omega(k)+ \lvert k \rvert_{}^2}$, with $v_{\sigma}$
the Nelson form factor given by \eqref{eq:3}, for $\sigma\leq +\infty$; then, the dressing
map $\e^{a((g_{\sigma})_x)}: \mathcal{N}_{g_x}\to \mathcal{N}_0$ (that is singular for $\mu=0$)
substitutes the unitary dressing $U(\sigma,0)$, with the advantage of being always
well-defined.

Let us recall that for all $\sigma<\infty$, we define the UV-regular renormalized
Nelson Hamiltonian to be
\begin{equation*}
  H_{\mathsf{ren}}(v_{\sigma})= H(v_{\sigma}) - E_{\sigma,\sigma_0}\;,
\end{equation*}
with the Nelson self-energy $E_{\sigma,\sigma_0}$ defined as in \cref{eq:selfenergy}. We then define the \emph{renormalized Nelson
  quadratic form} on $\mathcal{N}_{(g_{\sigma}-g_{\sigma_0})_x}$ as
\begin{equation*}
  q^{\mathsf{ren}}_{g_{\sigma}-g_{\sigma_0}}(\Psi,\Phi)\coloneqq  \frac{\langle \e^{a^{\dagger}((g_{\sigma}-g_{\sigma_0})_x)}\Psi  , H_{\mathsf{ren}}(v_\sigma)\:\e^{a^{\dagger}((g_{\sigma}-g_{\sigma_0})_x)}\Phi \rangle_{\mathcal{N}_0}}{\langle \e^{a^{\dagger}(g_{\sigma}-g_{\sigma_0})}\Omega  , \e^{a^{\dagger}(g_{\sigma}-g_{\sigma_0})}\Omega \rangle_{\mathcal{F}(L^2_k)}}\;,
\end{equation*}
well defined on
$U_{0,(g_{\sigma}-g_{\sigma_0})_x}\bigl(\sC_0^\infty(\IR^3,\FSfin(\sT))\bigr)$, as \cref{lemma:14} shows. Beforehand, let us
recall the definition of the dressed Nelson Hamiltonian
$\hat{H}_{\mathsf{ren}}(v_{\sigma})$ on $\mathcal{N}_0$:
\begin{multline}
  \label{eq:8}
  \hat{H}_{\mathsf{ren}}(v_{\sigma})\coloneqq  \bigl(-\Delta_x\overset{_{\cdot }}{+}V(x)\bigr)\otimes \mathds{1}+ 1\otimes \mathsf{d}\Gamma(\omega)\\+2i \Bigl(a^{\dagger}\bigl(k(g_{\sigma}-g_{\sigma_0})_x\bigr)\cdot \nabla_x+\nabla_x\cdot a\bigl(k(g_{\sigma}-g_{\sigma_0})_x\bigr)\Bigr)+2a^{\dagger}\bigl(k(g_{\sigma}-g_{\sigma_0})_x\bigr)\cdot a\bigl(k(g_{\sigma}-g_{\sigma_0})_x\bigr) \\+a^{\dagger}\bigl(k(g_{\sigma}-g_{\sigma_0})_x\bigr)\cdot a^{\dagger}\bigl(k(g_{\sigma}-g_{\sigma_0})_x\bigr)+ a\bigl(k(g_{\sigma}-g_{\sigma_0})_x\bigr)\cdot a\bigl(k(g_{\sigma}-g_{\sigma_0})_x\bigr)\;,
\end{multline}
whose spectral properties are also recalled in the next proposition.
\begin{prop}
  \label[proposition]{prop:15}
  For all $\sigma\leq +\infty$ and $\sigma_0\geq 0$, $\hat{H}_{\mathsf{ren}}(v_{\sigma})$ is
  essentially self-adjoint on $\sC_0^\infty(\IR^3,\FSfin(\sT))$, and bounded from below. Furthermore, it has a
  non-degenerate ground state $\Psi_{\mathsf{gs}}$ whenever $V$ satisfies BC and
  either $\mu>0$ or $\sigma_0=0$.
\end{prop}
\begin{proof}
	The claims are proved in \citep[][Thm.~1.7,~Prop.~2.3]{hiroshima2022rmp}.
\end{proof}
With the above definition of $\hat{H}_{\mathsf{ren}}(v_{\sigma})$, an explicit
computation, applying canonical commutation relations and the
Baker-Campbell-Hausdorff formula, yields the following lemma.
\begin{lem}
  \label[lemma]{lemma:14}
  For any $0\le \sigma_0\le\sigma\le\infty$ and any $\Psi,\Phi\in U_{0,(g_{\sigma}-g_{\sigma_0})_x}\bigl(\sC_0^\infty(\IR^3,\FSfin(\sT))\bigr)$, we have that
  \begin{equation*}
    q^{\mathsf{ren}}_{g_{\sigma}-g_{\sigma_0}}(\Psi,\Phi)= \langle \e^{a((g_{\sigma}-g_{\sigma_0})_x)}\Psi  , \hat{H}_{\mathsf{ren}}(v_{\sigma})\:\e^{a((g_{\sigma}-g_{\sigma_0})_x)}\Phi \rangle_{\mathcal{N}_0}\;,
  \end{equation*}
  and therefore it is the quadratic form associated to the essentially
  self-adjoint operator $H_{\mathsf{ren}}(v_{\sigma})_{g_{\sigma}-g_{\sigma_0}}$ on
  $\mathcal{N}_{(g_{\sigma}-g_{\sigma_0})_x}$, defined by
  \begin{equation*}
    H_{\mathsf{ren}}(v_{\sigma})_{g_{\sigma}-g_{\sigma_0}}\coloneqq  U_{0,(g_{\sigma}-g_{\sigma_0})_x}\,\hat{H}_{\mathsf{ren}}(v_{\sigma})\, U_{(g_{\sigma}-g_{\sigma_0})_x,0}\;.
  \end{equation*}
\end{lem}
We can now take the limit $\sigma\to \infty$ and/or $\sigma_0\to 0$ (in any order), and obtain
the UV/IR renormalized Nelson Hamiltonian.
\begin{thm}[Renormalization of the Nelson model]
  \label[theorem]{thm:1}
  The renormalized Nelson Hamiltonian has the following properties:
  \begin{itemize}
    \setlength{\itemsep}{3mm}

    \vspace{2mm}
    
  \item Let $\sigma<\infty$, and $\sigma_0>0$; then
    $H_{\mathsf{ren}}(v_{\sigma})_{g_{\sigma}-g_{\sigma_0}}$ converges to
    $H_{\mathsf{ren}}(v_{\sigma'})_{g_{\sigma'}-g_{\sigma'_0}}$, as $\sigma\to \sigma'$ and $\sigma_0\to \sigma_0'$
    for all $0\leq \sigma_0'<\sigma'\leq \infty$, in generalized norm resolvent sense, with the
    limits that can be taken in any order.

  \item In particular, for $\mu=0$, $H_{\mathsf{ren}}(v_{\sigma})_{g_{\sigma}-g_{\sigma_0}}$
    converges to $H_{\mathsf{ren}}(v_{\infty})_{g_{\infty}}$ as $\sigma\to \infty$ and $\sigma_0\to 0$,
    where $H_{\mathsf{ren}}(v_{\infty})_{g_{\infty}}$ is the renormalized Nelson
    Hamiltonian in its non-Fock ground state representation.

    
  \item In particular, $H_{\mathsf{ren}}(v_{\infty})_{g_{\infty}}= U_{0,(g_{\infty})_x}
    \,\hat{H}_{\mathsf{ren}}(v_{\infty})\, U_{(g_{\infty})_x,0}$, with $g_{\infty}=
    -v_{\infty}/(\omega+\lvert \cdot \rvert_{}^2)$, and therefore:
    \begin{itemize}
      \setlength{\itemsep}{1.5mm}

      \vspace{1mm}
      
    \item $H_{\mathsf{ren}}(v_{\infty})_{g_{\infty}}$ is self-adjoint on
      $\mathcal{N}_{(g_{\infty})_x}$ with domain
      $\mathscr{D}\bigl(H_{\mathsf{ren}}(v_{\infty})_{g_{\infty}}\bigr)=
      U_{0,(g_{\infty})_x}\mathscr{D}\bigl(\hat{H}_{\mathsf{ren}}(v_{\infty})\bigr)$;
    \item $\inf \sigma(H_{\mathsf{ren}}(v_{\infty})_{g_{\infty}})=\inf
      \sigma(\hat{H}_{\mathsf{ren}}(v_{\infty}))$;
    \item $H_{\mathsf{ren}}(v_{\infty})_{g_{\infty}}$ has a non-degenerate ground state
      whenever $V$ satisfies BC, and such ground state is
      $U_{(g_{\infty})_x,0}\Psi_{\mathsf{gs}}$, whence $\Psi_{\mathsf{gs}}$ is the
      ground state of $\hat{H}_{\mathsf{ren}}(v_{\infty})$.
    \end{itemize}
  \end{itemize}
\end{thm}
\begin{proof}
  The generalized norm resolvent convergence follows combining
  \cref{lemma:7,prop:12,lemma:14}, observing that the IR and UV cutoffs are
  completely independent. The spectral properties of the renormalized
  Hamiltonian follow from \cref{prop:15}.
\end{proof}

\subsection{The translation-invariant Nelson model}
Let us also briefly comment on the translation-invarant case, \emph{i.e.},\
$V=0$. The Nelson Hamiltonian of \cref{def:9} commutes with the total
momentum (vector) operator $-i\nabla_x + \dG(k)$, where $k$ is the one-boson
momentum operator acting on $L_k^2$. Thus, it decomposes into a direct
integral with respect to the spectrum of the total momentum, giving the
so-called {\em fiber Hamiltonians} $H(P,v)$ with $P\in\IR^3$ being the spectral
value of the total momentum.

\subsubsection*{Fock representation of the fiber Hamiltonians}
The fiber Hamiltonians of the translation-invariant Nelson model are defined
on the same Hilbert spaces as the van Hove model in
\cref{sec:van-hove-miyatake}, \emph{i.e.}, the Hilbert space of the fiber
Hamiltonians in the Fock representation is $\FS(\fh)$.  The free
translation-invariant Nelson fiber Hamiltonian is given by
\begin{align*}
	K_0(P) = (P-\dG(k))^2 + \dG(\omega)\;,
\end{align*}
where $k$ denotes the momentum operator on $L_k^2$.
For any $P\in\IR^3$, $H_0(P)$ is self-adjoint with domain $\sD(\dG(\omega))\cap\sD(|\dG(k)|^2)$, non-negative and has $\Omega$ as its unique ground state.
Again the interacting model is obtained by linearly coupling this with the field.
\begin{defn}[Fock Fiber Hamiltonian]\label{def:fiber}
	Let $\mu\ge0$ be the field's mass, let $v_\sigma\in L_k^2$ be the particle's charge distribution as defined in \cref{eq:3} and let $P\in\IR^3$ be the fixed total momentum. Then the Fock fiber Hamiltonian of the Nelson model is defined on $\FS(L_k^2)$ as
	\[K_\sigma(P) = K_0(P) + \ad(v_\sigma) + a(v_\sigma). \]
\end{defn}
\noindent Let us summarize some well-known facts.
\begin{prop} For all $\sigma\in[0,\infty)$, we have:
	\begin{itemize}
		\item $K_\sigma(P)$ is self-adjoint on $\sD(|\dG(k)|^2)\cap \sD(\dG(\omega))$ and bounded from below uniformly in $P$;
		\item if $\mu>0$, then $K_\sigma(P)$ has a unique ground state for all $P$ with $|P|<\frac{1}{2}$;
		\item if $\mu=0$ and $\sigma>0$, then $K_\sigma(P)$ has no ground state for any $P$.
	\end{itemize}
\end{prop}
\begin{proof}
	Self-adjointness follows by the Kato--Rellich theorem, \cite[cf.][Lemmas~3.1,\,3.2]{Dam.2018}. Existence of ground states (at arbitrary coupling) and $|P|<\frac{1}{2}$ for the massive model follows from \cite{Moller.2005}, since the model exhibits a spectral gap; also see \cite{HaslerHinrichsSiebert.2023} for a verification of M\o{}ller's assumption. Absence of ground states is proved in \cite[Thm.~3.3]{Dam.2018}.
\end{proof}
\noindent Integrating the fiber Hamiltonians over all total momenta $P$ gives
the Fock Nelson Hamiltonian with $V=0$.
\begin{prop}[Lee--Low--Pines Transformation, {\cite{LeeLowPines.1953}}]
	Define $U:\cN_0\to\cN_0$ to be the unitary given by
	\begin{align*}
		U = \bigl(\sF\otimes\Id\bigr)^* \e^{-i \dG(k)\cdot x},
	\end{align*}
	where $\sF$ denotes the Fourier transform on $L_x^2$.
	Then, for all $\Psi\in U\sD(H(v))$, we have $\Psi(P)\in \sD(K_\sigma(P))$ for almost all $P\in\IR^3$, and
	\begin{align*}
		\bigl(U H(v_\sigma)U^*\Psi\bigr)(P) = K_\sigma(P)\Psi(P).
	\end{align*}
\end{prop}
\noindent In terms of direct integrals:
\begin{align*}
	UH(v_\sigma)U^* = \int^\oplus_{\IR^3} K_\sigma(P)\rmd P\;.
\end{align*}
\subsubsection*{Self-Energy Renormalized Fibers}
Similar to the full Nelson model, the fiber Hamiltonians also allow for a
self-energy renormalization, as was first observed by J.T.\ Cannon
\cite{Cannon.1971}.  In fact, the strategy follows closely the one described
in \cref{sec:stand-nels-model}, thence we omit the details here, only
observing that (similar to the full Nelson model) the infrared properties do
not change after self-energy renormalization.
\begin{prop}
	Let $\sigma_0>0$ and $P\in\IR^3$.
	Then $K_\sigma(P)-E_{\sigma,\sigma_0}$ converges to a unique self-adjoint lower-bounded {\em renormalized Fock Nelson Hamiltonian} $K_{\infty}(P)$ in the norm resolvent sense as $\sigma\to\infty$, given explicitly by
	\[ K_{\infty}(P) = \e^{\ad(g_\infty-g_{\sigma_0})-a(g_\infty-g_{\sigma_0})}\hat K_{\infty,\sigma_0}(P)\e^{a(g_\infty-g_{\sigma_0})-\ad(g_\infty-g_{\sigma_0})}, \quad \sigma_0>0, \]
	where we define $\hat K_{\sigma,\sigma_0}(P)$ below.
	Furthermore:
	\begin{itemize}
		\item $K_\infty(P)$ is uniformly bounded from below in $P$;
		\item if $\mu>0$, then $K_{\infty}(P)$ has a unique ground state for all $P\in\IR^3$ with $|P|<\frac12$;
		\item if $\mu=0$, then $K_{\infty}(P)$ does not have a ground state for any $P\in\IR^3$;
		\item $UH_{\mathsf{ren}}(v_\infty)U^* = \int^\oplus_{\IR^3}K_\infty(P)\rmd P$.
	\end{itemize}
\end{prop}
\begin{proof}
	$K_\infty(P)$ was e.g. constructed in \cite{DamHinrichs.2021} and \cite{Hinrichs.2022b}, which also implies the direct integral decomposition of the full Nelson Hamiltonian. Absence of ground states for $\mu=0$ was also proven in \cite{DamHinrichs.2021}. Existence for $\mu>0$ easily follows, since the spectral gap is uniform in the UV cutoff, again see \cite{Moller.2005,HaslerHinrichsSiebert.2023}.
\end{proof}
\subsubsection*{Fiber-wise wave function renormalization}
Since we have already constructed the wave function renormalized Hilbert
spaces in \cref{sec:wave-funct-renorm}, let us now consider the wave function
renormalization for Nelson fiber Hamiltonians.  This is given by the {\em
  renormalized Nelson fiber quadratic form} on $\FS_{g_{\sigma}-g_{\sigma_0}}$, defined
as
\begin{align*}
	\fk^{\mathsf{ren}}_{P,g_\sigma-g_{\sigma_0}}(\psi,\phi) \coloneqq \frac{\Braket{\e^{\ad(g_{\sigma}-g_{\sigma_0})}\psi, \bigl(K_\sigma(P)-E_{\sigma,\sigma_0}\bigr) \e^{\ad(g_{\sigma}-g_{\sigma_0})} \phi }_{\FS(L_k^2)}}{\Braket{\e^{\ad(g_{\sigma}-g_{\sigma_0})}\Omega,\e^{\ad(g_\sigma-g_{\sigma_0})}\Omega}_{\FS(L_k^2)}}.
\end{align*}
In fact, it corresponds to the dressed fiber Hamiltonians $\hat K_{\sigma,\sigma_0}(P)$
on $\FS(L_k^2)$ for almost all $P\in\IR^3$, uniquely defined by
\begin{align*}
	U \hat H_{\mathsf{ren}}(v_\sigma) U^* = \int^\oplus_{\IR^3} \hat{K}_{\sigma,\sigma_0}(P),
\end{align*}
see for example \cite[App.~C]{DamHinrichs.2021}. Observe that in particular
this implies they are self-adjoint and bounded from below.  Similar to the
above, we now obtain:
\begin{prop}
	For $\psi,\phi\in\FSfin(\sT)$, we have that
	\begin{align*}
		\fk^{\mathsf{ren}}_{P,g_\sigma-g_{\sigma_0}}(\psi,\phi)
		= \Braket{e^{a(g_\sigma-g_{\sigma_0})}\psi,\hat K_{\sigma,\sigma_0}(P)e^{a(g_\sigma-g_{\sigma_0})}\phi}_{\FS_0}
	\end{align*}
	and therefore it is the quadratic form associated to the essentially self-adjoint operator $K_{\mathsf{ren}}(P)_{g_\sigma-g_{\sigma_0}}$ on $\FS_{g_\sigma-g_{\sigma_0}}$ defined by
	\begin{align*}
		K_{\mathsf{ren}}(P)_{g_\sigma-g_{\sigma_0}} = U_{0,g_\sigma-g_{\sigma_0}} \hat K_{\sigma,\sigma_0}(P)U_{g_\sigma-g_{\sigma_0},0}.
	\end{align*}
\end{prop}
Let us to conclude by returning back to the ground state problem.  The
following properties are known for the UV regular case, \emph{i.e.}, $\sigma<\infty$,
but shall carry over the to $\sigma=\infty$ case with no particular effort.
\begin{prop}
  Fix $\sigma\le\infty$. Then $K_{\mathsf{ren}}(P)_{g_\sigma-g_{\sigma_0}}$ has a ground state, if
  $P=0$ and $\sigma_0=0$ or $\lvert P \rvert_{}^{}<\frac12$ and $\mu>0$. If $P\ne 0$ and $\mu=0$, then
  $K_{\mathsf{ren}}(P)_{g_\sigma-g_{\sigma_0}}$ does not have a ground state for any
  $\sigma_0\in[0,\sigma]$.
\end{prop}
\begin{proof}
  Existence of a ground state at $\sigma_0=0$ for $P=0$ is, {\em e.g.}, proved in
  \cite{HaslerHinrichsSiebert.2023}, also see
  \cite{HaslerSiebert.2020}, at least for the case $\sigma<\infty$; the case $\sigma=\infty$ can be treated by minor adaptions of the strategy therein. Existence if $\mu>0$ is clear from the previous
  considerations. Absence at $P\ne0$ can be inferred along the lines of
  \cite{HaslerHerbst.2008a}, see \cite{Dam.2018,DamHinrichs.2021} for later
  refinements.
\end{proof}
The fact that the non-zero $P$ fibers do not exhibit a ground state, albeit
the model being wave function-renormalized also in the infrared sector,
corresponds to the observation by J.\ Fr\"ohlich \cite{Frohlich.1974} that
the infraparticle for different total momenta live in different CCR
representations. In our setting, this corresponds to choosing a $P$-dependent
dressing transformation in the construction of the renormalized fibers. Such
a construction has been first made rigorous by Pizzo \cite{Pizzo.2003} and
later refined in \cite{HaslerHinrichsSiebert.2023} in presence of an UV
cutoff and in \cite{BachmannDeckertPizzo.2012} without UV cutoff.  Our wave
function renormalization will also shed light on the nature of the non-Fock
infraparticles for different total momenta, but making their relation
precise, especially without UV cutoff, would require the introduction of a
variety of different techniques, which we defer to future work.

\renewcommand{\theHsection}{A\arabic{section}}
\appendix

\section{Double Operator Integrals via Quadratic Forms}
\label{app:DOI}

In this article, we use double operator integrals (DOI) of the form
\begin{align*}
	\iint f(\lambda,\mu)\sfd P_2(\lambda)T\sfd P_1(\mu) \in\cL(\HS)
\end{align*}
for projection-valued measures $P_1,P_2$ on a Hilbert space $\HS$,
$T\in\cL(\HS)$ and some scalar-valued measurable function $f$ to construct the
renormalized generalized spin-boson Hamiltonian.  Such double operator
integrals are studied in different contexts, starting with the seminal work
\cite{DaleckiiKrein.1956}.  Birman and Solomyak
\citep{BirmanSolomyak.1966,BirmanSolomyak.1967,BirmanSolomyak.1973} developed
an approach, in which a projection-valued measure $\mathcal G$ satisfying
$\cG(M_2\times M_1)T=P_2(M_2)TP_1(M_1)$ is defined on the Hilbert space of
Hilbert--Schmidt operators $T$, which can then be used to define the DOI by $\int
f \sfd \cG T$.  It is especially appealing about this approach that
$\sigma$-additivity of the canonically defined measure $\cG$ immediately follows
without additional assumptions on the underlying measurable space
\citep{BirmanSolomyak.1996}. To extend this definition to a broader class of
$T$, one approach is to restrict to essentially factorizable integrands $f$
and use an iterated definition of the DOI. This approach has, \emph{e.g.}, been
worked out in
\cite{dePagterWitvlietSukochev.2002,PotapovSukochev.2010,Nikitopoulos.2023,FerydouniSpiegel.2025}.
For more details and other previous approaches to the DOI, see for example
the textbook \cite{SkripkaTomskova.2019} or the reviews
\cite{BirmanSolomyak.2003,FerydouniSpiegel.2025}.

We here take a simple weak approach, in which we directly construct a quadratic form formally given by
\begin{align*}
	\fq_{T,f}(\psi,\ph) = \iint \braket{\sfd P_1(\lambda)\psi, T \sfd P_2(\mu)\ph}
\end{align*}
corresponding to the DOI. We then use this for the definition of the corresponding operator.

Let us define a notion of operators $T$ for which this type of construction is possible.
Throughout this appendix, let $(\sM_1,\Sigma_1)$, $(\sM_2,\Sigma_2)$ be measurable spaces,
let $\HS_1$, $\HS_2$ be Hilbert spaces and let $P_1,P_2$ be projection-valued measures defined on $(\sM_1,\Sigma_1)$ and $(\sM_2,\Sigma_2)$ with values in $\cL(\HS_1)$ and $\cL(\HS_2)$, respectively.
\begin{defn}
	\label[definition]{def:decomposable}
	We call an operator $T\in\cL(\HS_2,\HS_1)$ {\em $(P_1,P_2)$-decomposable} if for all $\psi_1\in\HS_1$, $\psi_2\in\HS_2$ there exists a complex measure $\mu_T(\psi_1,\psi_2)$ with bounded total variation norm $\|\mu_T(\psi_1,\psi_2)\|_{\sfT\sfV}$ on the product $\sigma$-algebra $\Sigma_1\otimes\Sigma_2$ such that
	\begin{align}
		\label{eq:productmeasure}
		&\mu_T(\psi_1,\psi_2)(M_1\times M_2) = \braket{P_1(M_1)\psi_1,TP_2(M_2)\psi_2}_{\HS_1}, \qquad M_1\in\Sigma_1,\ M_2\in\Sigma_2,
		\\
		\label{eq:TVnorm}
		&\|T\|_{(P_1,P_2)}\coloneqq\sup_{\|\psi_1\|\le1,\|\psi_2\|\le 1}\|\mu_T(\psi_1,\psi_2)\|_{\sfT\sfV} < \infty.
	\end{align}
\end{defn}
Let us first collect some properties of this definition.
\begin{prop}
	If a measure $\mu_T(\psi_1,\psi_2)$ satisfying \cref{eq:productmeasure} exists, then it is unique.
	Further, for all Borel functions $f_i:\sM_i\to \IC$ such that $\psi_i\in \sD\big(\int f_i \sfd P_i\big)$, $i=1,2$, and $M\in\fB(\sM_1\times\sM_2)$, we have
	\begin{align}
		\label{eq:PVMFubini}
		&\int_{\sM_1\times\sM_2} f_1(\lambda)f_2(\mu)\sfd \mu_T(\psi_1,\psi_2)(\lambda,\mu) = \Braket{\int_\IC f_1(\lambda)\sfd P_1(\lambda)\psi_1, T \int_\IC f_2(\mu)\sfd P_2(\mu)\psi_2 }_{\HS_1},
		\\&\label{eq:PVMdensity}
		\mu_T\Big(\int_\IC f_1(\lambda)\sfd P_1(\lambda)\psi_1, \int_\IC f_2(\mu)\sfd P_2(\mu)\psi_2\Big)(M) = \int_{M} f_1(\lambda)f_2(\mu)\d\mu_T(\psi_1,\psi_2).
	\end{align}
\end{prop}
\begin{proof}
  The uniqueness of $\mu_T$ is standard by the bounded total variation
  assumption \citep[cf.][Thm.\ A.1.5]{Durrett.2010}, since the rectangles
  generate the product $\sigma$-algebra.  We also note that
  \cref{eq:productmeasure} implies \cref{eq:PVMdensity} for $M$ being a
  rectangle, hence it again holds for arbitrary Borel sets by uniqueness
  since the density of the measure on the left hand side is uniquely given by
  the right hand side.  Clearly, \cref{eq:productmeasure} implies
  \cref{eq:PVMFubini} for simple functions, \emph{i.e.}, finite linear combinations
  of characteristic functions.  By standard measure-theoretic approximation
  arguments, \cref{eq:PVMFubini} can then be proven first for bounded Borel
  functions \citep[cf.][\textsection\,4.3.1]{Schmudgen.2012}, and subsequently for
  arbitrary Borel functions \citep[cf.][\textsection\,4.3.2]{Schmudgen.2012}.
\end{proof}
It thus remains to pose the question whether $\mu_T$ exists.
\begin{prop}We have that:
	\begin{itemize}
		\item The $(P_1,P_2)$-decomposable operators form a subspace of $\cL(\HS_2,\HS_1)$.
		\item Every Hilbert--Schmidt operator is $(P_1,P_2)$-decomposable and $\|T\|_{(P_1,P_2)} \le \|T\|_{\sfH\sfS}$.
		\item Assume
                  $(\sM_1,\Sigma_1,\HS_1,P_1)=(\sM_2,\Sigma_2,\HS_2,P_2)\coloneqq(\sM,\Sigma,\HS,P)$. Then,
                  if there is a measurable function $f:\sM\to\sM$ such that
                  $TP(M)=P\bigl(f^{-1}(M)\bigr)T$ for all $M\in\Sigma$, then $T$ is
                  $(P,P)$-decomposable with $\|T\|_{(P,P)}\le\|T\|_{\cL(\HS)}$. In
                  particular, this is true whenever
                  $(\sM,\Sigma,\HS,P)=(\IC,\fB(\IC),\HS,E_S)$ is the spectral
                  measure of a normal operator $S$ on $\HS$ such that
                  $TS=f(S)T$.
	\end{itemize}
\end{prop}
\begin{proof}
  The first claim is obvious from the fact that the complex measures with
  bounded total variation form a vector space and
  \cref{eq:productmeasure}. Furthermore, it is clear that
  \cref{eq:productmeasure} defines a finitely additive complex set function
  on the rectangles. Thus, in view of Carath\'eodory's extension theorem,
  the important part in the existence proof is the question whether one can
  also prove countable additivity. In the case that $T$ is Hilbert--Schmidt,
  this was proven by Birman and Solomyak \cite{BirmanSolomyak.1996}; and this
  also yields the upper bound on the total variation norm.
	
	It thus remains to treat the case $P_1=P_2$. Then, defining the measurable function $F:\sM\to\sM$ by $F(x)=(x,f(x))$, the pushforward measure $\mu_T(\psi_1,\psi_2)(M) = \braket{\psi_1,P(F^{-1}(M))T\psi_2}$ satisfies
	\begin{align*} \mu_T(\psi_1,\psi_2)(M_1\times M_2) &= \braket{\psi_1,P(F^{-1}(M_1\times M_2))T\psi_2} = \braket{\psi_1,P(M_1\cap f^{-1}(M_2)) T\psi_2} \\&= \braket{P(M_1)\psi_1,P(f^{-1}(M_2)) T\psi_2} = \braket{P(M_1)\psi_1,TP(M_2)\psi_2}\;, \end{align*}
	which proves the claim.
\end{proof}
\begin{rem}
  Even if $T$ is the identity, dropping the Hilbert--Schmidt assumption
  without further restrictions on $\sM_1$, $\sM_2$ will not in general
  suffice for countable additivity, as is known from
  \citep{BirmanVershikSolomyak.1979}. Also, in the case $P_1=P_2$, without
  any assumption on the commuting property of $T$ with $P$, the claim is not
  true in general, as can be seen by choosing $P(M)=\chr_M$ on $\HS=L^2(\IR)$
  with $(\cM,\Sigma)=(\IR,\fB(\IR))$ and picking $T$ to be the Fourier transform
  \citep[cf.][Ex.~3.1]{DragoMazzucchiMoretti.2020}.
\end{rem}
Integrating with respect to Borel measurable functions yields the quadratic form associated to the DOI.
\begin{prop}
	\label[proposition]{prop:DOI}
	For any $(P_1,P_2)$-decomposable $T\in\cL(\HS_2,\HS_1)$ and any bounded Borel function $f:\sM_1\times\sM_2\to \IC$, the integral
	\begin{align*}
		\fq_{T,f}(\psi_1,\psi_2) \coloneqq \int_{\sM_1\times\sM_2} f \sfd \mu_{T}(\psi_1,\psi_2)
	\end{align*}
	defines a bounded sesquilinear form $\fq_{T,f}:\HS_1\times\HS_2\to\IC$. In
        particular, there exists a bounded operator $S\in\cL(\HS_2,\HS_1)$ such
        that $\braket{\psi_1,S\psi_2}_{\HS_1} = \fq_{T,f}(\psi_1,\psi_2)$.
\end{prop}
\begin{proof}
	Clearly, the right hand side of \cref{eq:productmeasure} is sesquilinear.
	Thus, sesquilinearity follows, since linear combinations and limits of sesquilinear functions are again sesquilinear. Furthermore, the upper bound
	$|\fq_{T,f}(\psi_1,\psi_2)|\le \|T\|_{(P_1,P_2)}\|\psi_1\|_{\HS_1}\|\psi_2\|_{\HS_2}$
	follows from the definition.
	The existence of a corresponding bounded operator is thus a standard application of the Riesz representation theorem.
\end{proof}
\begin{defn}
	\label[definition]{def:DOI}
	In the main text of this article, we use DOI notation, \emph{i.e.}, in the setting of \cref{prop:DOI}, we denote
	\begin{align*}
		\fq_{T,f}(\psi,\ph) \eqqcolon \iint f(\lambda,\mu)\braket{\sfd P_1(\lambda)\psi, T \sfd P_2(\mu)\ph}
	\end{align*}
	and the corresponding bounded operator by
	\begin{align*}
		S\eqqcolon\iint f(\lambda,\mu) \sfd P_1(\lambda) T \sfd P_2(\mu).
	\end{align*}
\end{defn}


{\footnotesize \bibliographystyle{halpha-abbrv}
   \bibliography{bib,00lit,lit}

@article {AizenmanSimon.1982,
	AUTHOR = {Aizenman, Michael and Simon, Barry},
	TITLE = {Brownian Motion and {H}arnack Inequality for {S}chr\"{o}dinger	Operators},
	JOURNAL = {Commun. Pure Appl. Math.},
	FJOURNAL = {Communications on Pure and Applied Mathematics},
	VOLUME = {35},
	YEAR = {1982},
	NUMBER = {2},
	PAGES = {209-273},
	DOI = {10.1002/cpa.3160350206},
}

@unpublished{AlvarezLillLonigroValentinMartin.2025,
	author = {Benjamin Alvarez and Sascha Lill and Davide Lonigro and Valent\'in Mart\'in, Javier},
	title = {Renormalization of generalized spin--boson models with critical ultraviolet divergences},
	year = 2025,
	note = {Preprint},
	eprint = {2508.00803},
}

@article {Ammari.2000,
	AUTHOR = {Ammari, Zied},
	TITLE = {Asymptotic Completeness for a Renormalized Nonrelativistic	{H}amiltonian in Quantum Field Theory: The {N}elson Model},
	JOURNAL = {Math. Phys. Anal. Geom.},
	FJOURNAL = {Mathematical Physics, Analysis and Geometry},
	VOLUME = {3},
	YEAR = {2000},
	NUMBER = {3},
	PAGES = {217-285},
	DOI = {10.1023/A:1011408618527},
}

@article{Arai.2001,
	author = {Asao Arai},
	title = {Ground State of the Massless {N}elson Model Without Infrared Cutoff in a Non-{F}ock Representation},
	journal = {Rev. Math. Phys.},
	fjournal = {Reviews in Mathematical Physics},
	volume = 13,
	number = 9,
	pages = {1057-1094},
	year = 2001,
	doi = {10.1142/S0129055X01000934},
}

@article{Arai.2006,
	author = {Asao Arai},
	title = {Non-relativistic Limit of a {D}irac Polaron in Relativistic Quantum Electrodynamics},
	journal = {Lett. Math. Phys.},
	fjournal = {Letters in Mathematical Physics},
	volume = {77},
	number = {3},
	pages = {283-290},
	year = {2006},
	doi = {10.1007/s11005-006-0098-y},
}

@article{AraiHirokawa.1997,
	author = {Asao Arai and Masao Hirokawa},
	title = {On the Existence and Uniqueness of Ground States of a Generalized Spin-Boson Model},
	journal = {J. Funct. Anal.},
	fjournal = {Journal of Functional Analysis},
	volume = {151},
	number = 2,
	pages = {455-503},
	year = {1997},
	doi = {10.1006/jfan.1997.3140},
}

@article{AraiHirokawaHiroshima.1999,
	author = {Asao Arai and Masao Hirokawa and Fumio Hiroshima},
	title = {On the Absence of Eigenvectors of Hamiltonians	in a Class of Massless Quantum Field Models without Infrared Cutoff},
	journal = {J. Funct. Anal.},
	fjournal = {Journal of Functional Analysis},
	volume = {168},
	number = {2},
	pages = {470-497},
	year = {1999},
	doi = {10.1006/jfan.1999.3472},
}

@article{BachBallesterosKoenenbergMenrath.2017,
	author = {Volker Bach and Miguel Ballesteros and Martin K{\"o}nenberg and Lars Menrath},
	title = {Existence of ground state eigenvalues for the spin–boson model with critical infrared divergence and multiscale analysis},
	journal = {J. Math. Anal. Appl.},
	fjournal = {Journal of Mathematical Analysis and Applications},
	volume = 453,
	number = 2,
	pages = {773-797},
	year = 2017,
	doi = {10.1016/j.jmaa.2017.03.075},
	eprint = {1605.08348},
}

@article{BachmannDeckertPizzo.2012,
	author = {Sven Bachmann and Dirk-Andr\'{e} Deckert and Alessandro Pizzo},
	title = {The mass shell of the {N}elson model without cut-offs},
	journal = {J. Funct. Anal.},
	fjournal = {Journal of Functional Analysis},
	volume = 263,
	number = 5,
	pages = {1224-1282},
	year = 2012,
	doi = {10.1016/j.jfa.2012.04.021},
	eprint = {1104.3271},
}

@unpublished{BeaudDybalskiGraf.2021,
	author = {Vincent Beaud and Wojciech Dybalski and Gian Michele Graf},
	title = {Infraparticle States in the Massless Nelson Model: Revisited},
	journal = {Ann. Henri Poincar\'e},
	fjournal = {Annales Henri Poincar\'e},
	doi = {10.1007/s00023-022-01261-2},
	year = 2021,
	eprint = {2105.05723},
	note = {In Press},
}

@unpublished{BetzHinrichsKraftPolzer.2025,
	author = {Volker Betz and Benjamin Hinrichs and Mino Nicola Kraft and Steffen Polzer},
	title = {On the {I}sing Phase Transition in the Infrared-Divergent Spin Boson Model},
	note = {Preprint},
	year = 2025,
	eprint = {2501.19362},
}

@article {BetzPolzer.2022,
	AUTHOR = {Betz, Volker and Polzer, Steffen},
	TITLE = {A Functional Central Limit Theorem for Polaron Path Measures},
	JOURNAL = {Commun. Pure Appl. Math.},
	FJOURNAL = {Communications on Pure and Applied Mathematics},
	VOLUME = {75},
	YEAR = {2022},
	NUMBER = {11},
	PAGES = {2345-2392},
	doi = {10.1002/cpa.22080},
	eprint = {2106.06447},
}

@incollection {BirmanSolomyak.1966,
	AUTHOR = {Birman, M. {\v S}. and Solomyak, Mikhail Zakharovich},
	TITLE = {Double {S}tieltjes operator integrals},
	BOOKTITLE = {Problems of Mathematical Physics, No. 1: Spectral Theory
	and Wave Processes (Russian)},
	PAGES = {33-67},
	PUBLISHER = {Izdat. Leningrad. Univ., Leningrad},
	YEAR = {1966},
}

@incollection {BirmanSolomyak.1967,
	AUTHOR = {Birman, M. {\v S}. and Solomyak, Mikhail Zakharovich},
	TITLE = {Double {S}tieltjes operator integrals. {II}},
	BOOKTITLE = {Problems of Mathematical Physics, No. 2, Spectral
	Theory, Diffraction Problems (Russian)},
	PAGES = {26-60},
	PUBLISHER = {Izdat. Leningrad. Univ., Leningrad},
	YEAR = {1967},
}

@incollection {BirmanSolomyak.1973,
	AUTHOR = {Birman, M. {\v S}. and Solomyak, Mikhail Zakharovich},
	TITLE = {Double {S}tieltjes operator integrals. {III}},
	BOOKTITLE = {Problems of Mathematical Physics, {N}o. 6 ({R}ussian)},
	PAGES = {27-53},
	PUBLISHER = {Izdat. Leningrad. Univ., Leningrad},
	YEAR = {1973},
}

@article {BirmanSolomyak.1996,
	AUTHOR = {Birman, M. and Solomyak, M.},
	TITLE = {Tensor product of a finite number of spectral measures is always a spectral measure},
	JOURNAL = {Integr. Equ. Oper. Theory},
	FJOURNAL = {Integral Equations and Operator Theory},
	VOLUME = {24},
	YEAR = {1996},
	NUMBER = {2},
	PAGES = {179-187},
	DOI = {10.1007/BF01193459},
}

@article {BirmanSolomyak.2003,
	AUTHOR = {Birman, Mikhail Sh. and Solomyak, Michael},
	TITLE = {Double operator integrals in a {H}ilbert space},
	JOURNAL = {Integral Equ. Oper. Theory},
	FJOURNAL = {Integral Equations and Operator Theory},
	VOLUME = {47},
	YEAR = {2003},
	NUMBER = {2},
	PAGES = {131-168},
	DOI = {10.1007/s00020-003-1157-8},
}

@article {BirmanVershikSolomyak.1979,
	AUTHOR = {Birman, M. Sh. and Vershik, A. M. and Solomyak, M. Z.},
	TITLE = {Product of commuting spectral measures need not be countably additive},
	JOURNAL = {Funct. Anal. Appl.},
	VOLUME = {13},
	YEAR = {1979},
	PAGES = {48-49},
	doi = {10.1007/BF01076440},
}

@unpublished{Brooks.2024,
	author = {Morris Brooks},
	title = {Proof of the {L}andau-{P}ekar Formula for the effective Mass of the Polaron at strong coupling},
	year = 2024,
	eprint = {2409.08835},
	note = {Preprint},
}

@article{Cannon.1971,
	author = {John Cannon},
	title = {Quantum field theoretic properties of a model of {N}elson: Domain and eigenvector stability for perturbed linear operators},
	journal = {J. Funct. Anal.},
	fjournal = {Journal of Functional Analysis},
	volume = 8,
	number = 1,
	pages = {101-152},
	year = 1971,
	doi = {10.1016/0022-1236(71)90023-1},
}

@article {DaleckiiKrein.1956,
	AUTHOR = {Daleckii, Yu.\ L. and Krein, S. G.},
	TITLE = {Integration and differentiation of functions of {H}ermitian
	operators and applications to the theory of perturbations},
	JOURNAL = {Vorone{\v z}. Gos. Univ. Trudy Sem. Funkcional. Anal.},
	FJOURNAL = {Ministerstvo Vys\v sego Obrazovanija SSSR. Vorone\v zski\u i\
	Gosudarstvenny\u i\ Universitet. Trudy Seminara po
	Funkcional\cprime nomu Analizu},
	VOLUME = {1956},
	YEAR = {1956},
	NUMBER = {1},
	PAGES = {81-105},
}

@article{Dam.2018,
	author = {Thomas Norman Dam},
	title = {Absence of Ground States in the Translation Invariant Massless {N}elson Model },
	journal = {Ann. Henri Poincar\'e},
	fjournal = {Annales Henri Poincar\'e},
	volume = 21,
	number = 8,
	pages = {2655-2679},
	year = 2020,
	doi = {10.1007/s00023-020-00928-y},
	eprint = {1808.00088},
}

@article{DamHinrichs.2021,
	author = {Thomas Norman Dam and Benjamin Hinrichs},
	TITLE = {Absence of ground states in the renormalized massless translation-invariant {N}elson model},
	eprint = {1909.07661},
	JOURNAL = {Rev. Math. Phys.},
	FJOURNAL = {Reviews in Mathematical Physics},
	VOLUME = {34},
	YEAR = {2022},
	NUMBER = {10},
	PAGES = {2250033},
	doi = {10.1142/S0129055X22500337},
}

@article{DamMoller.2018a,
	author = {Thomas Norman Dam and Jacob Schach M\o{}ller},
	title = {Spin-Boson type models analysed using symmetries},
	journal = {Kyoto J. Math.},
	fjournal = {Kyoto Journal of Mathematics},
	volume = 60,
	number = 4,
	pages = {1261-1332},
	year = 2020,
	eprint = {1803.05812},
	doi = {10.1215/21562261-2019-0062},
}

@article{DamMoller.2018b,
	author = {Thomas Norman Dam and Jacob Schach M\o{}ller},
	title = {Asymptotics in Spin-Boson type models },
	journal = {Commun. Math. Phys.},
	fjournal = {Communications in Mathematical Physics},
	volume = 374,
	number = 3,
	pages = {1389-1415},
	year = 2020,
	eprint = {1808.00085},
	doi = {10.1007/s00220-020-03685-5},
}

@article {DeckertPizzo.2014,
	AUTHOR = {Deckert, Dirk-Andr\'{e} and Pizzo, Alessandro},
	TITLE = {Ultraviolet Properties of the Spinless, One-Particle {Y}ukawa
	Model},
	JOURNAL = {Commun. Math. Phys.},
	FJOURNAL = {Communications in Mathematical Physics},
	VOLUME = {327},
	YEAR = {2014},
	NUMBER = {3},
	PAGES = {887-920},
	DOI = {10.1007/s00220-013-1877-9},
	eprint = {1208.2646},
}

@article {dePagterWitvlietSukochev.2002,
	AUTHOR = {de Pagter, B. and Witvliet, H. and Sukochev, F. A.},
	TITLE = {Double Operator Integrals},
	JOURNAL = {J. Funct. Anal.},
	FJOURNAL = {Journal of Functional Analysis},
	VOLUME = {192},
	YEAR = {2002},
	NUMBER = {1},
	PAGES = {52-111},
	DOI = {10.1006/jfan.2001.3898},
}

@article{DerezinskiGerard.2004,
	author = {Jan Derezi\'{n}ski and Christian G{\'{e}}rard},
	title = {Scattering theory of infrared divergent {P}auli-{F}ierz {H}amiltonians},
	journal = {Ann. Henri Poincar\'e},
	fjournal = {Annales Henri Poincar\'e},
	volume = 5,
	number = 3,
	pages = {523-577},
	year = 2004,
	doi = {10.1007/s00023-004-0177-5},
}

@article {DragoMazzucchiMoretti.2020,
	AUTHOR = {Drago, Nicol{\`o{}} and Mazzucchi, Sonia and Moretti, Valter},
	TITLE = {An operational construction of the sum of two non-commuting
	observables in quantum theory and related constructions},
	JOURNAL = {Lett. Math. Phys.},
	FJOURNAL = {Letters in Mathematical Physics},
	VOLUME = {110},
	YEAR = {2020},
	NUMBER = {12},
	PAGES = {3197-3242},
	DOI = {10.1007/s11005-020-01332-7},
	eprint = {1909.10974},
}

@article{DybalskiPizzo.2014,
	AUTHOR = {Dybalski, Wojciech and Pizzo, Alessandro},
	TITLE = {Coulomb Scattering in the Massless {N}elson model {I}.	{F}oundations of Two-Electron Scattering},
	JOURNAL = {J. Stat. Phys.},
	FJOURNAL = {Journal of Statistical Physics},
	VOLUME = {154},
	YEAR = {2014},
	NUMBER = {1},
	PAGES = {543-587},
	DOI = {10.1007/s10955-013-0857-y},
	eprint = {1302.5001},
}

@article{DybalskiPizzo.2018,
	AUTHOR = {Dybalski, Wojciech and Pizzo, Alessandro},
	TITLE = {Coulomb Scattering in the Massless {N}elson Model {III}: Ground State Wave Functions and Non-Commutative Recurrence Relations},
	JOURNAL = {Ann. Henri Poincar\'{e}},
	FJOURNAL = {Annales Henri Poincar\'{e}},
	VOLUME = {19},
	YEAR = {2018},
	NUMBER = {2},
	PAGES = {463-514},
	DOI = {10.1007/s00023-017-0642-6},
	eprint = {1704.02924 },
}

@article {DybalskiPizzo.2019,
	AUTHOR = {Dybalski, Wojciech and Pizzo, Alessandro},
	TITLE = {Coulomb scattering in the massless {N}elson model {II}. {R}egularity of ground states},
	JOURNAL = {Rev. Math. Phys.},
	FJOURNAL = {Reviews in Mathematical Physics},
	VOLUME = {31},
	YEAR = {2019},
	NUMBER = {3},
	PAGES = {1950010},
	DOI = {10.1142/S0129055X19500107},
	eprint = {1302.5012},
}

@article {DybalskiPizzo.2022,
	AUTHOR = {Dybalski, Wojciech and Pizzo, Alessandro},
	TITLE = {Coulomb scattering in the massless {N}elson model {IV}. {A}tom-electron scattering},
	JOURNAL = {Rev. Math. Phys.},
	FJOURNAL = {Reviews in Mathematical Physics},
	VOLUME = {34},
	YEAR = {2022},
	NUMBER = {6},
	PAGES = {2250014},
	DOI = {10.1142/S0129055X22500143},
	eprint = {1902.08799},
}

@article {DybalskiSpohn.2020,
	AUTHOR = {Dybalski, Wojciech and Spohn, Herbert},
	TITLE = {Effective Mass of the Polaron -- Revisited},
	JOURNAL = {Ann. Henri Poincar\'{e}},
	FJOURNAL = {Annales Henri Poincar\'{e}},
	VOLUME = {21},
	YEAR = {2020},
	NUMBER = {5},
	PAGES = {1573-1594},
	DOI = {10.1007/s00023-020-00892-7},
	eprint = {1908.03432},
}

@article{EmeryLuther.1974,
	title = {Low-temperature properties of the {K}ondo {H}amiltonian},
	author = {Emery, Victor John and Luther, A.},
	journal = {Phys. Rev. B},
	fjournal = {Physical Reviews B},
	volume = {9},
	number = {1},
	pages = {215-226},
	year = {1974},
	doi = {10.1103/PhysRevB.9.215},
}

@article {Falconi.2015,
	AUTHOR = {Falconi, Marco},
	TITLE = {Self-Adjointness Criterion for Operators in {F}ock Spaces},
	JOURNAL = {Math. Phys. Anal. Geom.},
	FJOURNAL = {Mathematical Physics, Analysis and Geometry},
	VOLUME = {18},
	YEAR = {2015},
	NUMBER = {1},
	PAGES = {2},
	DOI = {10.1007/s11040-015-9173-x},
	eprint = {1405.6570},
}

@article{FannesNachtergaele.1988,
	author = {Mark Fannes  and Bruno Nachtergaele},
	title = {Translating the spin-boson model into a classical system},
	journal = {J. Math. Phys.},
	fjournal = {Journal of Mathematical Physics},
	volume = {29},
	number = {10},
	pages = {2288-2293},
	year = {1988},
	doi = {10.1063/1.528109},
}

@unpublished{FerydouniSpiegel.2025,
	author = {Robert Ferydouni and Daniel D. Spiegel},
	title = {Foundations of Double Operator Integrals with a Variant Approach to the Nonseparable Case},
	eprint = {2510.25789},
	year = 2025,
	note = {Preprint},
}

@article {FrankSeiringer.2021,
	AUTHOR = {Frank, Rupert L. and Seiringer, Robert},
	TITLE = {Quantum Corrections to the {P}ekar Asymptotics of a Strongly Coupled Polaron},
	JOURNAL = {Commun. Pure Appl. Math.},
	FJOURNAL = {Communications on Pure and Applied Mathematics},
	VOLUME = {74},
	YEAR = {2021},
	NUMBER = {3},
	PAGES = {544-588},
	DOI = {10.1002/cpa.21944},
	eprint = {1902.02489},
}

@article{Frohlich.1973,
	title = {On the infrared problem in a model of scalar electrons and massless scalar bosons},
	author = {J\"{u}rg Fr{\"{o}}hlich},
	journal = {Ann. de l'Inst. Henri Poincar\'{e}},
	fjournal = {Annales de l'Institut Henri Poincar\'{e}, physique th\'{e}orique},
	volume = 19,
	number = 1,
	pages = {1-103},
	year = 1973,
}

@article{Frohlich.1974,
	title = {Existence of dressed one-electron states in a class of persistent models},
	author = {J\"{u}rg Fr{\"{o}}hlich},
	journal = {Fortschr. Phys.},
	fjournal = {Fortschritte der Physik},
	volume = 22,
	number = 3,
	pages = {159-198},
	year = 1974,
	doi = {10.1002/prop.19740220304},
}

@article{Gerard.2000,
	author = {Christian G{\'{e}}rard},
	title = {On the Existence of Ground States for Massless {P}auli-{F}ierz {H}amiltonians},
	journal = {Ann. Henri Poincar\'e},
	fjournal = {Annales Henri Poincar\'e},
	volume = 1,
	number = 3,
	pages = {443-459},
	year = 2000,
	doi = {10.1007/s000230050002},
}

@article{GlimmJaffe.1970,
	author = {James Glimm and Arthur Jaffe},
	title = {Self-adjointness of the {Y}ukawa$_2$ {H}amiltonian},
	journal = {Ann. Phys.},
	fjournal = {Annals of Physics},
	volume = {60},
	number = {2},
	pages = {321-383},
	year = {1970},
	doi = {10.1016/0003-4916(70)90495-1},
}

@article {GlimmJaffe.1970b,
	AUTHOR = {Glimm, James and Jaffe, Arthur},
	TITLE = {The {$\lambda (\varphi\sp{4})\sb{2}$} quantum field theory without cutoffs: {II}. The field operators and the
	approximate vacuum},
	JOURNAL = {Ann. Math.},
	FJOURNAL = {Annals of Mathematics. Second Series},
	VOLUME = {91},
	YEAR = {1970},
	PAGES = {362-401},
	DOI = {10.2307/1970582},
}

@misc{GriesemerKussmaul.2026,
	title={On {R}ayleigh scattering in the massless {N}elson model}, 
	author={Marcel Griesemer and Valentin Kussmaul},
	year={2026},
	eprint={2512.21307},
	note = {Preprint},
}

@article{GriesemerWuensch.2018,
	author = {Marcel Griesemer and Andreas W{\"u}nsch},
	title = {On the domain of the {N}elson {H}amiltonian},
	journal = {J. Math. Phys.},
	fjournal = {Journal of Mathematical Physics},
	volume = {59},
	number = {4},
	pages = 042111,
	year = {2018},
	doi = {10.1063/1.5018579},
	eprint = {1711.10916},
}

@article {Gross.1973,
	AUTHOR = {Gross, Leonard},
	TITLE = {The relativistic polaron without cutoffs},
	JOURNAL = {Commun. Math. Phys.},
	FJOURNAL = {Communications in Mathematical Physics},
	VOLUME = {31},
	number = 1,
	YEAR = {1973},
	PAGES = {25-73},
	doi = {10.1007/BF01645589},
}

@article{GubinelliHiroshimaLorinczi.2014,
	author = {Massimiliano Gubinelli and Fumio Hiroshima and J\'{o}zsef L\H{o}rinczi},
	title = {Ultraviolet renormalization of the {N}elson {H}amiltonian through functional integration},
	journal = {J. Funct. Anal.},
	fjournal = {Journal of Functional Analysis},
	volume = 267,
	number = 9,
	pages = {3125-3153},
	year = 2014,
	doi = {10.1016/j.jfa.2014.08.002},
	eprint = {1304.6662},
}

@article{HaslerHerbst.2008a,
	author = {David Hasler and Ira Herbst},
	title = {Absence of Ground States for a Class of Translation Invariant Models of Non-relativistic QED},
	journal = {Commun. Math. Phys.},
	fjournal = {Communications in Mathematical Physics},
	volume = {279},
	number = {3},
	pages = {769-787},
	year = {2008},
	doi = {10.1007/s00220-008-0444-2},
	eprint = {math-ph/0702096},
}

@article{HaslerHerbst.2010,
	author = {David Hasler and Ira Herbst},
	title = {Ground States in the Spin Boson Model},
	journal = {Ann. Henri Poincar\'e},
	fjournal = {Annales Henri Poincar\'e},
	volume = {12},
	number = {4},
	pages = {621-677},
	year = {2011},
	doi = {10.1007/s00023-011-0091-6},
	eprint = {1003.5923},
}

@article {HaslerHerbst.2011b,
	author = {David Hasler and Ira Herbst},
	TITLE = {Smoothness and analyticity of perturbation expansions in qed},
	JOURNAL = {Adv. Math.},
	FJOURNAL = {Advances in Mathematics},
	VOLUME = {228},
	YEAR = {2011},
	NUMBER = {6},
	PAGES = {3249-3299},
	DOI = {10.1016/j.aim.2011.08.007},
	eprint = {1007.0969},
}

@article{HaslerHinrichsSiebert.2021a,
	author = {David Hasler and Benjamin Hinrichs and Oliver Siebert},
	title = {On Existence of Ground States in the Spin Boson Model},
	journal = {Commun. Math. Phys.},
	fjournal = {Communications in Mathematical Physics},
	doi = {10.1007/s00220-021-04185-w},
	volume = 388,
	number = 1,
	pages = {419-433},
	year = 2021,
	eprint = {2102.13373},
}

@article{HaslerHinrichsSiebert.2021b,
	author = {David Hasler and Benjamin Hinrichs and Oliver Siebert},
	title = {Correlation bound for a one-dimensional continuous long-range {I}sing model},
	journal = {Stoch. Proc. Appl.},
	fjournal = {Stochastic Processes and their Applications},
	volume = 146,
	pages = {60-79},
	year = 2022,
	eprint = {2104.03013},
	doi = {10.1016/j.spa.2021.12.010},
}

@article{HaslerHinrichsSiebert.2021c,
	author = {David Hasler and Benjamin Hinrichs and Oliver Siebert},
	title = {{FKN} Formula and Ground State Energy for the Spin Boson Model with External Magnetic Field},
	journal = {Ann. Henri Poincar\'e},
	FJOURNAL = {Annales Henri Poincar\'{e}},
	volume = 23,
	number = 8,
	PAGES = {2819-2853},
	year = 2022,
	eprint = {2106.08659},
	doi = {10.1007/s00023-022-01160-6},
}

@article{HaslerHinrichsSiebert.2023,
	author = {David Hasler and Benjamin Hinrichs and Oliver Siebert},
	title = {Non-{F}ock ground states in the translation-invariant {N}elson model revisited non-perturbatively},
	year = 2024,
	eprint = {2302.06998},
	JOURNAL = {J. Funct. Anal.},
	FJOURNAL = {Journal of Functional Analysis},
	VOLUME = {286},
	NUMBER = {7},
	PAGES = {110319},
	DOI = {10.1016/j.jfa.2024.110319},
}

@article{HaslerSiebert.2020,
	author = {David Hasler and Oliver Siebert},
	title = {Ground States for translationally invariant {P}auli-{F}ierz Models at zero Momentum},
	JOURNAL = {J. Funct. Anal.},
	FJOURNAL = {Journal of Functional Analysis},
	VOLUME = {284},
	YEAR = {2023},
	NUMBER = {1},
	PAGES = {109725},
	eprint = {2007.01250},
	doi = {10.1016/j.jfa.2022.109725},
}

@article{HidakaHiroshima.2016,
	author = {Takeru Hidaka and Fumio Hiroshima},
	title = {Spectrum of the semi-relativistic {P}auli-{F}ierz model I},
	journal = {J. Math. Anal. Appl.},
	fjournal = {Journal of Mathematical Analysis and Applications},
	volume = {437},
	number = {1},
	pages = {330-349},
	year = {2016},
	doi = {10.1016/j.jmaa.2015.11.081},
	eprint = {1402.1065},
}

@inproceedings{Hinrichs.2022b,
	author = {Benjamin Hinrichs},
	title = {Existence of Ground States in the Infrared-Critial Spin Boson Model},
	year = 2022,
	booktitle = {Mathematical Aspects of Quantum Fields and Related Topics},
	series = {RIMS K\^oky\^uroku},
	volume = {2235},
	pages = {60-72},
	eprint = {2204.00287},
}

@unpublished{HinrichsLampartValentinMartin.2024,
	author = {Hinrichs, Benjamin and Lampart, Jonas and Valent{\'i}n Mart{\'i}n, Javier},
	title = {Ultraviolet Renormalization of Spin Boson Models I. Normal and 2-Nilpotent Interactions},
	note = {Preprint},
	year = 2025,
	eprint = {2502.04876},
}

@article{HinrichsMatte.2022,
	author = {Benjamin Hinrichs and Oliver Matte},
	title = {{F}eynman--{K}ac Formula and Asymptotic Behavior of the Minimal Energy for the Relativistic {N}elson Model in Two Spatial Dimensions},
	journal = {Ann. Henri Poincar\'e},
    VOLUME = {25},
	YEAR = {2024},
	NUMBER = {6},
	PAGES = {2877--2940},
	eprint = {2211.14046},
	doi = {10.1007/s00023-023-01369-z}
}

@unpublished{HinrichsMatte.2024,
	author = {Benjamin Hinrichs and Oliver Matte},
	title = {{F}eynman--{K}ac formulas for semigroups generated by multi-polaron {H}amiltonians in magnetic fields and on general domains},
	year = 2024,
	eprint = {2403.12147},
	note = {Preprint},
}

@unpublished{HinrichsPolzer.2025,
	author = {Benjamin Hinrichs and Steffen Polzer},
	title = {Wiener-Type Theorems for the Laplace Transform. With Applications to Ground State Problems},
	year = 2025,
	eprint = {2511.02867},
	note = {Preprint},
}

@article{HirokawaHiroshimaLorinczi.2014,
	author = {Masao Hirokawa and Fumio Hiroshima and J\'{o}zsef L\H{o}rinczi},
	title = {Spin-boson model through a {P}oisson-driven stochastic process},
	journal = {Math. Z.},
	fjournal = {Mathematische Zeitschrift},
	volume = 277,
	number = 3,
	pages = {1165-1198},
	year = 2014,
	doi = {10.1007/s00209-014-1299-1},
	eprint = {1209.5521},	
}

@article {Hiroshima.2000,
	AUTHOR = {Hiroshima, Fumio},
	TITLE = {Essential Self-Adjointness of Translation-Invariant Quantum Field Models for Arbitrary Coupling Constants},
	JOURNAL = {Commun. Math. Phys.},
	FJOURNAL = {Communications in Mathematical Physics},
	VOLUME = {211},
	YEAR = {2000},
	NUMBER = {3},
	PAGES = {585-613},
	DOI = {10.1007/s002200050827},
}

@article{LampartSchmidt.2019,
	author = {Jonas Lampart and Julian Schmidt},
	title = {On {N}elson-Type {H}amiltonians and Abstract Boundary Conditions},
	journal = {Commun. Math. Phys.},
	fjournal = {Communications in Mathematical Physics},
	volume = {367},
	number = {2},
	pages={629-663},
	year = 2019,
	doi = {10.1007/s00220-019-03294-x},
	eprint = {1803.00872},
}

@article{LeeLowPines.1953,
	author = {Tsung-Dao Lee and Francis Eugene Low and David Pines},
	title = {The motion of slow electrons in a polar crystal},
	journal = {Phys. Rev.},
	volume = 90,
	number = 2,
	pages = {297-302},
	year = 1953,
	doi = {10.1103/PhysRev.90.297},
}

@article{Leggettetal.1987,
	author = {A. J. Leggett and S. Chakravarty and A. T. Dorsey and Matthew P. A. Fisher and Anupam Garg and W. Zwerger},
	title = {Dynamics of the dissipative two-state system},
	journal = {Rev. Mod. Phys.},
	fjournal = {Reviews of Modern Physics},
	volume = 59,
	number = 1,
	pages = {1-85},
	year = 1987,
}

@article {LeopoldMitrouskasRademacherSchleinSeiringer.2021,
	AUTHOR = {Leopold, Nikolai and Mitrouskas, David and Rademacher, Simone and Schlein, Benjamin and Seiringer, Robert},
	TITLE = {Landau---{P}ekar equations and quantum fluctuations for the dynamics of a strongly coupled polaron},
	JOURNAL = {Pure Appl. Anal.},
	FJOURNAL = {Pure and Applied Analysis},
	VOLUME = {3},
	YEAR = {2021},
	NUMBER = {4},
	PAGES = {653-676},
	DOI = {10.2140/paa.2021.3.653},
	eprint = {2005.02098},
}

@unpublished{LillLonigro.2024,
	title = {Self-adjointness and domain of generalized spin-boson models with mild ultraviolet divergences},
	author = {Lill, Sascha and Lonigro, Davide},
	eprint = {2307.14727},
	note = {Preprint},
	year = 2023,
}

@article {Lonigro.2022,
	AUTHOR = {Lonigro, Davide},
	TITLE = {Generalized spin-boson models with non-normalizable form factors},
	JOURNAL = {J. Math. Phys.},
	FJOURNAL = {Journal of Mathematical Physics},
	VOLUME = {63},
	YEAR = {2022},
	NUMBER = {7},
	PAGES = {072105},
	DOI = {10.1063/5.0085576},
	eprint = {2111.06121},
}

@article{LorincziMinlosSpohn.2002,
	author = {J\'{o}zsef L\H{o}rinczi and Robert Adol'fovich Minlos and Herbert Spohn},
	title = {The Infrared Behaviour in {N}elson's Model of a Quantum Particle Coupled to a Massless Scalar Field},
	journal = {Ann. Henri Poincar\'{e}},
	fjournal = {Annales Henri Poincar\'e},
	volume = 3,
	number = 2,
	pages = {269-295},
	year = 2002,
	doi = {10.1007/s00023-002-8617-6},
	eprint = {math-ph/0011043},
}

@article{Matte.2017,
	title={{P}auli-{F}ierz Type Operators with Singular Electromagnetic Potentials on General Domains},
	author={Matte, Oliver},
	journal={Math. Phys. Anal. Geom.},
	fjournal = {Mathematical Physics, Analysis and Geometry},
	volume={20},
	number={2},
	pages={18},
	year={2017},
	doi = {10.1007/s11040-017-9249-x},
	eprint = {1703.00404},
}

@article{MatteMoller.2018,
	title = {{F}eynman-{K}ac Formulas for the Ultra-Violet Renormalized {N}elson Model},
	author = {Oliver Matte and Jacob Schach M\o{}ller},
	journal = {Ast\'{e}risque},
	volume = 404,
	year = 2018,
	doi = {10.24033/ast.1054},
	eprint = {1701.02600},
}

@article{Moller.2005,
	AUTHOR = {Jacob Schach M{\o{}}ller},
	TITLE = {The Translation Invariant Massive {N}elson Model: {I}. {T}he Bottom of the Spectrum},
	JOURNAL = {Ann. Henri Poincar\'{e}},
	FJOURNAL = {Annales Henri Poincar\'{e}},
	VOLUME = {6},
	YEAR = {2005},
	NUMBER = {6},
	PAGES = {1091-1135},
	DOI = {10.1007/s00023-005-0234-8},
}

@article{Moller.2006b,
	AUTHOR = {M{\o{}}ller, Jacob Schach},
	TITLE = {The Polaron Revisited},
	JOURNAL = {Rev. Math. Phys.},
	FJOURNAL = {Reviews in Mathematical Physics},
	VOLUME = {18},
	YEAR = {2006},
	NUMBER = {5},
	PAGES = {485-517},
	DOI = {10.1142/S0129055X0600267X},
}

@article {MukherjeeVaradhan.2020,
	AUTHOR = {Mukherjee, Chiranjib and Varadhan, Sathamangalam Rangaiyenga Srininvasa},
	TITLE = {Identification of the Polaron Measure {I}: Fixed Coupling Regime and the Central Limit Theorem for Large Times},
	JOURNAL = {Commun. Pure Appl. Math.},
	FJOURNAL = {Communications on Pure and Applied Mathematics},
	VOLUME = {73},
	YEAR = {2020},
	NUMBER = {2},
	PAGES = {350-383},
	DOI = {10.1002/cpa.21858},
	eprint = {1802.05696},
}

@article {MukherjeeVaradhan.2020b,
	AUTHOR = {Mukherjee, Chiranjib and Varadhan, S. R. S.},
	TITLE = {Identification of the Polaron measure in strong coupling and the {P}ekar variational formula},
	JOURNAL = {Ann. Probab.},
	FJOURNAL = {The Annals of Probability},
	VOLUME = {48},
	YEAR = {2020},
	NUMBER = {5},
	PAGES = {2119-2144},
	DOI = {10.1214/19-AOP1392},
	eprint = {1812.06927},
}

@article {Nikitopoulos.2023,
	AUTHOR = {Nikitopoulos, Evangelos A.},
	TITLE = {Multiple operator integrals in non-separable von {N}eumann	algebras},
	JOURNAL = {J. Oper. Theory},
	FJOURNAL = {Journal of Operator Theory},
	VOLUME = {89},
	YEAR = {2023},
	NUMBER = {2},
	PAGES = {361-427},
	doi = {10.7900/jot.2021aug19.2357},
	eprint = {2107.03687},
}

@article{OtterpohlNalbachThorwart.2022,
	author = {Florian Otterphl and Peter Nalbach and Michael Thorwart},
	title = {Hidden Phase of the Spin-Boson Model},
	journal = {Phys. Rev. Lett.},
	fjournal = {Physical Review Letters},
	volume = 129,
	pages = 120406,
	year = 2022,
	doi = {10.1103/PhysRevLett.129.120406},
	eprint = {2208.10313},
}

@article {Panati.2009,
	AUTHOR = {Panati, Annalisa},
	TITLE = {Existence and nonexistence of a ground state for the massless {N}elson model under binding condition},
	JOURNAL = {Rep. Math. Phys.},
	FJOURNAL = {Reports on Mathematical Physics},
	VOLUME = {63},
	YEAR = {2009},
	NUMBER = {3},
	PAGES = {305-330},
	DOI = {10.1016/S0034-4877(09)00014-7},
	eprint = {math-ph/0609065},
}

@article{Pizzo.2003,
	author = {Alessandro Pizzo},
	title = {One-particle (improper) States in {N}elson's Massless Model},
	journal = {Ann. Henri Poincar\'{e}},
	fjournal = {Annales Henri Poincar\'e},
	volume = 4,
	number = 3,
	pages = {439-486},
	year = 2003,
	doi = {10.1007/s00023-003-0136-6},
}

@article{Pizzo.2005,
	AUTHOR = {Pizzo, Alessandro},
	TITLE = {Scattering of an {\it {I}nfraparticle}: The One Particle Sector in {N}elson's Massless Model},
	JOURNAL = {Ann. Henri Poincar\'{e}},
	FJOURNAL = {Annales Henri Poincar\'{e}},
	VOLUME = {6},
	YEAR = {2005},
	NUMBER = {3},
	PAGES = {553-606},
	DOI = {10.1007/s00023-005-0216-x},
}

@article {PotapovSukochev.2010,
	AUTHOR = {Potapov, D. and Sukochev, F.},
	TITLE = {Double Operator Integrals and Submajorization},
	JOURNAL = {Math. Model. Nat. Phenom.},
	FJOURNAL = {Mathematical Modelling of Natural Phenomena},
	VOLUME = {5},
	YEAR = {2010},
	NUMBER = {4},
	PAGES = {317-339},
	DOI = {10.1051/mmnp/20105414},
}

@article {RomanRocheSanchezZueco,
    AUTHOR = {Rom{\'a}n-Roche, Juan and S{\'a}nchez-Burillo, Eduardo and Zueco,
              David},
     TITLE = {Bound states in ultrastrong waveguide {QED}},
   JOURNAL = {Phys. Rev. A},
  FJOURNAL = {Physical Review A},
    VOLUME = {102},
      YEAR = {2020},
    NUMBER = {2},
     PAGES = {023702},
   doi = {10.1103/PhysRevA.102.023702},
	eprint = {2001.07643},
}

@article {Sasaki.2005,
	AUTHOR = {Sasaki, Itaru},
	TITLE = {Ground state of the massless {N}elson model in a non-{F}ock
	representation},
	JOURNAL = {J. Math. Phys.},
	FJOURNAL = {Journal of Mathematical Physics},
	VOLUME = {46},
	YEAR = {2005},
	NUMBER = {10},
	PAGES = {102107},
	DOI = {10.1063/1.2050507},
}

@article{Schmidt.2019,
	author = {Julian Schmidt},
	title = {On a direct description of pseudorelativistic {N}elson {H}amiltonians},
	journal = {J. Math. Phys.},
	fjournal = {Journal of Mathematical Physics},
	volume = {60},
	number = {10},
	pages = {102303},
	year = {2019},
	doi = {10.1063/1.5109640},
	eprint = {1810.03313},
}

@article{Sharafievetal.2021,
	doi = {10.22331/q-2021-06-10-474},
	title = {Visualizing the emission of a single photon with frequency and time resolved spectroscopy},
	author = {Sharafiev, Aleksei and Juan, Mathieu L. and Gargiulo, Oscar and Zanner, Maximilian and W{\"{o}}gerer, Stephanie and Garc{\'{i}}a-Ripoll, Juan Jos{\'{e}} and Kirchmair, Gerhard},
	journal = {Quantum},
	volume = {5},
	pages = {474},
	year = {2021},
	eprint = {2001.09737},
}

@article{Simon.1978,
	author = {Barry Simon},
	title = {Lower semicontinuhy of positive quadratic forms},
	journal = {Proceedings of the Royal Society of Edinburgh: Section A Mathematics},
	volume = 79,
	number = {3-4},
	pages = {267-273},
	year = 1978,
	doi = {10.1017/S0308210500019776},
	publisher={Royal Society of Edinburgh Scotland Foundation},
}

@article {Sloan.1974,
	AUTHOR = {Sloan, Alan David},
	TITLE = {The polaron without cutoffs in two space dimensions},
	JOURNAL = {J. Math. Phys.},
	FJOURNAL = {Journal of Mathematical Physics},
	VOLUME = {15},
	YEAR = {1974},
	PAGES = {190-201},
	DOI = {10.1063/1.1666620},
}

@article {Spohn.1987,
	AUTHOR = {Spohn, Herbert},
	TITLE = {Effective mass of the polaron: A functional integral approach},
	JOURNAL = {Ann. Phys.},
	FJOURNAL = {Annals of Physics},
	VOLUME = {175},
	YEAR = {1987},
	NUMBER = {2},
	PAGES = {278-318},
	DOI = {10.1016/0003-4916(87)90211-9},
}

@article{Spohn.1989,
	author = {Herbert Spohn},
	title = {Ground State(s) of the Spin-Boson {H}amiltonian},
	journal = {Commun. Math. Phys.},
	fjournal = {Communications in Mathematical Physics},
	volume = 123,
	number = 2,
	pages = {277-304},
	year = 1989,
	doi = {10.1007/BF01238859},
}

@article{SpohnDuemcke.1985,
	author = {Herbert Spohn and Rolf D\"umcke},
	title = {Quantum tunneling with dissipation and the {I}sing model over \ifx\mathbb\undefined{$R$}\else{$\mathbb R$}\fi},
	journal = {J. Stat. Phys.},
	fjournal = {Journal of Statistical Physics},
	volume = 41,
	number = 3,
	pages = {389-423},
	year = 1985,
	doi = {10.1007/BF01009015},
}

@incollection {TeufelTumulka.2016,
	AUTHOR = {Teufel, Stefan and Tumulka, Roderich},
	TITLE = {Avoiding Ultraviolet Divergence by Means of Interior-Boundary Conditions},
	BOOKTITLE = {Quantum Mathematical Physics},
	editor = {Felix Finster and Johannes Kleiner and Christian R\"oken and Jürgen Tolksdorf},
	PAGES = {293-311},
	PUBLISHER = {Birkh\"auser},
	address = {Cham},
	YEAR = {2016},
	doi = {10.1007/978-3-319-26902-3\_14},
	eprint = {1506.00497},
}

@article{WeisskopfWigner.1930,
	author = {Weisskopf, Victor and Wigner, Eugene},
	title = {Berechnung der natürlichen Linienbreite auf Grund der Diracschen Lichttheorie},
	journal = {Z. Phys.},
	fjournal = {Zeitschrift f\"ur Physik},
	volume = 63,
	number = 1,
	pages = {54-73},
	year = {1930},
	doi = {10.1007/BF01336768},
}

@article{YanShao.2023,
	title = {Exact dynamics of the spin-boson model at the Toulouse limit},
	author = {Yan, Yun-An and Shao, Jiushu},
	journal = {Phys. Rev. A},
	volume = {108},
	number = {1},
	pages = {012218},
	year = {2023},
	doi = {10.1103/PhysRevA.108.012218},
}

@book{Arai.2018,
	author = {Asao Arai},
	title = {Analysis on {F}ock Spaces and Mathematical Theory of Quantum Fields},
	subtitle = {An Introduction to Mathematical Analysis of Quantum Fields},
	publisher = {World Scientific},
	address = {New Jersey},
	year = {2018},
	DOI = {10.1142/10367},
}

@book {Durrett.2010,
	AUTHOR = {Durrett, Rick},
	TITLE = {Probability: Theory and Examples},
	SERIES = {Cambridge Series in Statistical and Probabilistic Mathematics},
	VOLUME = {31},
	EDITION = {4th},
	PUBLISHER = {Cambridge University Press},
	address = {Cambridge},
	YEAR = {2010},
	DOI = {10.1017/CBO9780511779398},
}

@book {GlimmJaffe.1985,
	AUTHOR = {Glimm, James and Jaffe, Arthur},
	TITLE = {Collected papers. Constructive Quantum Field Theory Selected Papers},
	NOTE = {Reprint of articles published 1968--1980},
	PUBLISHER = {Birkh\"{a}user},
	address = {Boston},
	YEAR = {1985},
}

@book{ReedSimon.1975,
	author = {Michael Reed and Barry Simon},
	title = {{F}ourier Analysis, Self-Adjointness},
	year = {1975},
	series = {Methods of Modern Mathematical Physics},
	volume = {2},
	publisher = {Academic Press},
	address = {San Diego},
}

@book{Schmudgen.2012,
	author = {Konrad Schm\"{u}dgen},
	title = {Unbounded Self-adjoint Operators on {H}ilbert Space},
	year = 2012,
	publisher = {Springer},
	address = {Dordrecht},
	series = {Graduate Texts in Mathematics},
	volume = 265,
      DOI = {10.1007/978-94-007-4753-1},
}

@book {SkripkaTomskova.2019,
	AUTHOR = {Skripka, Anna and Tomskova, Anna},
	TITLE = {Multilinear Operator Integrals},
	SERIES = {Lecture Notes in Mathematics},
	VOLUME = {2250},
	PUBLISHER = {Springer},
	address = {Cham},
	YEAR = {2019},
	DOI = {10.1007/978-3-030-32406-3},
}

@article{abdesselam2012cmp,
  Author =	 {Abdesselam, Abdelmalek and Hasler, David},
  opta =	 {A. Abdesselam, D. Hasler},
  Coden =	 {CMPHAY},
  Doi =		 {10.1007/s00220-011-1407-6},
  Fjournal =	 {Communications in Mathematical Physics},
  Issn =	 {0010-3616},
  Journal =	 {Comm. Math. Phys.},
  Mrclass =	 {81T10 (81Q10)},
  Mrnumber =	 {2890307},
  Mrreviewer =	 {Asao Arai},
  Number =	 {2},
  Pages =	 {511--536},
  Title =	 {Analyticity of the ground state energy for massless {N}elson models},
  optt =	 {Analyticity of the ground state energy for massless Nelson models},
  Url =		 {http://dx.doi.org/10.1007/s00220-011-1407-6},
  Volume =	 {310},
  Year =	 {2012},
  OPTkey = 	 {},
  OPTmonth = 	 {},
  OPTnote = 	 {},
  OPTannote = 	 {}
}

@article {aizenman2021am,
  AUTHOR =	 {Aizenman, Michael and Duminil-Copin, Hugo},
  TITLE =	 {Marginal triviality of the scaling limits of critical 4{D} {I}sing and {$\phi_4^4$} models},
  JOURNAL =	 {Ann. of Math. (2)},
  FJOURNAL =	 {Annals of Mathematics. Second Series},
  VOLUME =	 {194},
  YEAR =	 {2021},
  NUMBER =	 {1},
  PAGES =	 {163--235},
  ISSN =	 {0003-486X},
  MRCLASS =	 {82B20 (60G60 82B27)},
  MRNUMBER =	 {4276286},
  DOI =		 {10.4007/annals.2021.194.1.3},
  URL =		 {https://doi.org/10.4007/annals.2021.194.1.3},
  opta =	 {M. Aizenman, H. Duminil-Copin},
  optt =	 {Marginal triviality of the scaling limits of critical 4D Ising and Φ⁴₄ models},
  OPTkey =	 {},
  OPTFile =	 {},
  OPTmonth =	 {},
  OPTnote =	 {},
  OPTarchivePrefix ={},
  OPTeprint =	 {},
  OPTprimaryclass ={},
  OPTannote =	 {},
}

@article{ammari2004jfa,
  Author =	 {Ammari, Zied},
  opta =	 {Z. Ammari},
  Coden =	 {JFUAAW},
  Doi =		 {10.1016/S0022-1236(03)00217-9},
  Fjournal =	 {Journal of Functional Analysis},
  Issn =	 {0022-1236},
  Journal =	 {J. Funct. Anal.},
  Mrclass =	 {81U10 (47A40 47N50 81T10)},
  Mrnumber =	 {2035028 (2005g:81341)},
  Mrreviewer =	 {Fumio Hiroshima},
  Number =	 {2},
  Pages =	 {302--359},
  Title =	 {Scattering theory for a class of fermionic {P}auli-{F}ierz models},
  optt =	 {Scattering theory for a class of fermionic Pauli-Fierz models},
  Url =		 {http://dx.doi.org/10.1016/S0022-1236(03)00217-9},
  Volume =	 {208},
  Year =	 {2004},
  OPTkey =	 {},
  OPTmonth =	 {},
  OPTnote =	 {},
  OPTannote =	 {}
}

@article{ammari2014jsp,
  Author =	 {Ammari, Zied and Falconi, Marco},
  opta =	 {Z. Ammari, M. Falconi},
  Doi =		 {10.1007/s10955-014-1079-7},
  Journal =	 {J. Stat. Phys.},
  Fjournal =	 {Journal of Statistical Physics},
  Number =	 {2},
  Pages =	 {330-362},
  optt =	 {Wigner measures approach to the classical limit of the Nelson model: Convergence of dynamics and ground state energy},
  Title =	 {{W}igner measures approach to the classical limit of the {N}elson model: {C}onvergence of dynamics and ground state
                  energy},
  Url =		 {http://dx.doi.org/10.1007/s10955-014-1079-7},
  Volume =	 {157},
  Year =	 {2014},
  archivePrefix ={arXiv},
  eprint =	 {1403.2327},
  primaryClass = {math.AP},
  OPTkey =	 {},
  OPTmonth =	 {},
  OPTnote =	 {},
  OPTannote =	 {},
}

@Article{ammari2017sima,
  Author =	 {Ammari, Zied and Falconi, Marco},
  opta =	 {Z. Ammari, M. Falconi},
  Title =	 {{Bohr's correspondence principle for the renormalized Nelson model}},
  optt =	 {Bohr's correspondence principle for the renormalized Nelson model},
  Journal =	 {SIAM J. Math. Anal.},
  Fjournal =	 {SIAM Journal on Mathematical Analysis},
  year =	 {2017},
  archivePrefix ={arXiv},
  eprint =	 {1602.03212},
  primaryClass = {math-ph},
  OPTkey =	 {},
  doi = 	 {10.1137/17M1117598},
  volume =	 {49},
  number =	 {6},
  pages =	 {5031-5095},
  OPTmonth =	 {},
  note =	 {},
  OPTannote =	 {},
}

@book{arai2020mps,
  AUTHOR =	 {Arai, Asao},
  TITLE =	 {Inequivalent representations of canonical commutation and anti-commutation relations---representation-theoretical
                  viewpoint for quantum phenomena},
  SERIES =	 {Mathematical Physics Studies},
  PUBLISHER =	 {Springer, Singapore},
  YEAR =	 {2020},
  PAGES =	 {xix+493},
  ISBN =	 {978-981-15-2179-9; 978-981-15-2180-5},
  MRCLASS =	 {81-02 (46L60 47N50 81R10 81S05)},
  MRNUMBER =	 {4292535},
  DOI =		 {10.1007/978-981-15-2180-5},
  URL =		 {https://doi.org/10.1007/978-981-15-2180-5},
  opta =	 {A. Arai},
  ALTeditor =	 {},
  optt =	 {Inequivalent representations of canonical commutation and anti-commutation relations - representation-theoretical viewpoint for quantum phenomena},
  OPTkey =	 {},
  OPTvolume =	 {},
  OPTnumber =	 {},
  OPTaddress =	 {},
  OPTedition =	 {},
  OPTmonth =	 {},
  OPTnote =	 {},
  OPTarchivePrefix ={},
  OPTeprint =	 {},
  OPTprimaryclass ={},
  OPTannote =	 {},
}

@article{bach1998advm,
  Author =	 {Bach, Volker and Fr{\"o}hlich, J{\"u}rg and Sigal, Israel Michael},
  opta =	 {V. Bach, J. Fröhlich, I. M. Sigal},
  Coden =	 {ADMTA4},
  Doi =		 {10.1006/aima.1998.1734},
  Fjournal =	 {Advances in Mathematics},
  Issn =	 {0001-8708},
  Journal =	 {Adv. Math.},
  Mrclass =	 {81Q15 (47N50 81Q10 81V10)},
  Mrnumber =	 {1639713 (99e:81051b)},
  Mrreviewer =	 {H. Hogreve},
  Number =	 {2},
  Pages =	 {299--395},
  Title =	 {Quantum electrodynamics of confined nonrelativistic particles},
  optt =	 {Quantum electrodynamics of confined nonrelativistic particles},
  Url =		 {http://dx.doi.org/10.1006/aima.1998.1734},
  Volume =	 {137},
  Year =	 {1998},
  OPTkey = 	 {},
  OPTmonth = 	 {},
  OPTnote = 	 {},
  OPTannote = 	 {}
}

@article{bach1999cmp,
  Author =	 {Bach, Volker and Fr{\"o}hlich, J{\"u}rg and Sigal, Israel Michael},
  opta =	 {V. Bach, J. Fröhlich, I. M. Sigal},
  Coden =	 {CMPHAY},
  Doi =		 {10.1007/s002200050726},
  Fjournal =	 {Communications in Mathematical Physics},
  Issn =	 {0010-3616},
  Journal =	 {Comm. Math. Phys.},
  Mrclass =	 {81V10 (35J10 35Q40 81Q10 81T99)},
  Mrnumber =	 {1724854 (2001j:81256)},
  Mrreviewer =	 {H. Hogreve},
  Number =	 {2},
  Pages =	 {249--290},
  Title =	 {Spectral analysis for systems of atoms and molecules coupled to the quantized radiation field},
  optt =	 {Spectral analysis for systems of atoms and molecules coupled to the quantized radiation field},
  Url =		 {http://dx.doi.org/10.1007/s002200050726},
  Volume =	 {207},
  Year =	 {1999},
  OPTkey =	 {},
  OPTmonth =	 {},
  OPTnote =	 {},
  OPTannote =	 {}
}

@article{bach1999cmp2,
  Author =	 {Bach, Volker and Fr{\"o}hlich, J{\"u}rg and Sigal, Israel Michael and Soffer, Avy},
  opta =	 {V. Bach, J. Fröhlich, I. M. Sigal, A. Soffer},
  Coden =	 {CMPHAY},
  Doi =		 {10.1007/s002200050737},
  Fjournal =	 {Communications in Mathematical Physics},
  Issn =	 {0010-3616},
  Journal =	 {Comm. Math. Phys.},
  Mrclass =	 {81Q10 (47N50 81V45 81V55)},
  Mrnumber =	 {1727240 (2000m:81046)},
  Mrreviewer =	 {Fumio Hiroshima},
  Number =	 {3},
  Pages =	 {557--587},
  Title =	 {Positive commutators and the spectrum of {P}auli-{F}ierz {H}amiltonian of atoms and molecules},
  optt =	 {Positive commutators and the spectrum of Pauli-Fierz Hamiltonian of atoms and molecules},
  Url =		 {http://dx.doi.org/10.1007/s002200050737},
  Volume =	 {207},
  Year =	 {1999},
  OPTkey = 	 {},
  OPTmonth = 	 {},
  OPTnote = 	 {},
  OPTannote = 	 {}
}

@article {bach2006cmp,
  AUTHOR =	 {Bach, Volker and Fr\"{o}hlich, J\"{u}rg and Pizzo, Alessandro},
  TITLE =	 {Infrared-finite algorithms in {QED}: the groundstate of an atom interacting with the quantized radiation field},
  JOURNAL =	 {Comm. Math. Phys.},
  FJOURNAL =	 {Communications in Mathematical Physics},
  VOLUME =	 {264},
  YEAR =	 {2006},
  NUMBER =	 {1},
  PAGES =	 {145--165},
  ISSN =	 {0010-3616},
  MRCLASS =	 {81V10},
  MRNUMBER =	 {2212219},
  MRREVIEWER =	 {Fumio Hiroshima},
  DOI =		 {10.1007/s00220-005-1478-3},
  URL =		 {https://doi.org/10.1007/s00220-005-1478-3},
  opta = 	 {V. Bach, J. Fröhlich, A. Pizzo},
  optt = 	 {Infrared-finite algorithms in {QED}: the groundstate of an atom interacting with the quantized radiation field},
  OPTkey = 	 {},
  OPTFile = 	 {},
  OPTmonth = 	 {},
  OPTnote = 	 {},
  OPTarchivePrefix = {},
  OPTeprint = 	 {},
  OPTprimaryclass = {},
  OPTannote = 	 {},
}

@article {bach2022ahp,
  AUTHOR =	 {Bach, Volker and Hach, Alexander},
  TITLE =	 {On the ultraviolet limit of the {P}auli-{F}ierz {H}amiltonian in the {L}ieb-{L}oss model},
  JOURNAL =	 {Ann. Henri Poincar\'{e}},
  FJOURNAL =	 {Annales Henri Poincar\'{e}. A Journal of Theoretical and Mathematical Physics},
  VOLUME =	 {23},
  YEAR =	 {2022},
  NUMBER =	 {6},
  PAGES =	 {2207--2245},
  ISSN =	 {1424-0637},
  MRCLASS =	 {81T16 (81Q10 81V10)},
  MRNUMBER =	 {4420572},
  MRREVIEWER =	 {Marco Falconi},
  DOI =		 {10.1007/s00023-021-01124-2},
  URL =		 {https://doi.org/10.1007/s00023-021-01124-2},
  opta = 	 {V. Bach, A. Hach},
  optt = 	 {On the ultraviolet limit of the Pauli-Fierz Hamiltonian in the Lieb-Loss model},
  OPTkey = 	 {},
  OPTFile = 	 {},
  OPTmonth = 	 {},
  OPTnote = 	 {},
  OPTarchivePrefix = {},
  OPTeprint = 	 {},
  OPTprimaryclass = {},
  OPTannote = 	 {},
}

@Article{bazaes2023arxiv,
  author = 	 {Bazaes, Rodrigo and Mukherjee, Chiranjib and Sellke, Mark and Varadhan, S.R.S.},
  opta = 	 {R. Bazaes, C. Mukherjee, M. Sellke, S.R.S. Varadhan},
  title = 	 {{Effective mass of the Fr\"{o}hlich Polaron and the Landau-Pekar-Spohn conjecture}},
  optt = 	 {Effective mass of the Fröhlich Polaron and the Landau-Pekar-Spohn conjecture},
  journal = 	 {ArXiv e-prints},
  year = 	 {2023},
  OPTkey = 	 {},
  OPTFile = 	 {},
  OPTvolume = 	 {},
  OPTnumber = 	 {},
  OPTpages = 	 {},
  OPTmonth = 	 {},
  OPTnote = 	 {},
  archivePrefix = {arXiv},
  eprint = 	 {2307.13058},
  primaryclass = {math-ph},
  OPTdoi = 	 {},
  OPTannote = 	 {}
}

@article {betz2023cmp,
  AUTHOR =	 {Betz, Volker and Polzer, Steffen},
  TITLE =	 {Effective mass of the polaron: a lower bound},
  JOURNAL =	 {Comm. Math. Phys.},
  FJOURNAL =	 {Communications in Mathematical Physics},
  VOLUME =	 {399},
  YEAR =	 {2023},
  NUMBER =	 {1},
  PAGES =	 {173--188},
  ISSN =	 {0010-3616},
  MRCLASS =	 {81V70},
  MRNUMBER =	 {4567371},
  MRREVIEWER =	 {Marco Olivieri},
  DOI =		 {10.1007/s00220-022-04553-0},
  URL =		 {https://doi.org/10.1007/s00220-022-04553-0},
  opta =	 {V. Betz, S. Polzer},
  optt =	 {Effective mass of the polaron: a lower bound},
  OPTkey =	 {},
  OPTFile =	 {},
  OPTmonth =	 {},
  OPTnote =	 {},
  OPTarchivePrefix ={},
  OPTeprint =	 {},
  OPTprimaryclass ={},
  OPTannote =	 {},
  opta = 	 {},
  optt = 	 {},
}

@book{bratteli1997tmp2,
  Author =	 {Bratteli, Ola and Robinson, Derek W.},
  opta =	 {O. Bratteli, D. W. Robinson},
  Doi =		 {10.1007/978-3-662-03444-6},
  Edition =	 {Second Edition},
  Isbn =	 {3-540-61443-5},
  Mrclass =	 {82B10 (46Lxx 81S05 82-02 82C10)},
  Mrnumber =	 {1441540 (98e:82004)},
  Note =	 {Equilibrium states. Models in quantum statistical mechanics},
  Pages =	 {xiv+519},
  Publisher =	 {Springer-Verlag, Berlin},
  Series =	 {Texts and Monographs in Physics},
  Title =	 {Operator algebras and quantum statistical mechanics. 2},
  optt =	 {Operator algebras and quantum statistical mechanics. 2},
  Url =		 {http://dx.doi.org/10.1007/978-3-662-03444-6},
  Year =	 {1997},
  ALTeditor =	 {},
  OPTkey =	 {},
  OPTvolume =	 {},
  OPTnumber =	 {},
  OPTaddress =	 {},
  OPTmonth =	 {},
  OPTannote =	 {}
}

@Article{brooks2022arxiv,
  author = 	 {Brooks, Morris and Seiringer, Robert},
  opta = 	 {M. Brooks, R. Seiringer},
  title = 	 {{The Fr\"{o}hlich Polaron at Strong Coupling -- Part II: Energy-Momentum Relation and Effective Mass}},
  optt = 	 {The Fröhlich Polaron at Strong Coupling -- Part II: Energy-Momentum Relation and Effective Mass},
  journal = 	 {ArXiv e-prints},
  year = 	 {2022},
  OPTkey = 	 {},
  OPTFile = 	 {},
  OPTvolume = 	 {},
  OPTnumber = 	 {},
  OPTpages = 	 {},
  OPTmonth = 	 {},
  OPTnote = 	 {},
  archivePrefix = {arXiv},
  eprint = 	 {2211.03353},
  primaryclass = {math-ph},
  OPTdoi = 	 {},
  OPTannote = 	 {}
}

@article {brooks2023cmp,
  AUTHOR =	 {Brooks, Morris and Seiringer, Robert},
  TITLE =	 {The {F}r\"{o}hlich polaron at strong coupling: {P}art {I}---{T}he quantum correction to the classical energy},
  JOURNAL =	 {Comm. Math. Phys.},
  FJOURNAL =	 {Communications in Mathematical Physics},
  VOLUME =	 {404},
  YEAR =	 {2023},
  NUMBER =	 {1},
  PAGES =	 {287--337},
  ISSN =	 {0010-3616},
  MRCLASS =	 {81Q10 (81V70)},
  MRNUMBER =	 {4655801},
  DOI =		 {10.1007/s00220-023-04841-3},
  URL =		 {https://doi.org/10.1007/s00220-023-04841-3},
  opta = 	 {M. Brooks, R. Seiringer},
  optt = 	 {The Fröhlich polaron at strong coupling: Part I — The quantum correction to the classical energy},
  OPTkey = 	 {},
  OPTFile = 	 {},
  OPTmonth = 	 {},
  OPTnote = 	 {},
  OPTarchivePrefix = {},
  OPTeprint = 	 {},
  OPTprimaryclass = {},
  OPTannote = 	 {},
}

@article {bruned2019im,
  AUTHOR =	 {Bruned, Y. and Hairer, M. and Zambotti, L.},
  TITLE =	 {Algebraic renormalisation of regularity structures},
  JOURNAL =	 {Invent. Math.},
  FJOURNAL =	 {Inventiones Mathematicae},
  VOLUME =	 {215},
  YEAR =	 {2019},
  NUMBER =	 {3},
  PAGES =	 {1039--1156},
  ISSN =	 {0020-9910,1432-1297},
  MRCLASS =	 {35R60 (16T05 35A30 60H15 82C28)},
  MRNUMBER =	 {3935036},
  MRREVIEWER =	 {Peter\ E.\ Kloeden},
  DOI =		 {10.1007/s00222-018-0841-x},
  URL =		 {https://doi.org/10.1007/s00222-018-0841-x},
  opta = 	 {Y. Bruned, M. Hairer, L. Zambotti},
  optt = 	 {Algebraic renormalisation of regularity structures},
  OPTkey = 	 {},
  OPTFile = 	 {},
  OPTmonth = 	 {},
  OPTnote = 	 {},
  OPTarchivePrefix = {},
  OPTeprint = 	 {},
  OPTprimaryclass = {},
  OPTannote = 	 {},
}

@article {chatterjee2024cmp,
  AUTHOR =	 {Cao, Sky and Chatterjee, Sourav},
  TITLE =	 {A state space for 3{D} {E}uclidean {Y}ang-{M}ills theories},
  JOURNAL =	 {Comm. Math. Phys.},
  FJOURNAL =	 {Communications in Mathematical Physics},
  VOLUME =	 {405},
  YEAR =	 {2024},
  NUMBER =	 {1},
  PAGES =	 {Paper No. 3, 69},
  ISSN =	 {0010-3616,1432-0916},
  MRCLASS =	 {81T13 (60G15 60G60)},
  MRNUMBER =	 {4689958},
  DOI =		 {10.1007/s00220-023-04870-y},
  URL =		 {https://doi.org/10.1007/s00220-023-04870-y},
  opta = 	 {S. Cao, S. Chatterjeee},
  optt = 	 {A state space for 3D Euclidean Yang-Mills theories},
  OPTkey = 	 {},
  OPTFile = 	 {},
  OPTmonth = 	 {},
  OPTnote = 	 {},
  OPTarchivePrefix = {},
  OPTeprint = 	 {},
  OPTprimaryclass = {},
  OPTannote = 	 {},
}

@incollection{correggi2022iqm22,
  author =	 {Correggi, Michele and Falconi, Marco and Merkli, Marco},
  opta =	 {M. Correggi, M. Falconi, M. Merkli},
  title =	 {{Quasi-Classical Spin-Boson Models}},
  optt =	 {Quasi-Classical Spin-Boson Models},
  journal =	 {ArXiv e-prints},
  year =	 {2023},
  OPTkey =	 {},
  OPTvolume =	 {},
  OPTnumber =	 {},
  OPTpages =	 {},
  OPTmonth =	 {},
  note =	 {In press},
  archivePrefix ={artXiv},
  eprint =	 {2209.06477},
  primaryclass = {math-ph},
  OPTdoi =	 {},
  OPTannote =	 {},
  booktitle = 	 {Proceedings of the INdAM Quantum Meetings 2022},
  OPTcrossref =  {},
  OPTpublisher = {},
  editor = 	 {Correggi, Michele and Falconi, Marco},
  OPTseries = 	 {},
  OPTtype = 	 {},
  OPTchapter = 	 {},
  OPTedition = 	 {},
  OPTaddress = 	 {}
}

@Article{correggi2023apde,
  author = 	 {Correggi, Michele and Falconi, Marco and Olivieri, Marco},
  opta = 	 {M. Correggi, M. Falconi, M. Olivieri},
  title = 	 {{G}round {S}tate {P}roperties in the {Q}uasi-{C}lassical {R}egime},
  optt = 	 {Ground State Properties in the Quasi-Classical Regime},
  journal = 	 {Anal. PDE},
  fjournal =     {Analysis \& PDE},
  year = 	 {2023},
  OPTkey = 	 {},
  volume = 	 {16},
  number = 	 {8},
  pages = 	 {1745--1798},
  OPTmonth = 	 {},
  OPTnote = 	 {},
  archivePrefix = {arXiv},
  eprint = 	 {2007.09442},
  primaryclass = {math-ph},
  doi = 	 {10.2140/apde.2023.16.1745},
  OPTannote = 	 {}
}

@Article{correggi2023jems,
  title =	 {{Quasi-Classical Dynamics}},
  author =	 {Correggi, Michele and Falconi, Marco and Olivieri, Marco},
  year =	 {2023},
  eprint =	 {1909.13313},
  archivePrefix ={arXiv},
  primaryClass = {math-ph},
  opta =	 {M. Correggi, M. Falconi, M. Olivieri},
  optt =	 {Quasi-Classical Dynamics},
  Journal =	 {J. Eur. Math. Soc. (JEMS)},
  Fjournal =     {Journal of the European Mathematical Society},
  OPTkey =	 {},
  volume =	 {25},
  number =	 {2},
  pages =	 {731--783},
  OPTmonth =	 {},
  OPTnote =	 {},
  doi =          {10.4171/JEMS/1197},
  OPTannote =	 {},
}

@Article{correggi2024quantum,
  author = 	 {Correggi, Michele and Falconi, Marco and Fantechi, Michele and Merkli, Marco},
  opta = 	 {M. Correggi, M. Falconi, M. Fantechi, M. Merkli},
  title = 	 {Quasi-classical limit of a spin coupled to a reservoir},
  optt = 	 {Quasi-classical limit of a spin coupled to a reservoir},
  journal = 	 {Quantum},
  year = 	 {2024},
  OPTkey = 	 {},
  volume = 	 {8},
  OPTnumber = 	 {},
  pages = 	 {1561},
  month = 	 {dec},
  OPTnote = 	 {},
  archivePrefix = {arXiv},
  eprint = 	 {2408.02515},
  primaryclass = {math-ph},
  doi = 	 {10.22331/q-2024-12-11-1561},
  OPTannote = 	 {}
}

@book {dalmaso1993pnde,
  AUTHOR =	 {Dal Maso, Gianni},
  TITLE =	 {An introduction to {$\Gamma$}-convergence},
  SERIES =	 {Progress in Nonlinear Differential Equations and their Applications},
  VOLUME =	 {8},
  PUBLISHER =	 {Birkh\"{a}user Boston, Inc., Boston, MA},
  YEAR =	 {1993},
  PAGES =	 {xiv+340},
  ISBN =	 {0-8176-3679-X},
  MRCLASS =	 {49-02 (46N10 47H99 47N10 49J45 73B27)},
  MRNUMBER =	 {1201152},
  MRREVIEWER =	 {T. Zolezzi},
  DOI =		 {10.1007/978-1-4612-0327-8},
  URL =		 {https://doi.org/10.1007/978-1-4612-0327-8},
  opta = 	 {G. Dal Maso},
  ALTeditor = 	 {},
  optt = 	 {An introduction to Γ-convergence},
  OPTkey = 	 {},
  OPTnumber = 	 {},
  OPTaddress = 	 {},
  OPTedition = 	 {},
  OPTmonth = 	 {},
  OPTnote = 	 {},
  OPTarchivePrefix = {},
  OPTeprint = 	 {},
  OPTprimaryclass = {},
  OPTannote = 	 {},
}

@article{derezinski1999rmp,
  Author =	 {Derezi{\'n}ski, J. and G{\'e}rard, C.},
  opta =	 {J. Dereziński, C. Gérard},
  Doi =		 {10.1142/S0129055X99000155},
  Fjournal =	 {Reviews in Mathematical Physics},
  Issn =	 {0129-055X},
  Journal =	 {Rev. Math. Phys.},
  Mrclass =	 {81T08 (47N50 81Q10 81T10)},
  Mrnumber =	 {1682684 (2000h:81141)},
  Mrreviewer =	 {Hiroshi Isozaki},
  Number =	 {4},
  Pages =	 {383--450},
  Title =	 {Asymptotic completeness in quantum field theory. {M}assive {P}auli-{F}ierz {H}amiltonians},
  optt =	 {Asymptotic completeness in quantum field theory. Massive Pauli-Fierz Hamiltonians},
  Url =		 {http://dx.doi.org/10.1142/S0129055X99000155},
  Volume =	 {11},
  Year =	 {1999},
  OPTkey =	 {},
  OPTmonth =	 {},
  OPTnote =	 {},
  OPTannote =	 {}
}

@article{derezinski2001jfa,
  Author =	 {Derezi{{\'n}}ski, Jan and Jak{\v{s}}i{{\'c}}, Vojkan},
  opta =	 {J. Dereziński, V. Jakšić},
  Coden =	 {JFUAAW},
  Doi =		 {10.1006/jfan.2000.3681},
  Fjournal =	 {Journal of Functional Analysis},
  Issn =	 {0022-1236},
  Journal =	 {J. Funct. Anal.},
  Mrclass =	 {81Q10 (47N50 81T10)},
  Mrnumber =	 {1814991 (2002c:81049)},
  Mrreviewer =	 {Fumio Hiroshima},
  Number =	 {2},
  Pages =	 {243--327},
  Title =	 {Spectral theory of {P}auli-{F}ierz operators},
  optt =	 {Spectral theory of Pauli-Fierz operators},
  Url =		 {http://dx.doi.org/10.1006/jfan.2000.3681},
  Volume =	 {180},
  Year =	 {2001},
  OPTkey = 	 {},
  OPTmonth = 	 {},
  OPTnote = 	 {},
  OPTannote = 	 {},
  OPTarchivePrefix = {},
  OPTeprint = 	 {},
  OPTprimaryclass = {}
}

@article {derezinski2003ahp,
  AUTHOR =	 {Derezi\'{n}ski, Jan},
  TITLE =	 {Van {H}ove {H}amiltonians---exactly solvable models of the infrared and ultraviolet problem},
  JOURNAL =	 {Ann. Henri Poincar\'{e}},
  FJOURNAL =	 {Annales Henri Poincar\'{e}. A Journal of Theoretical and Mathematical Physics},
  VOLUME =	 {4},
  YEAR =	 {2003},
  NUMBER =	 {4},
  PAGES =	 {713--738},
  ISSN =	 {1424-0637},
  MRCLASS =	 {81U05 (47N50 81Q10 81T10)},
  MRNUMBER =	 {2015428},
  MRREVIEWER =	 {Fumio Hiroshima},
  DOI =		 {10.1007/s00023-003-0145-5},
  URL =		 {https://doi.org/10.1007/s00023-003-0145-5},
  opta = 	 {J. Dereziński},
  optt = 	 {Van Hove Hamiltonians - exactly solvable models of the infrared and ultraviolet problem},
  OPTkey = 	 {},
  OPTmonth = 	 {},
  OPTnote = 	 {},
  OPTarchivePrefix = {},
  OPTeprint = 	 {},
  OPTprimaryclass = {},
  OPTannote = 	 {},
}

@article {donsker1983cpam,
  AUTHOR =	 {Donsker, M. D. and Varadhan, S. R. S.},
  TITLE =	 {Asymptotics for the polaron},
  JOURNAL =	 {Comm. Pure Appl. Math.},
  FJOURNAL =	 {Communications on Pure and Applied Mathematics},
  VOLUME =	 {36},
  YEAR =	 {1983},
  NUMBER =	 {4},
  PAGES =	 {505--528},
  ISSN =	 {0010-3640},
  MRCLASS =	 {82A05 (60F10 82A31)},
  MRNUMBER =	 {709647},
  MRREVIEWER =	 {Ronald F. Fox},
  DOI =		 {10.1002/cpa.3160360408},
  URL =		 {https://doi.org/10.1002/cpa.3160360408},
  opta = 	 {M.D. Donsker, S.R.S. Varadhan},
  optt = 	 {Asymptotics for the polaron},
  OPTkey = 	 {},
  OPTFile = 	 {},
  OPTmonth = 	 {},
  OPTnote = 	 {},
  OPTarchivePrefix = {},
  OPTeprint = 	 {},
  OPTprimaryclass = {},
  OPTannote = 	 {},
}

@Article{falconi2025prims,
  author =	 {Falconi, Marco and Hinrichs, Benjamin},
  opta =	 {M. Falconi, B. Hinrichs},
  title =	 {Ultraviolet Renormalization of the van Hove--Miyatake Model: an Algebraic and Hamiltonian Approach},
  optt =	 {Ultraviolet Renormalization of the van Hove-Miyatake Model: an Algebraic and Hamiltonian Approach},
  journal =	 {ArXiv e-prints},
  year =	 {2025},
  OPTkey =	 {},
  OPTvolume =	 {},
  OPTnumber =	 {},
  OPTpages =	 {},
  OPTmonth =	 {},
  OPTnote =	 {},
  archivePrefix ={arXiv},
  eprint =	 {2505.19977},
  primaryclass = {math-ph},
  OPTdoi =	 {},
  OPTannote =	 {},
}

@article{frohlich2001am,
  Author =	 {Fr{\"o}hlich, J. and Griesemer, M. and Schlein, B.},
  opta =	 {J. Fröhlich, M. Griesemer, B. Schlein},
  Coden =	 {ADMTA4},
  Doi =		 {10.1006/aima.2001.2026},
  Fjournal =	 {Advances in Mathematics},
  Issn =	 {0001-8708},
  Journal =	 {Adv. Math.},
  Mrclass =	 {81V10 (81Q10 81V45)},
  Mrnumber =	 {1878286 (2003a:81218)},
  Mrreviewer =	 {Fumio Hiroshima},
  Number =	 {2},
  Pages =	 {349--398},
  Title =	 {Asymptotic electromagnetic fields in models of quantum-mechanical matter interacting with the quantized radiation
                  field},
  optt =	 {Asymptotic electromagnetic fields in models of quantum-mechanical matter interacting with the quantized radiation
                  field},
  Url =		 {http://dx.doi.org/10.1006/aima.2001.2026},
  Volume =	 {164},
  Year =	 {2001},
  OPTkey = 	 {},
  OPTmonth = 	 {},
  OPTnote = 	 {},
  OPTannote = 	 {}
}

@article{frohlich2004cmp251,
  Author =	 {Fr{\"o}hlich, J{\"u}rg and Merkli, Marco},
  opta =	 {J. Fröhlich, M. Merkli},
  Doi =		 {10.1007/s00220-004-1176-6},
  Issn =	 {0010-3616},
  Journal =	 {Comm. Math. Phys.},
  Fjournal =	 {Communications in Mathematical Physics},
  Language =	 {English},
  Number =	 {2},
  Pages =	 {235-262},
  Publisher =	 {Springer-Verlag},
  Title =	 {{Another Return of ``Return to Equilibrium''}},
  optt =	 {Another Return of ``Return to Equilibrium''},
  Url =		 {http://dx.doi.org/10.1007/s00220-004-1176-6},
  Volume =	 {251},
  Year =	 {2004},
  OPTkey =	 {},
  OPTmonth =	 {},
  OPTnote =	 {},
  OPTannote =	 {}
}

@article{frohlich2004cmp252,
  Author =	 {Fr{\"o}hlich, J. and Griesemer, M. and Schlein, B.},
  opta =	 {J. Fröhlich, M. Griesemer, B. Schlein},
  Coden =	 {CMPHAY},
  Doi =		 {10.1007/s00220-004-1180-x},
  Fjournal =	 {Communications in Mathematical Physics},
  Issn =	 {0010-3616},
  Journal =	 {Comm. Math. Phys.},
  Mrclass =	 {81U05 (47N50)},
  Mrnumber =	 {2104885 (2006d:81331)},
  Mrreviewer =	 {Fumio Hiroshima},
  Number =	 {1-3},
  Pages =	 {415--476},
  Title =	 {Asymptotic completeness for {C}ompton scattering},
  optt =	 {Asymptotic completeness for Compton scattering},
  Url =		 {http://dx.doi.org/10.1007/s00220-004-1180-x},
  Volume =	 {252},
  Year =	 {2004},
  OPTkey = 	 {},
  OPTmonth = 	 {},
  OPTnote = 	 {},
  OPTannote = 	 {}
}

@article {gerard2012jfa,
  AUTHOR =	 {G\'{e}rard, C. and Hiroshima, F. and Panati, A. and Suzuki, A.},
  TITLE =	 {Absence of ground state for the {N}elson model on static space-times},
  JOURNAL =	 {J. Funct. Anal.},
  FJOURNAL =	 {Journal of Functional Analysis},
  VOLUME =	 {262},
  YEAR =	 {2012},
  NUMBER =	 {1},
  PAGES =	 {273--299},
  ISSN =	 {0022-1236},
  MRCLASS =	 {81T10 (81T20)},
  MRNUMBER =	 {2852262},
  MRREVIEWER =	 {Salvatore Esposito},
  DOI =		 {10.1016/j.jfa.2011.09.010},
  URL =		 {https://doi.org/10.1016/j.jfa.2011.09.010},
  opta = 	 {C. Gérard, F. Hiroshima, A. Panati, A. Suzuki},
  optt = 	 {Absence of ground state for the Nelson model on static space-times},
  OPTkey = 	 {},
  OPTFile = 	 {},
  OPTmonth = 	 {},
  OPTnote = 	 {},
  OPTarchivePrefix = {},
  OPTeprint = 	 {},
  OPTprimaryclass = {},
  OPTannote = 	 {},
}

@article{ginibre1970cmpren,
  Author =	 {Ginibre, J. and Velo, G.},
  opta =	 {J. Ginibre, G. Velo},
  Fjournal =	 {Communications in Mathematical Physics},
  Issn =	 {0010-3616},
  Journal =	 {Comm. Math. Phys.},
  Mrclass =	 {81.46},
  Mrnumber =	 {0266533 (42 \#1436)},
  Mrreviewer =	 {P. Pascual},
  Pages =	 {65--81},
  Title =	 {Renormalization of a quadratic interaction in the {H}amiltonian formalism},
  optt =	 {Renormalization of a quadratic interaction in the Hamiltonian formalism},
  Volume =	 {18},
  Year =	 {1970},
  OPTkey = 	 {},
  OPTnumber = 	 {},
  OPTmonth = 	 {},
  OPTnote = 	 {},
  OPTannote = 	 {}
}

@article{glimm1967cmpyuka1,
  Author =	 {Glimm, James},
  opta =	 {J. Glimm},
  Fjournal =	 {Communications in Mathematical Physics},
  Issn =	 {0010-3616},
  Journal =	 {Comm. Math. Phys.},
  Mrclass =	 {81.47},
  Mrnumber =	 {0215585 (35 \#6425)},
  Mrreviewer =	 {G. Heber},
  Pages =	 {343--386},
  Title =	 {Yukawa coupling of quantum fields in two dimensions. {I}},
  optt =	 {Yukawa coupling of quantum fields in two dimensions. I},
  Volume =	 {5},
  Year =	 {1967},
  OPTkey = 	 {},
  OPTnumber = 	 {},
  OPTmonth = 	 {},
  OPTnote = 	 {},
  OPTannote = 	 {}
}

@article{glimm1967cmpyuka2,
  Author =	 {Glimm, James},
  opta =	 {J. Glimm},
  Fjournal =	 {Communications in Mathematical Physics},
  Issn =	 {0010-3616},
  Journal =	 {Comm. Math. Phys.},
  Mrclass =	 {81.46},
  Mrnumber =	 {0218073 (36 \#1162)},
  Mrreviewer =	 {G. Heber},
  Pages =	 {61--76},
  Title =	 {The {Y}ukawa coupling of quantum fields in two dimensions. {II}},
  optt =	 {The Yukawa coupling of quantum fields in two dimensions. II},
  Volume =	 {6},
  Year =	 {1967},
  OPTkey =	 {},
  OPTnumber =	 {},
  OPTmonth =	 {},
  OPTnote =	 {},
  OPTannote =	 {}
}

@article{glimm1968cmp,
  AUTHOR =	 {Glimm, James},
  opta =	 {J. Glimm},
  TITLE =	 {Boson fields with the {$:\Phi ^{4}:$} interaction in three dimensions},
  optt =	 {Boson fields with the :Φ⁴: interaction in three dimensions},
  JOURNAL =	 {Comm. Math. Phys.},
  FJOURNAL =	 {Communications in Mathematical Physics},
  VOLUME =	 {10},
  YEAR =	 {1968},
  PAGES =	 {1--47},
  ISSN =	 {0010-3616},
  MRCLASS =	 {81.46},
  MRNUMBER =	 {0231601},
  MRREVIEWER =	 {W. G\"uttinger},
  URL =		 {http://projecteuclid.org/euclid.cmp/1103840981},
  OPTkey =	 {},
  OPTnumber =	 {},
  OPTmonth =	 {},
  OPTnote =	 {},
  OPTannote =	 {},
}

@article{glimm1968pr1,
  Author =	 {Glimm, James and Jaffe, Arthur},
  opta =	 {J. Glimm, A. Jaffe},
  Fjournal =	 {Physical Reviews},
  Journal =	 {Phys. Rev. (2)},
  Mrclass =	 {81.47},
  Mrnumber =	 {0247845 (40 \#1106)},
  Mrreviewer =	 {K. Veseli{{\'c}}},
  Pages =	 {1945--1951},
  Title =	 {A {$\lambda \phi^{4}$} quantum field without cutoffs. {I}},
  optt =	 {A λφ⁴ quantum field without cutoffs. I},
  Volume =	 {176},
  Year =	 {1968},
  OPTkey =	 {},
  OPTnumber =	 {},
  OPTmonth =	 {},
  OPTnote =	 {},
  OPTannote =	 {}
}

@incollection{glimm1972mcp,
  Address =	 {New York},
  Author =	 {Glimm, J. and Jaffe, A.},
  opta =	 {J. Glimm, A. Jaffe},
  Booktitle =	 {Mathematics of contemporary physics ({P}roc. {I}nstructional {C}onf. ({NATO} {A}dvanced {S}tudy {I}nst.), {B}edford
                  {C}oll., {L}ondon, 1971)},
  Mrclass =	 {81.60},
  Mrnumber =	 {0674511 (58 \#32630)},
  Pages =	 {77--143},
  Publisher =	 {Academic Press},
  Title =	 {Boson quantum field models},
  optt =	 {Boson quantum field models},
  Year =	 {1972},
  OPTcrossref =  {},
  OPTkey = 	 {},
  OPTeditor = 	 {},
  OPTvolume = 	 {},
  OPTnumber = 	 {},
  OPTseries = 	 {},
  OPTtype = 	 {},
  OPTchapter = 	 {},
  OPTedition = 	 {},
  OPTmonth = 	 {},
  OPTnote = 	 {},
  OPTannote = 	 {}
}

@article{glimm1974am,
  Author =	 {Glimm, James and Jaffe, Arthur and Spencer, Thomas},
  opta =	 {J. Glimm, A. Jaffe, T. Spencer},
  Fjournal =	 {Annals of Mathematics. Second Series},
  Issn =	 {0003-486X},
  Journal =	 {Ann. of Math. (2)},
  Mrclass =	 {81.46},
  Mrnumber =	 {0363256},
  Mrreviewer =	 {H. Araki},
  Pages =	 {585--632},
  Title =	 {The {W}ightman axioms and particle structure in the {$P(\phi)_{2}$} quantum field model},
  optt =	 {The Wightman axioms and particle structure in the P(φ)₂ quantum field model},
  Volume =	 {100},
  Year =	 {1974},
  OPTkey = 	 {},
  OPTnumber = 	 {},
  OPTmonth = 	 {},
  OPTnote = 	 {},
  OPTannote = 	 {}
}

@book{glimm1985cp2,
  Author =	 {Glimm, James and Jaffe, Arthur},
  opta =	 {J. Glimm, A. Jaffe},
  Isbn =	 {0-8176-3272-7},
  Mrclass =	 {81-06 (01A75 81T05 81T08)},
  Mrnumber =	 {947959 (91m:81003)},
  Mrreviewer =	 {P. D. F. Ion},
  Note =	 {Constructive quantum field theory. Selected papers, Reprint of articles published 1968--1980},
  Pages =	 {x+533},
  Publisher =	 {Birkh{\"a}user Boston, Inc., Boston, MA},
  Title =	 {Collected papers. {V}ol. 2},
  optt =	 {Collected papers. Vol. 2},
  Year =	 {1985},
  ALTeditor = 	 {},
  OPTkey = 	 {},
  OPTvolume = 	 {},
  OPTnumber = 	 {},
  OPTseries = 	 {},
  OPTaddress = 	 {},
  OPTedition = 	 {},
  OPTmonth = 	 {},
  OPTannote = 	 {}
}

@book{glimm1987qp,
  Address =	 {New York},
  Author =	 {Glimm, James and Jaffe, Arthur},
  opta =	 {J. Glimm, A. Jaffe},
  Doi =		 {10.1007/978-1-4612-4728-9},
  Edition =	 {Second},
  Isbn =	 {0-387-96476-2},
  Mrclass =	 {81-02 (81E05)},
  Mrnumber =	 {887102 (89k:81001)},
  Mrreviewer =	 {P. D. F. Ion},
  Note =	 {A functional integral point of view},
  Pages =	 {xxii+535},
  Publisher =	 {Springer-Verlag},
  Title =	 {Quantum physics},
  optt =	 {Quantum physics},
  Url =		 {http://dx.doi.org/10.1007/978-1-4612-4728-9},
  Year =	 {1987},
  ALTeditor =	 {},
  OPTkey =	 {},
  OPTvolume =	 {},
  OPTnumber =	 {},
  OPTseries =	 {},
  OPTmonth =	 {},
  OPTannote =	 {}
}

@article{griesemer2001im,
  Author =	 {Griesemer, Marcel and Lieb, Elliott H. and Loss, Michael},
  opta =	 {M. Griesemer, E. H. Lieb, M. Loss},
  Coden =	 {INVMBH},
  Doi =		 {10.1007/s002220100159},
  Fjournal =	 {Inventiones Mathematicae},
  Issn =	 {0020-9910},
  Journal =	 {Invent. Math.},
  Mrclass =	 {81V10 (81Q10)},
  Mrnumber =	 {1856401 (2002g:81203)},
  Mrreviewer =	 {Masao Hirokawa},
  Number =	 {3},
  Pages =	 {557--595},
  Title =	 {Ground states in non-relativistic quantum electrodynamics},
  optt =	 {Ground states in non-relativistic quantum electrodynamics},
  Url =		 {http://dx.doi.org/10.1007/s002220100159},
  Volume =	 {145},
  Year =	 {2001},
  OPTkey = 	 {},
  OPTmonth = 	 {},
  OPTnote = 	 {},
  OPTannote = 	 {},
  OPTarchivePrefix = {},
  OPTeprint = 	 {},
  OPTprimaryclass = {}
}

@article{gross1962ap,
  Author =	 {Gross, E. P.},
  opta =	 {E. P. Gross},
  Doi =		 {10.1016/0003-4916(62)90216-6},
  Issn =	 {0003-4916},
  Journal =	 {Ann. Physics},
  Fjournal =	 {Annals of Physics},
  Number =	 {2},
  Pages =	 {219 - 233},
  Title =	 {Particle-like solutions in field theory},
  optt =	 {Particle-like solutions in field theory},
  Volume =	 {19},
  Year =	 {1962},
  OPTkey =	 {},
  OPTmonth =	 {},
  OPTnote =	 {},
  OPTannote =	 {}
}

@article {gubinelli2021cmp,
  AUTHOR =	 {Gubinelli, Massimiliano and Hofmanov\'a, Martina},
  TITLE =	 {A {PDE} construction of the {E}uclidean {$\phi_3^4$} quantum field theory},
  JOURNAL =	 {Comm. Math. Phys.},
  FJOURNAL =	 {Communications in Mathematical Physics},
  VOLUME =	 {384},
  YEAR =	 {2021},
  NUMBER =	 {1},
  PAGES =	 {1--75},
  ISSN =	 {0010-3616,1432-0916},
  MRCLASS =	 {81S20 (35Q40 37D35 37E20 60G60 81T08 81T27 82B31)},
  MRNUMBER =	 {4252872},
  MRREVIEWER =	 {Peter\ Bernard\ Weichman},
  DOI =		 {10.1007/s00220-021-04022-0},
  URL =		 {https://doi.org/10.1007/s00220-021-04022-0},
  opta = 	 {M. Gubinelli, M. Hofmanovà},
  optt = 	 {A PDE construction of the Euclidean ϕ⁴₃ quantum field theory},
  OPTkey = 	 {},
  OPTFile = 	 {},
  OPTmonth = 	 {},
  OPTnote = 	 {},
  OPTarchivePrefix = {},
  OPTeprint = 	 {},
  OPTprimaryclass = {},
  OPTannote = 	 {},
}

@article{guerra1975am,
  Author =	 {Guerra, F. and Rosen, L. and Simon, B.},
  opta =	 {F. Guerra, L. Rosen, B. Simon},
  Fjournal =	 {Annals of Mathematics. Second Series},
  Issn =	 {0003-486X},
  Journal =	 {Ann. of Math. (2)},
  Mrclass =	 {82.46 (81.46)},
  Mrnumber =	 {0378670 (51 \#14836)},
  Mrreviewer =	 {R. A. Minlos},
  Pages =	 {111--189; ibid. (2) 101\ (1975), 191--259},
  Title =	 {The {$\mathbf{P}(\phi)_{2}$} {E}uclidean quantum field theory as classical statistical mechanics. {I}, {II}},
  optt =	 {The P(φ)₂ Euclidean quantum field theory as classical statistical mechanics. I, II},
  Volume =	 {101},
  Year =	 {1975},
  OPTkey = 	 {},
  OPTnumber = 	 {},
  OPTmonth = 	 {},
  OPTnote = 	 {},
  OPTannote = 	 {}
}

@article{haag1964jmp,
  Author =	 {Haag, Rudolf and Kastler, Daniel},
  opta =	 {R. Haag, D. Kastler},
  Doi =		 {10.1063/1.1704187},
  Fjournal =	 {Journal of Mathematical Physics},
  Issn =	 {0022-2488},
  Journal =	 {J. Math. Phys.},
  Mrclass =	 {81.46},
  Mrnumber =	 {0165864},
  Mrreviewer =	 {H. Araki},
  Pages =	 {848--861},
  Title =	 {An algebraic approach to quantum field theory},
  optt =	 {An algebraic approach to quantum field theory},
  Url =		 {http://dx.doi.org/10.1063/1.1704187},
  Volume =	 {5},
  Year =	 {1964},
  OPTkey = 	 {},
  OPTnumber = 	 {},
  OPTmonth = 	 {},
  OPTnote = 	 {},
  OPTannote = 	 {}
}

@book{haag1992tmp,
  Address =	 {Berlin},
  Author =	 {Haag, Rudolf},
  opta =	 {R. Haag},
  Isbn =	 {3-540-53610-8},
  Mrclass =	 {81-01 (46L60 46N50 81P15 81T05)},
  Mrnumber =	 {1182152 (94d:81001)},
  Mrreviewer =	 {Stephen J. Summers},
  Note =	 {Fields, particles, algebras},
  Pages =	 {xiv+356},
  Publisher =	 {Springer-Verlag},
  Series =	 {Texts and Monographs in Physics},
  Title =	 {Local quantum physics},
  optt =	 {Local quantum physics},
  Year =	 {1992},
  ALTeditor = 	 {},
  OPTkey = 	 {},
  OPTvolume = 	 {},
  OPTnumber = 	 {},
  OPTedition = 	 {},
  OPTmonth = 	 {},
  OPTannote = 	 {}
}

@article{hairer2014im,
  Author =	 {Hairer, M.},
  opta =	 {M. Hairer},
  Doi =		 {10.1007/s00222-014-0505-4},
  Fjournal =	 {Inventiones Mathematicae},
  Issn =	 {0020-9910},
  Journal =	 {Invent. Math.},
  Mrclass =	 {60H15 (35R60 60H40 81S20 82C28)},
  Mrnumber =	 {3274562},
  Mrreviewer =	 {Dora Sele{\v{s}}i},
  Number =	 {2},
  Pages =	 {269--504},
  Title =	 {A theory of regularity structures},
  optt =	 {A theory of regularity structures},
  Url =		 {http://dx.doi.org/10.1007/s00222-014-0505-4},
  Volume =	 {198},
  Year =	 {2014},
  OPTkey = 	 {},
  OPTmonth = 	 {},
  OPTnote = 	 {},
  OPTannote = 	 {}
}

@incollection {hairer2016cdm,
  AUTHOR =	 {Hairer, Martin},
  TITLE =	 {Regularity structures and the dynamical {$\Phi^4_3$} model},
  BOOKTITLE =	 {Current developments in mathematics 2014},
  PAGES =	 {1--49},
  PUBLISHER =	 {Int. Press, Somerville, MA},
  YEAR =	 {2016},
  ISBN =	 {978-1-57146-313-5; 1-57146-313-5},
  MRCLASS =	 {60H15 (35R60)},
  MRNUMBER =	 {3468250},
  MRREVIEWER =	 {Peter\ Karl\ Friz},
  opta = 	 {M. Hairer},
  optt = 	 {Regularity structures and the dynamical ϕ⁴₃ model},
  OPTcrossref =  {},
  OPTkey = 	 {},
  OPTFile = 	 {},
  OPTeditor = 	 {},
  OPTvolume = 	 {},
  OPTnumber = 	 {},
  OPTseries = 	 {},
  OPTtype = 	 {},
  OPTchapter = 	 {},
  OPTedition = 	 {},
  OPTmonth = 	 {},
  OPTaddress = 	 {},
  OPTnote = 	 {},
  OPTdoi = 	 {},
  OPTarchivePrefix = {arXiv},
  OPTeprint = 	 {},
  OPTprimaryclass = {},
  OPTannote = 	 {},
}

@article {hairer2016cmp,
  AUTHOR =	 {Hairer, Martin and Shen, Hao},
  TITLE =	 {The dynamical sine-{G}ordon model},
  JOURNAL =	 {Comm. Math. Phys.},
  FJOURNAL =	 {Communications in Mathematical Physics},
  VOLUME =	 {341},
  YEAR =	 {2016},
  NUMBER =	 {3},
  PAGES =	 {933--989},
  ISSN =	 {0010-3616},
  MRCLASS =	 {60H15 (35K58 35R60 37H10)},
  MRNUMBER =	 {3452276},
  MRREVIEWER =	 {Dirk Bl\"{o}mker},
  DOI =		 {10.1007/s00220-015-2525-3},
  URL =		 {https://doi.org/10.1007/s00220-015-2525-3},
  opta =	 {M. Hairer, H. Shen},
  optt =	 {The dynamical sine-Gordon model},
  OPTkey =	 {},
  OPTFile =	 {},
  OPTmonth =	 {},
  OPTnote =	 {},
  OPTarchivePrefix ={},
  OPTeprint =	 {},
  OPTprimaryclass ={},
  OPTannote =	 {},
}

@article{hasler2008rmp,
  Author =	 {Hasler, D. and Herbst, I.},
  opta =	 {D. Hasler, I. Herbst},
  Doi =		 {10.1142/S0129055X08003389},
  Fjournal =	 {Reviews in Mathematical Physics},
  Issn =	 {0129-055X},
  Journal =	 {Rev. Math. Phys.},
  Mrclass =	 {81Q10 (47B25 47N50 81T08 81V10)},
  Mrnumber =	 {2436496 (2009f:81066)},
  Mrreviewer =	 {Andrea Posilicano},
  Number =	 {7},
  Pages =	 {787--800},
  Title =	 {On the self-adjointness and domain of {P}auli-{F}ierz type {H}amiltonians},
  optt =	 {On the self-adjointness and domain of Pauli-Fierz type Hamiltonians},
  Url =		 {http://dx.doi.org/10.1142/S0129055X08003389},
  Volume =	 {20},
  Year =	 {2008},
  archivePrefix ={arXiv},
  eprint =	 {0707.1713},
  primaryClass = {math-ph},
  OPTkey =	 {},
  OPTmonth =	 {},
  OPTnote =	 {},
  OPTannote =	 {}
}

@article {hepp1965cmp,
  AUTHOR =	 {Hepp, Klaus},
  TITLE =	 {On the connection between the {LSZ} and {W}ightman quantum field theory},
  JOURNAL =	 {Comm. Math. Phys.},
  FJOURNAL =	 {Communications in Mathematical Physics},
  VOLUME =	 {1},
  YEAR =	 {1965},
  PAGES =	 {95--111},
  ISSN =	 {0010-3616},
  MRCLASS =	 {81.41},
  MRNUMBER =	 {0181277},
  MRREVIEWER =	 {H. Araki},
  URL =		 {http://projecteuclid.org/euclid.cmp/1103758732},
  opta = 	 {K. Hepp},
  optt = 	 {On the connection between the LSZ and Wightman quantum field theory},
  OPTkey = 	 {},
  OPTnumber = 	 {},
  OPTmonth = 	 {},
  OPTnote = 	 {},
  OPTarchivePrefix = {},
  OPTeprint = 	 {},
  OPTprimaryclass = {},
  OPTdoi = 	 {},
  OPTannote = 	 {},
}

@book{hepp1969,
  Adsurl =	 {http://adsabs.harvard.edu/abs/1969LNP.....2.....H},
  Author =	 {Hepp, K.},
  opta =	 {K. Hepp},
  Booktitle =	 {Lecture Notes in Physics},
  Series =	 {Lecture Notes in Physics},
  Title =	 {{Th{\'e}orie de la renormalisation}},
  optt =	 {Théorie de la renormalisation},
  Volume =	 {2},
  Year =	 {1969},
  ALTeditor = 	 {},
  publisher = 	 {Berlin Springer Verlag},
  OPTkey = 	 {},
  OPTnumber = 	 {},
  OPTaddress = 	 {},
  OPTedition = 	 {},
  OPTmonth = 	 {},
  OPTnote = 	 {},
  OPTannote = 	 {}
}

@article {hidaka2021jst,
  AUTHOR =	 {Hidaka, Takeru and Hiroshima, Fumio and Sasaki, Itaru},
  TITLE =	 {Spectrum of the semi-relativistic {P}auli-{F}ierz model {II}},
  JOURNAL =	 {J. Spectr. Theory},
  FJOURNAL =	 {Journal of Spectral Theory},
  VOLUME =	 {11},
  YEAR =	 {2021},
  NUMBER =	 {4},
  PAGES =	 {1779--1830},
  ISSN =	 {1664-039X,1664-0403},
  MRCLASS =	 {81Q10 (35Q40 47A10 47B25 81R12)},
  MRNUMBER =	 {4349675},
  DOI =		 {10.4171/jst/386},
  URL =		 {https://doi.org/10.4171/jst/386},
  opta = 	 {T. Hidaka, F. Hiroshima, I. Sasaki},
  optt = 	 {Spectrum of the semi-relativistic Pauli-Fierz model II},
  OPTkey = 	 {},
  OPTFile = 	 {},
  OPTmonth = 	 {},
  OPTnote = 	 {},
  OPTarchivePrefix = {},
  OPTeprint = 	 {},
  OPTprimaryclass = {},
  OPTannote = 	 {},
}

@article {hirokawa2005am,
  AUTHOR =	 {Hirokawa, Masao and Hiroshima, Fumio and Spohn, Herbert},
  TITLE =	 {Ground state for point particles interacting through a massless scalar {B}ose field},
  JOURNAL =	 {Adv. Math.},
  FJOURNAL =	 {Advances in Mathematics},
  VOLUME =	 {191},
  YEAR =	 {2005},
  NUMBER =	 {2},
  PAGES =	 {339--392},
  ISSN =	 {0001-8708},
  MRCLASS =	 {81T10 (81T15 81V10)},
  MRNUMBER =	 {2103217},
  MRREVIEWER =	 {Andrea Posilicano},
  DOI =		 {10.1016/j.aim.2004.03.011},
  URL =		 {https://doi.org/10.1016/j.aim.2004.03.011},
  opta = 	 {M. Hirokawa, F. Hiroshima, H. Spohn},
  optt = 	 {Ground state for point particles interacting through a massless scalar Bose field},
  OPTkey = 	 {},
  OPTFile = 	 {},
  OPTmonth = 	 {},
  OPTnote = 	 {},
  OPTarchivePrefix = {},
  OPTeprint = 	 {},
  OPTprimaryclass = {},
  OPTannote = 	 {},
}

@article {hirokawa2006prims,
  AUTHOR =	 {Hirokawa, Masao},
  TITLE =	 {Infrared catastrophe for {N}elson's model---non-existence of ground state and soft-boson divergence},
  JOURNAL =	 {Publ. Res. Inst. Math. Sci.},
  FJOURNAL =	 {Kyoto University. Research Institute for Mathematical Sciences. Publications},
  VOLUME =	 {42},
  YEAR =	 {2006},
  NUMBER =	 {4},
  PAGES =	 {897--922},
  ISSN =	 {0034-5318},
  MRCLASS =	 {81T10 (47N50 81T08 81V10)},
  MRNUMBER =	 {2289153},
  MRREVIEWER =	 {Fumio Hiroshima},
  URL =		 {http://projecteuclid.org/euclid.prims/1166642191},
  opta = 	 {M. Hirokawa},
  optt = 	 {Infrared catastrophe for Nelson's model - non-existence of ground state and soft-boson divergence},
  OPTkey = 	 {},
  OPTmonth = 	 {},
  OPTnote = 	 {},
  OPTarchivePrefix = {},
  OPTeprint = 	 {},
  OPTprimaryclass = {},
  OPTdoi = 	 {},
  OPTannote = 	 {},
}

@article {hiroshima2001tams,
  AUTHOR =	 {Hiroshima, Fumio},
  TITLE =	 {Ground states and spectrum of quantum electrodynamics of nonrelativistic particles},
  JOURNAL =	 {Trans. Amer. Math. Soc.},
  FJOURNAL =	 {Transactions of the American Mathematical Society},
  VOLUME =	 {353},
  YEAR =	 {2001},
  NUMBER =	 {11},
  PAGES =	 {4497--4528},
  ISSN =	 {0002-9947},
  MRCLASS =	 {81Q10 (35P05 35Q40 47N50 81V10)},
  MRNUMBER =	 {1851181},
  MRREVIEWER =	 {Lech Jak\'{o}bczyk},
  DOI =		 {10.1090/S0002-9947-01-02719-2},
  URL =		 {https://doi.org/10.1090/S0002-9947-01-02719-2},
  opta = 	 {F. Hiroshima},
  optt = 	 {Ground states and spectrum of quantum electrodynamics of nonrelativistic particles},
  OPTkey = 	 {},
  OPTmonth = 	 {},
  OPTnote = 	 {},
  OPTarchivePrefix = {},
  OPTeprint = 	 {},
  OPTprimaryclass = {},
  OPTannote = 	 {},
}

@article {hiroshima2022rmp,
  AUTHOR =	 {Hiroshima, Fumio and Matte, Oliver},
  TITLE =	 {Ground states and associated path measures in the renormalized {N}elson model},
  JOURNAL =	 {Rev. Math. Phys.},
  FJOURNAL =	 {Reviews in Mathematical Physics. A Journal for Both Review and Original Research Papers in the Field of Mathematical
                  Physics},
  VOLUME =	 {34},
  YEAR =	 {2022},
  NUMBER =	 {2},
  PAGES =	 {Paper No. 2250002, 84},
  ISSN =	 {0129-055X},
  MRCLASS =	 {47A75 (47D08 60H30 81Q10 81T16)},
  MRNUMBER =	 {4386953},
  MRREVIEWER =	 {Roberto Quezada},
  DOI =		 {10.1142/S0129055X22500027},
  URL =		 {https://doi.org/10.1142/S0129055X22500027},
  opta = 	 {F. Hiroshima, O. Matte},
  optt = 	 {Ground states and associated path measures in the renormalized Nelson model},
  OPTkey = 	 {},
  OPTFile = 	 {},
  OPTmonth = 	 {},
  OPTnote = 	 {},
  OPTarchivePrefix = {},
  OPTeprint = 	 {},
  OPTprimaryclass = {},
  OPTannote = 	 {},
}

@article{jaffe2014arxiv,
  Author =	 {Arthur Jaffe},
  opta =	 {A. Jaffe},
  Eprint =	 {1411.2964},
  Title =	 {{Quantum Fields, Stochastic PDE, and Reflection Positivity}},
  optt =	 {Quantum Fields, Stochastic PDE, and Reflection Positivity},
  Url =		 {http://arxiv.org/abs/1411.2964},
  Year =	 {2014},
  journal = 	 {ArXiv e-prints},
  OPTkey = 	 {},
  OPTvolume = 	 {},
  OPTnumber = 	 {},
  OPTpages = 	 {},
  OPTmonth = 	 {},
  OPTnote = 	 {},
  OPTannote = 	 {}
}

@article{jaksic1995aihppt,
  Author =	 {Jak{\v{s}}i{\'c}, V. and Pillet, C.-A.},
  opta =	 {V. Jakšić, C.-A. Pillet},
  Coden =	 {AHPAAO},
  Fjournal =	 {Annales de l'Institut Henri Poincar\'e. Physique Th\'eorique},
  Issn =	 {0246-0211},
  Journal =	 {Ann. Inst. H. Poincar\'e Phys. Th\'eor.},
  Mrclass =	 {82C21 (82C10)},
  Mrnumber =	 {1313360 (95m:82090)},
  Mrreviewer =	 {Asao Arai},
  Number =	 {1},
  Pages =	 {47--68},
  Title =	 {On a model for quantum friction. {I}. {F}ermi's golden rule and dynamics at zero temperature},
  optt =	 {On a model for quantum friction. I. Fermi's golden rule and dynamics at zero temperature},
  Url =		 {http://www.numdam.org/item?id=AIHPA_1995__62_1_47_0},
  Volume =	 {62},
  Year =	 {1995},
  OPTkey = 	 {},
  OPTmonth = 	 {},
  OPTnote = 	 {},
  OPTannote = 	 {},
}

@article{jaksic1996cmp,
  Author =	 {Jak{\v{s}}i{\'c}, V. and Pillet, C.-A.},
  opta =	 {V. Jakšić, C.-A. Pillet},
  Coden =	 {CMPHAY},
  Fjournal =	 {Communications in Mathematical Physics},
  Issn =	 {0010-3616},
  Journal =	 {Comm. Math. Phys.},
  Mrclass =	 {82C21 (81V70 82C10)},
  Mrnumber =	 {1376434 (97b:82074)},
  Mrreviewer =	 {Asao Arai},
  Number =	 {3},
  Pages =	 {619--644},
  Title =	 {On a model for quantum friction. {II}. {F}ermi's golden rule and dynamics at positive temperature},
  optt =	 {On a model for quantum friction. II. Fermi's golden rule and dynamics at positive temperature},
  Url =		 {http://projecteuclid.org/getRecord?id=euclid.cmp/1104286117},
  Volume =	 {176},
  Year =	 {1996},
  OPTkey = 	 {},
  OPTmonth = 	 {},
  OPTnote = 	 {},
  OPTannote = 	 {}
}

@article{jaksic1996cmp2,
  Author =	 {Jak{\v{s}}i{\'c}, Vojkan and Pillet, Claude-Alain},
  opta =	 {V. Jakšić, C.-A. Pillet},
  Coden =	 {CMPHAY},
  Fjournal =	 {Communications in Mathematical Physics},
  Issn =	 {0010-3616},
  Journal =	 {Comm. Math. Phys.},
  Mrclass =	 {82C10 (46L60)},
  Mrnumber =	 {1395208 (98e:82035)},
  Mrreviewer =	 {Bruno Nachtergaele},
  Number =	 {3},
  Pages =	 {627--651},
  Title =	 {On a model for quantum friction. {III}. {E}rgodic properties of the spin-boson system},
  optt =	 {On a model for quantum friction. III. Ergodic properties of the spin-boson system},
  Url =		 {http://projecteuclid.org/getRecord?id=euclid.cmp/1104286769},
  Volume =	 {178},
  Year =	 {1996},
  OPTkey = 	 {},
  OPTmonth = 	 {},
  OPTnote = 	 {},
  OPTannote = 	 {}
}

@article {jaksic2002cmp,
  AUTHOR =	 {Jak\v{s}i\'{c}, V. and Pillet, C.-A.},
  TITLE =	 {Non-equilibrium steady states of finite quantum systems coupled to thermal reservoirs},
  JOURNAL =	 {Comm. Math. Phys.},
  FJOURNAL =	 {Communications in Mathematical Physics},
  VOLUME =	 {226},
  YEAR =	 {2002},
  NUMBER =	 {1},
  PAGES =	 {131--162},
  ISSN =	 {0010-3616},
  MRCLASS =	 {82C10 (46L60)},
  MRNUMBER =	 {1889995},
  MRREVIEWER =	 {Izumi Ojima},
  DOI =		 {10.1007/s002200200602},
  URL =		 {https://doi.org/10.1007/s002200200602},
  opta = 	 {V. Jakšić, C.-A. Pillet},
  optt = 	 {Non-equilibrium steady states of finite quantum systems coupled to thermal reservoirs},
  OPTkey = 	 {},
  OPTFile = 	 {},
  OPTmonth = 	 {},
  OPTnote = 	 {},
  OPTarchivePrefix = {},
  OPTeprint = 	 {},
  OPTprimaryclass = {},
  OPTannote = 	 {},
}

@article {koenenberg2015cmp,
  AUTHOR =	 {K\"{o}nenberg, M. and Merkli, M. and Song, H.},
  TITLE =	 {Ergodicity of the spin-boson model for arbitrary coupling strength},
  JOURNAL =	 {Comm. Math. Phys.},
  FJOURNAL =	 {Communications in Mathematical Physics},
  VOLUME =	 {336},
  YEAR =	 {2015},
  NUMBER =	 {1},
  PAGES =	 {261--285},
  ISSN =	 {0010-3616},
  MRCLASS =	 {81Q20 (46N40)},
  MRNUMBER =	 {3322374},
  DOI =		 {10.1007/s00220-014-2242-3},
  URL =		 {https://doi.org/10.1007/s00220-014-2242-3},
  opta = 	 {M. Könenberg, M. Merkli, H. Song},
  optt = 	 {Ergodicity of the spin-boson model for arbitrary coupling strength},
  OPTkey = 	 {},
  OPTmonth = 	 {},
  OPTnote = 	 {},
  OPTarchivePrefix = {},
  OPTeprint = 	 {},
  OPTprimaryclass = {},
  OPTannote = 	 {},
}

@Article{lieb1958pr,
  Author =	 {Lieb, Elliott H. and Yamazaki, Kazuo},
  Title =	 {Ground-state energy and effective mass of the polaron},
  FJournal =	 {Physical Review, II. Series},
  Journal =	 {Phys. Rev., II. Ser.},
  ISSN =	 {0031-899X},
  Volume =	 {111},
  Pages =	 {728--733},
  Year =	 {1958},
  Language =	 {English},
  DOI =		 {10.1103/PhysRev.111.728},
  opta = 	 {E.H. Lieb, K. Yamazaki},
  optt = 	 {},
  OPTkeywords = 	 {},
  OPTkey = 	 {},
  OPTFile = 	 {},
  OPTnumber = 	 {},
  OPTmonth = 	 {},
  OPTnote = 	 {},
  OPTarchivePrefix = {},
  OPTeprint = 	 {},
  OPTprimaryclass = {},
  OPTannote = 	 {}
}

@article{lieb1997cmp,
  Author =	 {Lieb, Elliott H. and Thomas, Lawrence E.},
  opta =	 {E. H. Lieb, L. E. Thomas},
  Coden =	 {CMPHAY},
  Doi =		 {10.1007/s002200050040},
  Fjournal =	 {Communications in Mathematical Physics},
  Issn =	 {0010-3616},
  Journal =	 {Comm. Math. Phys.},
  Mrclass =	 {81Q15 (81R30 81V70)},
  Mrnumber =	 {1462224 (99a:81034a)},
  Mrreviewer =	 {Nicolae Angelescu},
  Number =	 {3},
  Pages =	 {511--519},
  Title =	 {Exact ground state energy of the strong-coupling polaron},
  optt =	 {Exact ground state energy of the strong-coupling polaron},
  Url =		 {http://dx.doi.org/10.1007/s002200050040},
  Volume =	 {183},
  Year =	 {1997},
  OPTkey = 	 {},
  OPTmonth = 	 {},
  OPTnote = 	 {},
  OPTannote = 	 {}
}

@ARTICLE{lieb2019arxiv,
  author =	 {{Lieb}, Elliott H. and {Seiringer}, Robert},
  title =	 {Divergence of the effective mass of a polaron in the strong coupling limit},
  journal =	 {J. Stat. Phys.},
  year =	 {2020},
  optmonth =	 {Feb},
  opteid =		 {arXiv:1902.04025},
  pages =	 {23==33},
  archivePrefix ={arXiv},
  eprint =	 {1902.04025},
  primaryClass = {math-ph},
  opta =	 {E. H. Lieb, R. Seiringer},
  optt =	 {Divergence of the effective mass of a polaron in the strong coupling limit},
  OPTkey =	 {},
  volume =	 {180},
  number =	 {1},
  OPTnote =	 {},
  doi =	         {10.1007/s10955-019-02322-3},
  OPTannote =	 {}
}

@book{mastropietro2008npr,
  Author =	 {Mastropietro, Vieri},
  opta =	 {V. Mastropietro},
  Doi =		 {10.1142/9789812792402},
  Isbn =	 {978-981-279-239-6; 981-279-239-2},
  Mrclass =	 {81T16 (81-02 81T17 82B10 82B20)},
  Mrnumber =	 {2414874 (2009m:81156)},
  Mrreviewer =	 {Domingos H. U. Marchetti},
  Pages =	 {xii+290},
  Publisher =	 {World Scientific Publishing Co. Pte. Ltd., Hackensack, NJ},
  Title =	 {Non-perturbative renormalization},
  optt =	 {Non-perturbative renormalization},
  Url =		 {http://dx.doi.org/10.1142/9789812792402},
  Year =	 {2008},
  ALTeditor = 	 {},
  OPTkey = 	 {},
  OPTvolume = 	 {},
  OPTnumber = 	 {}, 
  OPTseries = 	 {},
  OPTaddress = 	 {},
  OPTedition = 	 {},
  OPTmonth = 	 {},
  OPTnote = 	 {},
  OPTannote = 	 {},
}

@article {merkli2007prl,
  AUTHOR =	 {Merkli, M. and Sigal, I. M. and Berman, G. P.},
  TITLE =	 {Decoherence and thermalization},
  JOURNAL =	 {Phys. Rev. Lett.},
  FJOURNAL =	 {Physical Review Letters},
  VOLUME =	 {98},
  YEAR =	 {2007},
  NUMBER =	 {13},
  PAGES =	 {130401, 4},
  ISSN =	 {0031-9007},
  MRCLASS =	 {82C10 (81V80)},
  MRNUMBER =	 {2295681},
  DOI =		 {10.1103/PhysRevLett.98.130401},
  URL =		 {https://doi.org/10.1103/PhysRevLett.98.130401},
  opta = 	 {M. Merkli, I.M. Sigal, G.P. Berman},
  optt = 	 {Decoherence and thermalization},
  OPTkey = 	 {},
  OPTmonth = 	 {},
  OPTnote = 	 {},
  OPTarchivePrefix = {},
  OPTeprint = 	 {},
  OPTprimaryclass = {},
  OPTannote = 	 {},
}

@article {miyatake1952osaka,
  AUTHOR =	 {Miyatake, Osamu},
  TITLE =	 {On the non-existence of solution of field equations in quantum mechanics},
  JOURNAL =	 {J. Inst. Polytech. Osaka City Univ. Ser. A},
  FJOURNAL =	 {Journal of the Institute of Polytechnics. Osaka City University. Series A. Mathematics},
  VOLUME =	 {2},
  YEAR =	 {1952},
  PAGES =	 {89--99},
  ISSN =	 {0388-0516},
  MRCLASS =	 {81.0X},
  MRNUMBER =	 {52990},
  MRREVIEWER =	 {Freeman J. Dyson},
  opta = 	 {O. Miyatake},
  optt = 	 {On the non-existence of solution of field equations in quantum mechanics},
  OPTkey = 	 {},
  OPTFile = 	 {},
  OPTnumber = 	 {},
  OPTmonth = 	 {},
  OPTnote = 	 {},
  OPTarchivePrefix = {},
  OPTeprint = 	 {},
  OPTprimaryclass = {},
  OPTdoi = 	 {},
  OPTannote = 	 {},
}

@article{nelson1964jmp,
  AUTHOR =	 {Nelson, Edward},
  opta =	 {E. Nelson},
  TITLE =	 {Interaction of nonrelativistic particles with a quantized scalar field},
  optt =	 {Interaction of nonrelativistic particles with a quantized scalar field},
  JOURNAL =	 {J. Math. Phys.},
  FJOURNAL =	 {Journal of Mathematical Physics},
  VOLUME =	 {5},
  YEAR =	 {1964},
  PAGES =	 {1190--1197},
  ISSN =	 {0022-2488},
  MRCLASS =	 {81.47},
  MRNUMBER =	 {0175537},
  MRREVIEWER =	 {O. W. Greenberg},
  URL =		 {https://doi.org/10.1063/1.1704225},
  OPTkey = 	 {},
  OPTnumber = 	 {},
  OPTmonth = 	 {},
  OPTnote = 	 {},
  OPTannote = 	 {},
}

@article{osterwalder1973cmp,
  Author =	 {Osterwalder, Konrad and Schrader, Robert},
  opta =	 {K. Osterwalder, R. Schrader},
  Fjournal =	 {Communications in Mathematical Physics},
  Issn =	 {0010-3616},
  Journal =	 {Comm. Math. Phys.},
  Mrclass =	 {81.46},
  Mrnumber =	 {0329492},
  Mrreviewer =	 {W. Guttinger},
  Pages =	 {83--112},
  Title =	 {Axioms for {E}uclidean {G}reen's functions},
  optt =	 {Axioms for Euclidean Green's functions},
  Url =		 {http://projecteuclid.org/euclid.cmp/1103858969},
  Volume =	 {31},
  Year =	 {1973},
  OPTkey = 	 {},
  OPTnumber = 	 {},
  OPTmonth = 	 {},
  OPTnote = 	 {},
  OPTannote = 	 {}
}

@article{osterwalder1975cmp,
  Author =	 {Osterwalder, Konrad and Schrader, Robert},
  opta =	 {K. Osterwalder, R. Schrader},
  Fjournal =	 {Communications in Mathematical Physics},
  Issn =	 {0010-3616},
  Journal =	 {Comm. Math. Phys.},
  Mrclass =	 {81.46},
  Mrnumber =	 {0376002},
  Mrreviewer =	 {W. Guttinger},
  Note =	 {With an appendix by Stephen Summers},
  Pages =	 {281--305},
  Title =	 {Axioms for {E}uclidean {G}reen's functions. {II}},
  optt =	 {Axioms for Euclidean Green's functions. II},
  Url =		 {http://projecteuclid.org/euclid.cmp/1103899050},
  Volume =	 {42},
  Year =	 {1975},
  OPTkey = 	 {},
  OPTnumber = 	 {},
  OPTmonth = 	 {},
  OPTannote = 	 {}
}

@article {posilicano2020mpag,
  AUTHOR =	 {Posilicano, Andrea},
  TITLE =	 {On the self-adjointness of {$H+A^*+A$}},
  JOURNAL =	 {Math. Phys. Anal. Geom.},
  FJOURNAL =	 {Mathematical Physics, Analysis and Geometry. An International Journal Devoted to the Theory and Applications of
                  Analysis and Geometry to Physics},
  VOLUME =	 {23},
  YEAR =	 {2020},
  NUMBER =	 {4},
  PAGES =	 {Paper No. 37, 31},
  ISSN =	 {1385-0172},
  MRCLASS =	 {47B25 (47A10 47A55 81Q10 81T16)},
  MRNUMBER =	 {4163511},
  MRREVIEWER =	 {Hatem Najar},
  DOI =		 {10.1007/s11040-020-09359-x},
  URL =		 {https://doi.org/10.1007/s11040-020-09359-x},
  opta = 	 {A. Posilicano},
  optt = 	 {On the self-adjointness of H+A^*+A},
  OPTkey = 	 {},
  OPTFile = 	 {},
  OPTmonth = 	 {},
  OPTnote = 	 {},
  OPTarchivePrefix = {},
  OPTeprint = 	 {},
  OPTprimaryclass = {},
  OPTannote = 	 {},
}

@book{reed1972I,
  Address =	 {New York},
  Author =	 {Reed, Michael and Simon, Barry},
  opta =	 {M. Reed, B. Simon},
  Mrclass =	 {47-02 (81.47)},
  Mrnumber =	 {0493419 (58 \#12429a)},
  Mrreviewer =	 {P. R. Chernoff},
  Pages =	 {xvii+325},
  Publisher =	 {Academic Press},
  Title =	 {Methods of modern mathematical physics. {I}. {F}unctional analysis},
  optt =	 {Methods of modern mathematical physics. I. Functional analysis},
  Year =	 {1972},
  ALTeditor = 	 {},
  OPTkey = 	 {},
  OPTvolume = 	 {},
  OPTnumber = 	 {},
  OPTseries = 	 {},
  OPTedition = 	 {},
  OPTmonth = 	 {},
  OPTnote = 	 {},
  OPTannote = 	 {}
}

@book{reed1979III,
  Address =	 {New York},
  Author =	 {Reed, Michael and Simon, Barry},
  opta =	 {M. Reed, B. Simon},
  Isbn =	 {0-12-585003-4},
  Mrclass =	 {81F99},
  Mrnumber =	 {529429 (80m:81085)},
  Mrreviewer =	 {P. R. Chernoff},
  Note =	 {Scattering theory},
  Pages =	 {xv+463},
  Publisher =	 {Academic Press},
  Title =	 {Methods of modern mathematical physics. {III}. {S}cattering {T}heory},
  optt =	 {Methods of modern mathematical physics. III. Scattering Theory},
  Year =	 {1979},
  ALTeditor = 	 {},
  OPTkey = 	 {},
  OPTvolume = 	 {},
  OPTnumber = 	 {},
  OPTseries = 	 {},
  OPTedition = 	 {},
  OPTmonth = 	 {},
  OPTannote = 	 {}
}

@article{segal1956am+tams,
  Author =	 {Segal, I. E.},
  opta =	 {I. E. Segal},
  Journal =	 {Trans. Amer. Math. Soc. 81 {\&} Ann. Math. (2) 63},
  OPTPages =	 {},
  Title =	 {Tensor algebras over {H}ilbert spaces. {I \& II}},
  optt =	 {Tensor algebras over Hilbert spaces. I&II},
  OPTVolume =	 {},
  Year =	 {1956},
  OPTFile = 	 {},
  OPTkey = 	 {},
  OPTnumber = 	 {},
  OPTmonth = 	 {},
  OPTnote = 	 {},
  OPTannote = 	 {}
}

@article{segal1959mfmdvs,
  Author =	 {Segal, I. E.},
  opta =	 {I. E. Segal},
  Fjournal =	 {Det Kongelige Danske Videnskabernes Selskab Matematisk-Fysiske Meddelelser},
  Journal =	 {Mat.-Fys. Medd. Danske Vid. Selsk.},
  Mrclass =	 {81.00 (46.00)},
  Mrnumber =	 {0112626 (22 \#3477)},
  Mrreviewer =	 {G. W. Mackey},
  Number =	 {12},
  Pages =	 {39 pp. (1959)},
  Title =	 {Foundations of the theory of dynamical systems of infinitely many degrees of freedom. {I}},
  optt =	 {Foundations of the theory of dynamical systems of infinitely many degrees of freedom. I},
  Volume =	 {31},
  Year =	 {1959},
  OPTkey = 	 {},
  OPTmonth = 	 {},
  OPTnote = 	 {},
  OPTannote = 	 {}
}

@article{segal1960jmp,
  Author =	 {Segal, I. E.},
  opta =	 {I. E. Segal},
  Fjournal =	 {Journal of Mathematical Physics},
  Issn =	 {0022-2488},
  Journal =	 {J. Math. Phys.},
  Mrclass =	 {81.46 (57.50)},
  Mrnumber =	 {0135093 (24 \#B1144)},
  Mrreviewer =	 {A. S. Wightman},
  Pages =	 {468--488},
  Title =	 {Quantization of nonlinear systems},
  optt =	 {Quantization of nonlinear systems},
  Volume =	 {1},
  Year =	 {1960},
  OPTkey = 	 {},
  OPTnumber = 	 {},
  OPTmonth = 	 {},
  OPTnote = 	 {},
  OPTannote = 	 {}
}

@article{segal1961cjm,
  Author =	 {Segal, I. E.},
  opta =	 {I. E. Segal},
  Fjournal =	 {Canadian Journal of Mathematics. Journal Canadien de Math\'ematiques},
  Issn =	 {0008-414X},
  Journal =	 {Canad. J. Math.},
  Mrclass =	 {81.00},
  Mrnumber =	 {0128839 (23 \#B1877)},
  Mrreviewer =	 {G. W. Mackey},
  Pages =	 {1--18},
  Title =	 {Foundations of the theory of dyamical systems of infinitely many degrees of freedom. {II}},
  optt =	 {Foundations of the theory of dyamical systems of infinitely many degrees of freedom. II},
  Volume =	 {13},
  Year =	 {1961},
  OPTkey =	 {},
  OPTnumber =	 {},
  OPTmonth =	 {},
  OPTnote =	 {},
  OPTannote =	 {}
}

@Article{sellke2022duke,
  author = 	 {Sellke, Mark},
  opta = 	 {M. Sellke},
  title = 	 {{Almost Quartic Lower Bound for the Fr\"{o}hlich Polaron's Effective Mass via Gaussian Domination}},
  optt = 	 {Almost Quartic Lower Bound for the Fröhlich Polaron's Effective Mass via Gaussian Domination},
  journal = 	 {Duke Math. J.},
  year = 	 {2024},
  OPTkey = 	 {},
  OPTFile = 	 {},
  OPTvolume = 	 {},
  OPTnumber = 	 {},
  OPTpages = 	 {},
  OPTmonth = 	 {},
  note = 	 {in press},
  archivePrefix = {arXiv},
  eprint = 	 {2212.14023},
  primaryclass = {math-ph},
  OPTdoi = 	 {},
  OPTannote = 	 {}
}

@article{simon1972jfa,
  Author =	 {Simon, Barry and H{\o}egh-Krohn, Raphael},
  opta =	 {B. Simon, R. Høegh-Krohn},
  Journal =	 {J. Funct. Anal.},
  Fjournal =	 {Journal of Functional Analysis},
  Mrclass =	 {47A55 (81.47)},
  Mrnumber =	 {0293451 (45 \#2528)},
  Mrreviewer =	 {S. Gudder},
  Pages =	 {121--180},
  Title =	 {Hypercontractive semigroups and two dimensional self-coupled {B}ose fields},
  optt =	 {Hypercontractive semigroups and two dimensional self-coupled Bose fields},
  Volume =	 {9},
  Year =	 {1972},
  OPTkey = 	 {},
  OPTnumber = 	 {},
  OPTmonth = 	 {},
  OPTnote = 	 {},
  OPTannote = 	 {}
}

@book{simon1974psp,
  Address =	 {Princeton, N.J.},
  Author =	 {Simon, Barry},
  opta =	 {B. Simon},
  Mrclass =	 {81.60},
  Mrnumber =	 {0489552 (58 \#8968)},
  Mrreviewer =	 {Raymond Streater},
  Note =	 {Princeton Series in Physics},
  Pages =	 {xx+392},
  Publisher =	 {Princeton University Press},
  Title =	 {The {$P(\phi )_{2}$} {E}uclidean (quantum) field theory},
  optt =	 {The P(φ)₂ Euclidean (quantum) field theory},
  Year =	 {1974},
  ALTeditor = 	 {},
  OPTkey = 	 {},
  OPTvolume = 	 {},
  OPTnumber = 	 {},
  OPTseries = 	 {},
  OPTedition = 	 {},
  OPTmonth = 	 {},
  OPTannote = 	 {}
}

@book{streater2000plp,
  Author =	 {Streater, R. F. and Wightman, A. S.},
  opta =	 {R. F. Streater, A. S. Wightman},
  Isbn =	 {0-691-07062-8},
  Mrclass =	 {81T08 (01A75 81-02)},
  Mrnumber =	 {1884336 (2003f:81154)},
  Note =	 {Corrected third printing of the 1978 edition},
  Pages =	 {x+207},
  Publisher =	 {Princeton University Press, Princeton, NJ},
  Series =	 {Princeton Landmarks in Physics},
  Title =	 {P{CT}, spin and statistics, and all that},
  optt =	 {PCT, spin and statistics, and all that},
  Year =	 {2000},
  ALTeditor = 	 {},
  OPTkey = 	 {},
  OPTvolume = 	 {},
  OPTnumber = 	 {},
  OPTaddress = 	 {},
  OPTedition = 	 {},
  OPTmonth = 	 {},
  OPTannote = 	 {}
}

@article{trubatch1968pr,
  title =	 {{Resolution of the Runaway Problem for the Polaron}},
  author =	 {Trubatch, S. L.},
  journal =	 {Phys. Rev.},
  volume =	 {174},
  issue =	 {5},
  pages =	 {1555--1558},
  OPTnumpages =	 {0},
  year =	 {1968},
  month =	 {Oct},
  publisher =	 {American Physical Society},
  doi =		 {10.1103/PhysRev.174.1555},
  url =		 {https://link.aps.org/doi/10.1103/PhysRev.174.1555},
  opta = 	 {S. L. Trubatch},
  optt = 	 {Resolution of the Runaway Problem for the Polaron},
  OPTkey = 	 {},
  OPTnumber = 	 {},
  OPTnote = 	 {},
  OPTarchivePrefix = {},
  OPTeprint = 	 {},
  OPTprimaryclass = {},
  OPTannote = 	 {}
}

@article {vanhove1952physica,
  AUTHOR =	 {Van Hove, L\'{e}on},
  TITLE =	 {Les difficult\'{e}s de divergences pour un modelle particulier de champ quantifi\'{e}},
  JOURNAL =	 {Physica},
  FJOURNAL =	 {Physica},
  VOLUME =	 {18},
  YEAR =	 {1952},
  PAGES =	 {145--159},
  ISSN =	 {0031-8914},
  MRCLASS =	 {81.0X},
  MRNUMBER =	 {49093},
  MRREVIEWER =	 {A. J. Coleman},
  opta = 	 {L. Van Hove},
  optt = 	 {Les difficultés de divergences pour un modelle particulier de champ quantifié},
  OPTkey = 	 {},
  OPTFile = 	 {},
  OPTnumber = 	 {},
  OPTmonth = 	 {},
  OPTnote = 	 {},
  OPTarchivePrefix = {},
  OPTeprint = 	 {},
  OPTprimaryclass = {},
  OPTdoi = 	 {},
  OPTannote = 	 {},
}

@book {weinberg1996I,
  AUTHOR =	 {Weinberg, Steven},
  TITLE =	 {The quantum theory of fields. {V}ol. {I}},
  NOTE =	 {Foundations, Corrected reprint of the 1995 original},
  PUBLISHER =	 {Cambridge University Press, Cambridge},
  YEAR =	 {1996},
  PAGES =	 {xxvi+609},
  ISBN =	 {0-521-55001-7},
  MRCLASS =	 {81Txx (81-01 81-02)},
  MRNUMBER =	 {1410064},
  MRREVIEWER =	 {R. Delbourgo},
  DOI =		 {10.1017/CBO9781139644174},
  opta =	 {S. Weinberg},
  ALTeditor =	 {},
  optt =	 {The quantum theory of fields. Vol. I},
  OPTkey =	 {},
  OPTvolume =	 {},
  OPTnumber =	 {},
  OPTseries =	 {},
  OPTaddress =	 {},
  OPTedition =	 {},
  OPTmonth =	 {},
  OPTarchivePrefix ={},
  OPTeprint =	 {},
  OPTprimaryclass ={},
  OPTannote =	 {},
}

@book {weinberg1996II,
  AUTHOR =	 {Weinberg, Steven},
  TITLE =	 {The quantum theory of fields. {V}ol. {II}},
  NOTE =	 {Modern applications},
  PUBLISHER =	 {Cambridge University Press, Cambridge},
  YEAR =	 {1996},
  PAGES =	 {xxii+489},
  ISBN =	 {0-521-55002-5},
  MRCLASS =	 {81Txx (81-01 81-02)},
  MRNUMBER =	 {1411911},
  MRREVIEWER =	 {R. Delbourgo},
  DOI =		 {10.1017/CBO9781139644174},
  opta = 	 {S. Weinberg},
  ALTeditor = 	 {},
  optt = 	 {The quantum theory of fields. Vol. II},
  OPTkey = 	 {},
  OPTvolume = 	 {},
  OPTnumber = 	 {},
  OPTseries = 	 {},
  OPTaddress = 	 {},
  OPTedition = 	 {},
  OPTmonth = 	 {},
  OPTarchivePrefix = {},
  OPTeprint = 	 {},
  OPTprimaryclass = {},
  OPTannote = 	 {},
}

@book {weinberg2000III,
  AUTHOR =	 {Weinberg, Steven},
  TITLE =	 {The quantum theory of fields. {V}ol. {III}},
  NOTE =	 {Supersymmetry},
  PUBLISHER =	 {Cambridge University Press, Cambridge},
  YEAR =	 {2000},
  PAGES =	 {xxii+419},
  ISBN =	 {0-521-66000-9},
  MRCLASS =	 {81T60 (81-02 81R40 81T10 81T18)},
  MRNUMBER =	 {1737296},
  MRREVIEWER =	 {R. Delbourgo},
  DOI =		 {10.1017/CBO9781139644198},
  opta =	 {S. Weinberg},
  ALTeditor =	 {},
  optt =	 {The quantum theory of fields. Vol. III},
  OPTkey =	 {},
  OPTvolume =	 {},
  OPTnumber =	 {},
  OPTseries =	 {},
  OPTaddress =	 {},
  OPTedition =	 {},
  OPTmonth =	 {},
  OPTarchivePrefix ={},
  OPTeprint =	 {},
  OPTprimaryclass ={},
  OPTannote =	 {},
}

@article{wightman1956pr,
  Author =	 {Wightman, A. S.},
  opta =	 {A. S. Wightman},
  Doi =		 {10.1103/PhysRev.101.860},
  Issue =	 {2},
  Journal =	 {Phys. Rev.},
  Fjournal =	 {Physical Review},
  Numpages =	 {0},
  Pages =	 {860--866},
  Publisher =	 {American Physical Society},
  Title =	 {{Quantum Field Theory in Terms of Vacuum Expectation Values}},
  optt =	 {Quantum Field Theory in Terms of Vacuum Expectation Values},
  Url =		 {http://link.aps.org/doi/10.1103/PhysRev.101.860},
  Volume =	 {101},
  Year =	 {1956},
  OPTkey = 	 {},
  OPTnumber = 	 {},
  OPTmonth = 	 {},
  OPTnote = 	 {},
  OPTannote = 	 {}
}

@article{wightman1965afksv,
  Author =	 {Wightman, A.S. and G{{\aa}}rding, L.},
  opta =	 {A. S. Wightman, L. Gårding},
  Journal =	 {Arkiv Fysik, Kungl. Svenska Vetenskapsak},
  Pages =	 {129-189},
  Title =	 {{Fields As Operator-Valued Distributions In Relativstic Quantum Theory}},
  optt =	 {Fields As Operator-Valued Distributions In Relativstic Quantum Theory},
  Volume =	 {28},
  Year =	 {1964},
  OPTkey = 	 {},
  OPTnumber = 	 {},
  OPTmonth = 	 {},
  OPTnote = 	 {},
  OPTannote = 	 {}
}

@article{wilson1975rmp,
  title =	 {The renormalization group: Critical phenomena and the Kondo problem},
  author =	 {Wilson, Kenneth G.},
  journal =	 {Rev. Mod. Phys.},
  volume =	 {47},
  issue =	 {4},
  pages =	 {773--840},
  numpages =	 {0},
  year =	 {1975},
  month =	 {Oct},
  publisher =	 {American Physical Society},
  doi =		 {10.1103/RevModPhys.47.773},
  url =		 {https://link.aps.org/doi/10.1103/RevModPhys.47.773},
  opta =	 {},
  optt =	 {},
  OPTkey =	 {},
  OPTFile =	 {},
  OPTnumber =	 {},
  OPTnote =	 {},
  OPTarchivePrefix ={},
  OPTeprint =	 {},
  OPTprimaryclass ={},
  OPTannote =	 {}
}

@book {weidmann2000Ide,
  AUTHOR =	 {Weidmann, Joachim},
  TITLE =	 {Lineare {O}peratoren in {H}ilbertr\"{a}umen. {T}eil 1},
  SERIES =	 {Mathematische Leitf\"{a}den. [Mathematical Textbooks]},
  NOTE =	 {Grundlagen. [Foundations]},
  PUBLISHER =	 {B. G. Teubner, Stuttgart},
  YEAR =	 {2000},
  PAGES =	 {475},
  ISBN =	 {3-519-02236-2},
  MRCLASS =	 {47-01 (46-01 47B25)},
  MRNUMBER =	 {1887367},
  DOI =		 {10.1007/978-3-322-80094-7},
  URL =		 {https://doi.org/10.1007/978-3-322-80094-7},
  opta = 	 {},
  ALTeditor = 	 {},
  optt = 	 {},
  OPTkey = 	 {},
  OPTFile = 	 {},
  OPTvolume = 	 {},
  OPTnumber = 	 {},
  OPTaddress = 	 {},
  OPTedition = 	 {},
  OPTmonth = 	 {},
  OPTarchivePrefix = {},
  OPTeprint = 	 {},
  OPTprimaryclass = {},
  OPTannote = 	 {},
}

@article {bogli2017ieot,
  AUTHOR =	 {Bögli, Sabine},
  TITLE =	 {Convergence of sequences of linear operators and their spectra},
  JOURNAL =	 {Integral Equations Operator Theory},
  FJOURNAL =	 {Integral Equations and Operator Theory},
  VOLUME =	 {88},
  YEAR =	 {2017},
  NUMBER =	 {4},
  PAGES =	 {559--599},
  ISSN =	 {0378-620X,1420-8989},
  MRCLASS =	 {47A10 (47A55 47A58 47B07)},
  MRNUMBER =	 {3694623},
  MRREVIEWER =	 {Florian\ Horia\ Vasilescu},
  DOI =		 {10.1007/s00020-017-2389-3},
  URL =		 {https://doi.org/10.1007/s00020-017-2389-3},
  opta = 	 {},
  optt = 	 {},
  OPTkey = 	 {},
  OPTFile = 	 {},
  OPTmonth = 	 {},
  OPTnote = 	 {},
  OPTarchivePrefix = {},
  OPTeprint = 	 {},
  OPTprimaryclass = {},
  OPTannote = 	 {},
}

@article {bogli2018jst,
  AUTHOR =	 {Bögli, Sabine},
  TITLE =	 {Local convergence of spectra and pseudospectra},
  JOURNAL =	 {J. Spectr. Theory},
  FJOURNAL =	 {Journal of Spectral Theory},
  VOLUME =	 {8},
  YEAR =	 {2018},
  NUMBER =	 {3},
  PAGES =	 {1051--1098},
  ISSN =	 {1664-039X,1664-0403},
  MRCLASS =	 {47A10},
  MRNUMBER =	 {3831156},
  DOI =		 {10.4171/JST/222},
  URL =		 {https://doi.org/10.4171/JST/222},
  opta = 	 {},
  optt = 	 {},
  OPTkey = 	 {},
  OPTFile = 	 {},
  OPTmonth = 	 {},
  OPTnote = 	 {},
  OPTarchivePrefix = {},
  OPTeprint = 	 {},
  OPTprimaryclass = {},
  OPTannote = 	 {},
}

@article {post2022jst,
  AUTHOR =	 {Post, Olaf and Zimmer, Sebastian},
  TITLE =	 {Generalised norm resolvent convergence: comparison of different concepts},
  JOURNAL =	 {J. Spectr. Theory},
  FJOURNAL =	 {Journal of Spectral Theory},
  VOLUME =	 {12},
  YEAR =	 {2022},
  NUMBER =	 {4},
  PAGES =	 {1459--1506},
  ISSN =	 {1664-039X,1664-0403},
  MRCLASS =	 {47A55 (47A58)},
  MRNUMBER =	 {4590010},
  MRREVIEWER =	 {Olga\ Yuriivna\ Dyuzhenkova},
  DOI =		 {10.4171/jst/442},
  URL =		 {https://doi.org/10.4171/jst/442},
  opta = 	 {},
  optt = 	 {},
  OPTkey = 	 {},
  OPTFile = 	 {},
  OPTmonth = 	 {},
  OPTnote = 	 {},
  OPTarchivePrefix = {},
  OPTeprint = 	 {},
  OPTprimaryclass = {},
  OPTannote = 	 {},
}

@book {bauerschmidt2019lnm,
  AUTHOR =	 {Bauerschmidt, Roland and Brydges, David C. and Slade, Gordon},
  TITLE =	 {Introduction to a renormalisation group method},
  SERIES =	 {Lecture Notes in Mathematics},
  VOLUME =	 {2242},
  PUBLISHER =	 {Springer, Singapore},
  YEAR =	 {2019},
  PAGES =	 {xii+281},
  ISBN =	 {978-981-32-9591-9; 978-981-32-9593-3},
  MRCLASS =	 {82B28 (81T17 82C28)},
  MRNUMBER =	 {3969983},
  DOI =		 {10.1007/978-981-32-9593-3},
  URL =		 {https://doi.org/10.1007/978-981-32-9593-3},
}

@article{benfatto2007cmp,
  author =	 {Benfatto, G. and Falco, P. and Mastropietro, V.},
  title =	 {Functional integral construction of the massive {Thirring} model: verification of axioms and massless limit},
  fjournal =	 {Communications in Mathematical Physics},
  journal =	 {Commun. Math. Phys.},
  issn =	 {0010-3616},
  volume =	 {273},
  number =	 {1},
  pages =	 {67--118},
  year =	 {2007},
  language =	 {English},
  doi =		 {10.1007/s00220-007-0254-y},
  zbMATH =	 {5199710},
  Zbl =		 {1124.81031},
  archivePrefix ={arXiv},
  eprint =	 {0606177},
  primaryclass = {hep-th},
  OPTannote =	 {}
}

@article {brydges1983cmp,
    AUTHOR = {Brydges, David C. and Fr\"ohlich, J\"urg and Sokal, Alan D.},
     TITLE = {A new proof of the existence and nontriviality of the
              continuum {$\varphi \sp{4}\sb{2}$}\ and {$\varphi
              \sp{4}\sb{3}$}\ quantum field theories},
   JOURNAL = {Comm. Math. Phys.},
  FJOURNAL = {Communications in Mathematical Physics},
    VOLUME = {91},
      YEAR = {1983},
    NUMBER = {2},
     PAGES = {141--186},
      ISSN = {0010-3616,1432-0916},
   MRCLASS = {81E08},
  MRNUMBER = {723546},
MRREVIEWER = {Gerhard\ C.\ Hegerfeldt},
       URL = {http://projecteuclid.org/euclid.cmp/1103940528},
}

@article {gawedzki1985cmp,
    AUTHOR = {Gaw\k{e}dzki, K. and Kupiainen, A.},
     TITLE = {Massless lattice {$\varphi^4_4$} theory: rigorous control of a
              renormalizable asymptotically free model},
   JOURNAL = {Comm. Math. Phys.},
  FJOURNAL = {Communications in Mathematical Physics},
    VOLUME = {99},
      YEAR = {1985},
    NUMBER = {2},
     PAGES = {197--252},
      ISSN = {0010-3616,1432-0916},
   MRCLASS = {81E25 (81E15)},
  MRNUMBER = {790736},
MRREVIEWER = {Giovanni\ Felder},
       URL = {http://projecteuclid.org/euclid.cmp/1103942678},
}

@incollection {brydges1994arg,
    AUTHOR = {Brydges, D. and Dimock, J. and Hurd, T. R.},
     TITLE = {Applications of the renormalization group},
 BOOKTITLE = {Mathematical quantum theory. {I}. {F}ield theory and many-body
              theory ({V}ancouver, {BC}, 1993)},
    SERIES = {CRM Proc. Lecture Notes},
    VOLUME = {7},
     PAGES = {171--189},
 PUBLISHER = {Amer. Math. Soc., Providence, RI},
      YEAR = {1994},
      ISBN = {0-8218-0365-4},
   MRCLASS = {81T17 (81T08 82B27 82B28)},
  MRNUMBER = {1311037},
MRREVIEWER = {Domingos\ H. U. Marchetti},
       DOI = {10.1090/crmp/007/07},
       URL = {https://doi.org/10.1090/crmp/007/07},
}

@article {chandra2024im,
    AUTHOR = {Chandra, Ajay and Chevyrev, Ilya and Hairer, Martin and Shen,
              Hao},
     TITLE = {Stochastic quantisation of {Y}ang-{M}ills-{H}iggs in 3{D}},
   JOURNAL = {Invent. Math.},
  FJOURNAL = {Inventiones Mathematicae},
    VOLUME = {237},
      YEAR = {2024},
    NUMBER = {2},
     PAGES = {541--696},
      ISSN = {0020-9910,1432-1297},
   MRCLASS = {60H15 (60L30 81S20 81T13)},
  MRNUMBER = {4768631},
MRREVIEWER = {Sergey\ V.\ Lototsky},
       DOI = {10.1007/s00222-024-01264-2},
       URL = {https://doi.org/10.1007/s00222-024-01264-2},
}

@article{goulko2025prl,
  title = {{Transient Dynamical Phase Diagram of the Spin-Boson Model}},
  author = {Goulko, Olga and Chen, Hsing-Ta and Goldstein, Moshe and Cohen, Guy},
  journal = {Phys. Rev. Lett.},
  volume = {134},
  issue = {5},
  pages = {056502},
  numpages = {6},
  year = {2025},
  month = {Feb},
  publisher = {American Physical Society},
  doi = {10.1103/PhysRevLett.134.056502},
  url = {https://link.aps.org/doi/10.1103/PhysRevLett.134.056502}
}

@article{magazzu2018nc,
  title = {Probing the strongly driven spin-boson model in a superconducting quantum circuit},
  author = {Magazz\`{u}, L. and Forn-Diaz, P. and Belyansky, R. and Orgiazzi, J.-L. and Yurtalan, M. A. and Otto, M. R. and Lupascu, A. and Wilson, C. M. and Grifoni, M.},
  journal = {Nat. Commun.},
  volume = {9},
  pages = {1403},
  year = {2018},
  doi = {10.1038/s41467-018-03626-w},
}

@article{guo2025pla,
  title =	 {Re-examining steady-state dynamics and heat transfer in a strongly coupled spin-boson model},
  journal =	 {Phys. Lett. A},
  volume =	 {560},
  pages =	 {130939},
  year =	 {2025},
  doi =		 {10.1016/j.physleta.2025.130939},
  author =	 {Guo, Bao-qing and Peng, Jia-long and Zhong, Cheng and Liu, Yu-qiang},
}

@article{sun2025nc,
  title = {Quantum simulation of spin-boson models with structured bath},
  author = {Ke Sun and Mingyu Kang and Hanggai Nuomin and George Schwartz and David N. Beratan and Kenneth R. Brown and Jungsang Kim},
  journal = {Nat. Commun.},
  volume = {16},
  pages = {4042},
  year = {2025},
  doi = {10.1038/s41467-025-59296-y},
}
 }

\end{document}